\documentclass
[
    fontsize = 11 pt,        % The size of the font.
    american,                % Support for American English.
    captions = tableheading, % Places the correct amount of space when the caption of a table is above the table.
    numbers = noenddot,      % Does not use a period at the end of numbered titles, such as sections or figures.
    abstracton,
    paper = letter,
]
{scrartcl}

\usepackage[T1]{fontenc}
\usepackage[utf8]{inputenc}
\usepackage
[
    babel = true, % Enables language-specific tuning.
]
{microtype}           % Uses the text space more efficiently.
\usepackage{csquotes} % Uses the correct quotes according to the current language.

\usepackage[dvipsnames]{xcolor} % Allows it to define colors. The option says that common names can be used.

\definecolor{stroke1}{HTML}{2574A9} % This color is used as the standard color to highlight things.

\usepackage{amsmath}
\usepackage{amssymb}
\usepackage{amsthm}
\usepackage{thmtools}
\usepackage{mathtools}
\usepackage{thm-restate}
\usepackage{dsfont}        % Yields far better blackboard-bold letters than \mathbb. Use \mathds in order to write such letters.
\usepackage{braceMnSymbol} % Adjusts overbraces and underbraces such that longer versions are put together seamlessly.

\usepackage
[
    ttscale = 0.85, % Scales the typewriter font.
]
{libertine} % The main font used in this thesis.
\usepackage
[
    libertine,    % Changes the math font to libertine (the main font).
    slantedGreek, % Makes all greek letters italic by default. If you want to use an upright greek letter, use ›\up‹ immediately followed by the letter’s name. For example, \upGamma displays an upright uppercase gamma.
    vvarbb,       % Changes the \mathbb font to another font. However, \mathbb remains ugly and should not be used. Use \mathds instead.
    libaltvw,     % Uses different characters for v und w that look far better than the default ones.
]
{newtxmath} % The main math font of this thesis. It fits well with the main font.
\usepackage{url} % Responsible for URL formatting.
\usepackage{bm}  % Allows to use sensible bold letters in math mode. This package has to go after the font packages. Otherwise it does not work correctly!

\usepackage{graphicx} % The standard package for including graphics into your document.
\usepackage
[
    subrefformat = simple, % Formats the label of the \subref command without parentheses.
    labelformat = simple,  % Formats the mark of a subfigure without parentheses.
]
{subcaption}         % Enables it to have subfigures inside of a single figure.

\usepackage{array}     % Improves the way that tables can be formatted.
\usepackage{booktabs}  % Adds lines (called ›rules‹) that can be used in tables and improves spacing.
\usepackage{longtable} % Allows to make tables that span multiple pages.
\usepackage{pdflscape} % Allows to change a page into landscape. This is handy if a table is very wide.

\usepackage[shortlabels]{enumitem} % Adds tons of useful features to enumeration environments.

\usepackage
[
    ruled,         % Creates lines at the top and at the bottom. Further, the caption is now above the algorithm.
    vlined,        % Shows the scope of a statement spanning multiple lines via a small vertical bar. Thus, no closing tags are needed.
    linesnumbered, % Shows line numbers.
]
{algorithm2e} % Allows to write pseudocode.

\usepackage{xspace}   % Adds the functionality that a space after a command will be shown as a space in the output.
\usepackage
[
    shortcuts, % Allows to use short symbols for non-breaking hyphens and dashes instead of lengthy commands.
]
{extdash}             % Adds non-breaking hyphens and dashes.
\usepackage{setspace} % Allows to easily chnage the spacing inside of the document.

\usepackage{xparse}    % Is used in order to define reasonable commands.
\usepackage{footnote}  % Allows it to extend the environments footnotes can be used in. It is said that this package is in conflict with ›hyperref‹. I did not note any troubles. However, if something is fishy, it is probably best to not use this package.
\usepackage{afterpage} % Adds the \afterpage command, which specifies that the provided argument shall be processed after the current page is finished.
\usepackage
[
    textsize = scriptsize, % Determines the text size of the TODO note.
]
{todonotes}            % Adds TODO notes to the document. These are small text areas inside of the margin of a page.

\usepackage
[
    sortcites,              % Sort multiple references when citing them together.
    style = alphabetic,     % The style of a citation mark.
    defernumbers,           % Makes sure that references always have unique numbers. This is important if you use multiple bibliographies.
    safeinputenc,           % Allows to use UTF8 characters in the bibliography and tries to translate them into TeX automatically.
    backref = true,         % Creates back references in the bibliography.
    backrefstyle = three,   % Compresses three or more consecutive pages in the back references into a range.
    hyperref = true,        % Makes links generated by biblatex clickable. If hyperref is not used, a warning is issued.
    maxbibnames = 99,       % The maximum number of names displayed in the bibliography.
    mincitenames = 99,       % The minimum number of names displayed when using commands like ›textcite‹. The default is 3. After that, ›et al.‹ is used.
    maxcitenames = 99,       % The maximum number of names displayed when using commands like ›textcite‹. The default is 3. After that, ›et al.‹ is used.
    minalphanames = 3,       % The minimum number of names to display for the alphabetic sorting style.
]
{biblatex} % Used in order to format the bibliography.

\DefineBibliographyStrings{english}%
{%
    backrefpage  = {\lowercase{s}ee page}, % For a single page number.
    backrefpages = {\lowercase{s}ee pages} % For multiple page numbers.
}

\usepackage
[
    letterpaper,
    margin = 1 in,
    includefoot,
]
{geometry}
\usepackage{mfirstuc}
\usepackage{algorithmicx}
\usepackage{amsfonts}
\usepackage
[
    labelfont = bf,
    format = plain,
]
{caption}

\usepackage[american]{babel}
\usepackage{mathtools}

\usepackage
[
    bookmarks = true,                 % Generates boodmarks for the PDF.
    bookmarksopen = false,            % The bookmarks are closed by default.
    bookmarksnumbered = true,         % The bookmarks use the numbers of the corresponding headline.
    pdfstartpage = 1,                 % The first page seen when opening the PDF.
    colorlinks = true,                % The text of hyperlinks is colored instead of having a colored box around it.
    allcolors = stroke1,              % Every hyperlink uses the same color. If you want to change specific colors, use the commands below.
]
{hyperref} % The standard package that is used for creating hyperlinks inside of a document.

\usepackage
[
    noabbrev,   % The words in front of the labels are not abbreviated.
    nameinlink, % Extends the link of a reference to the word in front of it.
]
{cleveref} % This package must be included after ›hyperref‹. It creates clever references that know what they refer to.

\date{}

\makeatletter
    \def\IfEmptyTF#1%
    {%
        \if\relax\detokenize{#1}\relax%
            \expandafter\@firstoftwo%
        \else%
            \expandafter\@secondoftwo%
        \fi%
    }
\makeatother

\NewDocumentCommand{\mathOrText}{m}
{%
    \ensuremath{#1}\xspace%
}

\let\originalleft\left
\let\originalright\right
\renewcommand{\left}{\mathopen{}\mathclose\bgroup\originalleft}
\renewcommand{\right}{\aftergroup\egroup\originalright}

\makeatletter
    \DeclareRobustCommand{\bfseries}%
    {%
        \not@math@alphabet\bfseries\mathbf%
        \fontseries\bfdefault\selectfont%
        \boldmath%
    }

\xspaceaddexceptions{]\}}

\allowdisplaybreaks

\crefname{ineq}{inequality}{inequalities}
\creflabelformat{ineq}{#2{\upshape(#1)}#3}

\crefname{term}{term}{terms}
\creflabelformat{term}{#2{\upshape(#1)}#3}

\crefname{cond}{condition}{conditions}
\creflabelformat{cond}{#2{\upshape(#1)}#3}

\crefname{assume}{assumption}{assumptions}
\creflabelformat{assume}{#2{\upshape(#1)}#3}

\deffootnote[1.2 em]{1.2 em}{0 em}{\makebox[1.2 em][l]{\textbf{\thefootnotemark}}}

\let\oldfootnote\footnote

\newlength{\spaceBeforeFootnote} % Denotes the space before the footnote mark in em.
\newlength{\spaceAfterFootnote}  % Denotes the space after the footnote mark in em.

\RenewDocumentCommand{\footnote}{o o o m}%
{%
    \IfNoValueTF{#1}%
    {%
        \oldfootnote{#4}%
    }%
    {%
        \setlength{\spaceBeforeFootnote}{\IfEmptyTF{#1}{0}{#1} em}%
        \IfNoValueTF{#2}%
        {%
            \hspace*{\spaceBeforeFootnote}\oldfootnote{#4}%
        }%
        {%
            \setlength{\spaceAfterFootnote}{\IfEmptyTF{#2}{0}{#2} em}%
            \hspace*{\spaceBeforeFootnote}\IfNoValueTF{#3}{\oldfootnote{#4}}{\oldfootnote[#3]{#4}}\hspace*{\spaceAfterFootnote}%
        }%
    }%
}

\makesavenoteenv{figure}
\makesavenoteenv{table}
\makesavenoteenv{tabular}

\declaretheoremstyle
[
   	spaceabove = \topsep,
   	spacebelow = \topsep,
   	headfont = \bfseries,
   	headformat = \textcolor{stroke1}{$\blacktriangleright$} \NAME~\NUMBER \NOTE,
   	notefont = \bfseries,
   	notebraces = {(}{)},
   	bodyfont = \normalfont,
   	postheadspace = 0.5 em,
   	qed = \textcolor{stroke1}{\bfseries$\blacktriangleleft$},
]
{myTheoremStyle}

\declaretheorem
[
    style = myTheoremStyle,
    name = Proposition,
    sharenumber = conjecture,
]
{proposition}

\declaretheorem
[
   	style = myTheoremStyle,
   	name = Lemma,
    sharenumber = conjecture,
]
{lemma}
\declaretheorem
[
   	style = myTheoremStyle,
   	name = Corollary,
    sharenumber = conjecture,
]
{corollary}
\declaretheorem
[
   	style = myTheoremStyle,
   	name = Theorem,
    sharenumber = conjecture,
]
{theorem}

\declaretheorem
[
    style = myTheoremStyle,
    name = Remark,
    sharenumber = conjecture,
]
{remark}
\declaretheorem
[
    style = myTheoremStyle,
    name = Observation,
    sharenumber = conjecture,
]
{observation}

\NewDocumentCommand{\functionTemplate}{m m m m o}%
{%
    \IfNoValueTF{#5}%
    {%
        \mathOrText{#1\left#2{#4}\right#3}%
    }%
    {%
        \mathOrText{#1#5#2{#4}#5#3}%
    }%
}

\newcommand*{\leftBracketType}{(}
\newcommand*{\rightBracketType}{)}

\NewDocumentCommand{\createFunction}{m m o o}%
{%
    \renewcommand*{\leftBracketType}{\IfNoValueTF{#3}{(}{#3}}%
    \renewcommand*{\rightBracketType}{\IfNoValueTF{#4}{)}{#4}}%
    \NewDocumentCommand{#1}{o o}%
    {%
        \IfNoValueTF{##1}%
        {%
            \mathOrText{#2}%
        }%
        {%
            \functionTemplate{#2}{\leftBracketType}{\rightBracketType}{##1}[##2]%
        }%
    }%
}

\DeclareDocumentCommand{\probabilisticFunctionTemplate}{m m O{} o}
{%
    \functionTemplate{#1}%
    {\lbrack}%
    {\rbrack}%
    {#2\IfEmptyTF{#3}{}{\ \IfNoValueTF{#4}{\left}{#4}\vert\ \vphantom{#2}#3\IfNoValueTF{#4}{\right.}{}}}%
    [#4]%
}

\newcommand*{\N}{\mathOrText{\mathds{N}}}
\newcommand*{\Z}{\mathOrText{\mathds{Z}}}

\newcommand*{\R}{\mathOrText{\mathds{R}}}

\newcommand*{\indicatorFunctionSymbol}{\mathds{1}}

\RenewDocumentCommand{\Pr}{m O{} o}%
{%
    \probabilisticFunctionTemplate{\mathrm{Pr}}{#1}[#2][#3]%
}

\NewDocumentCommand{\E}{m O{} o}%
{%
    \probabilisticFunctionTemplate{\mathrm{E}}{#1}[#2][#3]%
}

\NewDocumentCommand{\Var}{m O{} o}%
{%
    \probabilisticFunctionTemplate{\mathrm{Var}}{#1}[#2][#3]%
}

\DeclareDocumentCommand{\bigO}{m o}%
{%
    \functionTemplate{\mathrm{O}}{(}{)}{#1}[#2]%
}

\DeclareDocumentCommand{\smallO}{m o}%
{%
    \functionTemplate{\mathrm{o}}{(}{)}{#1}[#2]%
}

\DeclareDocumentCommand{\bigTheta}{m o}%
{%
    \functionTemplate{\upTheta}{(}{)}{#1}[#2]%
}

\DeclareDocumentCommand{\bigOmega}{m o}%
{%
    \functionTemplate{\upOmega}{(}{)}{#1}[#2]%
}

\DeclareDocumentCommand{\smallOmega}{m o}%
{%
    \functionTemplate{\upomega}{(}{)}{#1}[#2]%
}

\DeclareDocumentCommand{\eulerE}{o}%
{%
    \mathOrText{\mathrm{e}\IfNoValueTF{#1}{}{^{#1}}}%
}

\DeclareDocumentCommand{\poly}{m o}%
{%
    \functionTemplate{\mathrm{poly}}{(}{)}{#1}[#2]%
}

\createFunction{\id}{\mathrm{id}}

\NewDocumentCommand{\ind}{m o o}%
{%
    \IfNoValueTF{#2}%
    {%
        \mathOrText{\indicatorFunctionSymbol_{#1}}%
    }%
    {%
        \functionTemplate{\indicatorFunctionSymbol_{#1}}{(}{)}{#2}[#3]%
    }%
}

\DeclareDocumentCommand{\dom}{m o}%
{%
    \functionTemplate{\mathrm{dom}}{(}{)}{#1}[#2]%
}

\DeclareDocumentCommand{\rng}{m o}%
{%
    \functionTemplate{\mathrm{rng}}{(}{)}{#1}[#2]%
}

\DeclareDocumentCommand{\d}{o}%
{%
    \mathrm{d}\IfNoValueTF{#1}{}{^{#1}}%
}

\DeclareDocumentCommand{\set}{m m o}%
{%
    \mathOrText{\IfNoValueTF{#3}{\left}{#3}\{#1\ \IfNoValueTF{#3}{\left}{#3}\vert\ \vphantom{#1}#2\IfNoValueTF{#3}{\right.}{}\IfNoValueTF{#3}{\right}{#3}\}}%
}

\newcommand\sbullet[1][.6]{\mathbin{\vcenter{\hbox{\scalebox{#1}{$\bullet$}}}}}
\newcommand*{\separator}{$\sbullet$\xspace}

\newcommand*{\concatenation}{\mathOrText{}}

\crefname{observation}{observation}{observations}

\declaretheorem
[
	style = myTheoremStyle,
	name = Assumption,
	sharenumber = conjecture,
]
{assumption}

\newcommand*{\mathbold}[1]{\mathOrText{\bm{#1}}}

\newcommand*{\indicator}[1]{\mathOrText{\indicatorFunctionSymbol{\left\{#1\right\}}}}

\newcommand*{\lOneNorm}[1]{\mathOrText{\left\lVert#1\right\rVert_{1}}}

\newcommand*{\size}[1]{\mathOrText{\left\vert#1\right\vert}}

\newcommand*{\error}{\mathOrText{\varepsilon}}

\newcommand*{\dtv}[2]{\mathOrText{d_{\text{TV}} \left(#1, #2\right)}}

\newcommand*{\absolute}[1]{\mathOrText{\left\vert#1\right\vert}}

\DeclareDocumentCommand{\order}{m o}
{
	\mathOrText{\text{ord}_{#1}\IfNoValueF{#2}{\left(#2\right)}}
}

\newcommand*{\defeq}{\mathOrText{\coloneqq}}

\newcommand*{\graph}{\mathOrText{G}}
\newcommand*{\edges}{\mathOrText{E}}
\newcommand*{\vertices}{\mathOrText{V}}
\newcommand*{\vertex}{\mathOrText{v}}

\DeclareDocumentCommand{\degree}{o}
{
	\mathOrText{\Delta\IfNoValueF{#1}{\left(#1\right)}}
}

\DeclareDocumentCommand{\neighbors}{m o}%
{%
	\mathOrText{N\IfNoValueTF{#2}{}{_{#2}}\left(#1\right)}%
}
\DeclareDocumentCommand{\neighborsClosed}{m o}%
{%
	\mathOrText{N\IfNoValueTF{#2}{}{_{#2}}\left[#1\right]}%
}

\newcommand*{\partitionFunctionSymbol}{\mathOrText{Z}}

\newcommand*{\gibbsSymbol}{\mathOrText{\mu}}

\newcommand*{\weightSymbol}{\mathOrText{w}}

\newcommand*{\composition}{\mathOrText{\circ}}

\newcommand*{\NP}{\mathOrText{\textsf{NP}}}
\newcommand*{\weightSymbolHC}{\mathOrText{\lambda}}

\DeclareDocumentCommand{\weight}{o}
{
	\mathOrText{
		\weightSymbolHC \IfNoValueF{#1}{\left(#1\right)}
	}
}

\newcommand*{\criticalWeight}[1]{\mathOrText{\weight_{\text{c}}\left(#1\right)}}

\newcommand*{\independentSetsSymbol}{\mathOrText{\mathcal{I}}}
\DeclareDocumentCommand{\independentSets}{m o}
{
	\mathOrText{
		\independentSetsSymbol\IfNoValueF{#2}{_{#2}}\left(#1\right)
	}
}

\newcommand*{\independentSet}{\mathOrText{I}}

\newcommand*{\hcPartitionFunction}[2]{\mathOrText{\partitionFunctionSymbol_{\text{HC}}\left(#1, #2\right)}}

\DeclareDocumentCommand{\hcGibbs}{m m o}
{
	\mathOrText{
		\gibbsSymbol_{\text{HC}}^{\left(#1, #2\right)} \IfNoValueF{#3}{\left(#3\right)}
	}
}

\DeclareDocumentCommand{\hcGibbsAppx}{m m o}
{
	\mathOrText{
		\widehat{\gibbsSymbol}_{\text{HC}}^{\left(#1, #2\right)} \IfNoValueF{#3}{\left(#3\right)}
	}
}

\newcommand*{\mhcPartitionFunction}[2]{\mathOrText{\partitionFunctionSymbol_{\text{MHC}}\left(#1, #2\right)}}

\DeclareDocumentCommand{\mhcGibbs}{m m o}
{
	\mathOrText{
		\gibbsSymbol_{\text{MHC}}^{\left(#1, #2\right)} \IfNoValueF{#3}{\left(#3\right)}
	}
}

\DeclareDocumentCommand{\countingFunc}{o}
{
	\mathOrText{
		c\IfNoValueF{#1}{\left(#1\right)}
	}
}

\newcommand*{\scaleUp}{\mathOrText{\alpha}}

\newcommand*{\dimension}{\mathOrText{d}}

\newcommand*{\sidelength}{\mathOrText{\ell}}

\DeclareDocumentCommand{\volume}{o}
{
	\mathOrText{
		\mathds{V} \IfNoValueF{#1}{_{#1}}
	}
}

\DeclareDocumentCommand{\numParticles}{o}
{
	\mathOrText{
		k \IfNoValueF{#1}{_{#1}}
	}
}

\DeclareDocumentCommand{\fugacity}{o}
{
	\mathOrText{
		\lambda \IfNoValueF{#1}{\left(#1\right)}
	}
}

\DeclareDocumentCommand{\radius}{o}
{
	\mathOrText{
		r \IfNoValueF{#1}{\left(#1\right)}
	}
}

\newcommand*{\interactionMatrixSymbol}{\mathOrText{R}}
\DeclareDocumentCommand{\interactionMatrix}{o o}
{
	\mathOrText{
		\interactionMatrixSymbol \IfNoValueF{#1}{\IfNoValueF{#2}{\left(#1, #2\right)}}
	}
}
\DeclareDocumentCommand{\interactionMatrixHS}{o o}
{
	\mathOrText{
		\interactionMatrixSymbol_{\text{HS}} \IfNoValueF{#1}{\IfNoValueF{#2}{\left(#1, #2\right)}}
	}
}
\DeclareDocumentCommand{\interactionMatrixWR}{o o}
{
	\mathOrText{
		\interactionMatrixSymbol_{\text{WR}} \IfNoValueF{#1}{\IfNoValueF{#2}{\left(#1, #2\right)}}
	}
}
\DeclareDocumentCommand{\interactionMatrixZero}{o o}
{
	\mathOrText{
		\interactionMatrixSymbol_{0} \IfNoValueF{#1}{\IfNoValueF{#2}{\left(#1, #2\right)}}
	}
}

\DeclareDocumentCommand{\interactions}{o o}
{
	\mathOrText{
		H \IfNoValueF{#1}{\IfNoValueF{#2}{\left(#1, #2\right)}}
	}
}

\DeclareDocumentCommand{\volumeMatrix}{o o}
{
	\mathOrText{
		B \IfNoValueF{#1}{\IfNoValueF{#2}{\left(#1, #2\right)}}
	}
}

\newcommand*{\numTypes}{\mathOrText{q}}

\newcommand*{\generalPartitionFunction}[3]{\mathOrText{\partitionFunctionSymbol\left(#1, #2, #3 \right)}}

\DeclareDocumentCommand{\generalGibbs}{m m m o}
{
	\mathOrText{
		\gibbsSymbol^{\left(#1, #2, #3\right)} \IfNoValueF{#4}{\left(#4\right)}
	}
}

\DeclareDocumentCommand{\generalWeight}{m m o}
{
	\mathOrText{
		\weightSymbol^{\left(#1, #2\right)} \IfNoValueF{#3}{\left(#3\right)}
	}
}

\newcommand*{\typeSymbol}{\mathOrText{\tau}}
\DeclareDocumentCommand{\typeOf}{o}
{
	\mathOrText{
		\typeSymbol \IfNoValueF{#1}{\left(#1\right)}
	}
}
\DeclareDocumentCommand{\invTypeOf}{o}
{
	\mathOrText{
		\typeSymbol^{-1} \IfNoValueF{#1}{\left(#1\right)}
	}
}

\newcommand*{\combineAssignments}{\mathOrText{\oplus}}

\newcommand*{\validSymbol}{\mathOrText{D}}
\DeclareDocumentCommand{\valid}{m m o}%
{%
	\mathOrText{\validSymbol_{#1}^{\left(#2\right)}
		\IfNoValueF{#3}{\left(#3\right)}}
}

\newcommand*{\dist}[2]{\mathOrText{d \left(#1, #2 \right)}}

\newcommand*{\lebesgueSymbol}{\mathOrText{\nu}}
\DeclareDocumentCommand{\lebesgue}{m o}%
{%
	\mathOrText{\lebesgueSymbol^{#1} \IfNoValueTF{#2}{}{\left(#2\right)}}%
}

\newcommand*{\vol}[1]{\mathOrText{\lebesgueSymbol\left(#1\right)}}

\DeclareDocumentCommand{\ball}{m o}%
{%
	\mathOrText{\mathds{B}\IfNoValueF{#2}{_{#2}}\left(#1\right)}%
}

\DeclareDocumentCommand{\integerSphere}{m o}%
{%
	\mathOrText{b_{#1}\IfNoValueTF{#2}{}{\left(#2\right)}}%
}

\DeclareDocumentCommand{\projection}{m o}%
{%
	\mathOrText{\varrho_{#1}\IfNoValueF{#2}{\left(#2\right)}}%
}

\newcommand*{\numParticlesFrom}[2]{\mathOrText{\numParticles_{#1}\left(#2\right)}}

\newcommand*{\degreeError}{\mathOrText{\gamma}}

\newcommand*{\permutationSymbol}{\mathOrText{\pi}}
\DeclareDocumentCommand{\permutation}{o}%
{%
	\mathOrText{\permutationSymbol\IfNoValueF{#1}{\left(#1\right)}}%
}

\DeclareDocumentCommand{\permutationOf}{m o}%
{%
	\mathOrText{\permutationSymbol_{#1}\IfNoValueF{#2}{\left(#2\right)}}%
}

\DeclareDocumentCommand{\intermediateWeight}{m m m o}{
	\mathOrText{
		\weightSymbol^{\left(#1, #2\right)}_{#3}\IfNoValueF{#4}{\left(#4\right)}
	}
}

\newcommand*{\hsPartitionFunction}[3]{\mathOrText{\partitionFunctionSymbol_{\text{HS}}\left(#1, #2, #3\right)}}
\DeclareDocumentCommand{\hsGibbs}{m m m o}
{
	\mathOrText{
		\gibbsSymbol_{\text{HS}}^{\left(#1, #2, #3\right)} \IfNoValueF{#4}{\left(#4\right)}
	}
}

\newcommand*{\wrPartitionFunction}[3]{\mathOrText{\partitionFunctionSymbol_{\text{WR}}\left(#1, #2, #3\right)}}

\DeclareDocumentCommand{\wrGibbs}{m m m o}
{
	\mathOrText{
		\gibbsSymbol_{\text{WR}}^{\left(#1, #2, #3\right)} \IfNoValueF{#4}{\left(#4\right)}
	}
}

\newcommand*{\hcPPGraph}[1]{\mathOrText{\graph_{#1}}}

\DeclareDocumentCommand{\hcPPVertices}{m o}%
{%
	\mathOrText{V_{#1}\IfNoValueF{#2}{^{\left(#2\right)}}}%
}

\newcommand*{\hcPPEdges}[1]{\mathOrText{\edges_{#1}}}

\DeclareDocumentCommand{\hcPPWeight}{m o o}%
{%
	\mathOrText{\weightSymbolHC_{#1}\IfNoValueF{#2}{\IfNoValueTF{#3}{\left(#2\right)}{#3(#2#3)}}}%
}

\newcommand*{\hcPPVertex}[2]{\mathOrText{v_{#1}^{(#2)}}}

\DeclareDocumentCommand{\hcPPDegree}{m o o}%
{%
	\mathOrText{\degree_{#1}
	\IfNoValueF{#3}{^{\left(#3\right)}}
	\IfNoValueF{#2}{\left(#2\right)}}%
}

\DeclareDocumentCommand{\hcPPValidContinuous}{m m m o}{
	\mathOrText{\validSymbol_{#1}^{\left(#2, #3\right)}
		\IfNoValueF{#4}{\left(#4\right)}}
}

\DeclareDocumentCommand{\mhcPPValidContinuous}{m m m o}{
	\mathOrText{\widehat{\validSymbol}_{#1}^{\left(#2, #3\right)}
		\IfNoValueF{#4}{\left(#4\right)}}
}

\DeclareDocumentCommand{\hcPPWeightContinuous}{m m m o}{
	\mathOrText{
		\weightSymbol^{\left(#1, #2\right)}_{#3}\IfNoValueF{#4}{\left(#4\right)}
	}
}

\DeclareDocumentCommand{\hcPPGibbsContinuous}{m m m m o}{
	\mathOrText{
		\gibbsSymbol^{\left(#1, #2, #3\right)}_{#4}\IfNoValueF{#5}{\left(#5\right)}
	}
}

\DeclareDocumentCommand{\hcPPGibbsContinuousAppx}{m m m m o}{
	\mathOrText{
		\widehat{\gibbsSymbol}^{\left(#1, #2, #3\right)}_{#4}\IfNoValueF{#5}{\left(#5\right)}
	}
}

\DeclareDocumentCommand{\hcPPPartitionFunctionContinuous}{m m m m}{
	\mathOrText{
		\partitionFunctionSymbol_{#4}\left(#1, #2, #3\right)
	}
}

\DeclareDocumentCommand{\allocationPartitionFunction}{m m m m}{
	\mathOrText{
		\partitionFunctionSymbol_{#4}\left(#1, #2, #3\right)
	}
}

\newcommand*{\resolution}{\mathOrText{\rho}}

\newcommand*{\canonicalPoints}[1]{\mathOrText{X_{#1}}}

\newcommand*{\grid}[1]{\mathOrText{\mathds{G}_{#1}}}

\newcommand*{\discretizationError}{\mathOrText{\varepsilon_{\text{D}}}}

\newcommand*{\volumeError}{\mathOrText{\delta}}
\newcommand*{\volumeScaling}{\mathOrText{\volumeError}}
\newcommand*{\distanceError}{\mathOrText{\varepsilon}}
\newcommand*{\numPartitions}{\mathOrText{m}}
\newcommand*{\numPoints}{\mathOrText{n}}

\newcommand*{\allocationSymbol}{\mathOrText{\Phi}}
\DeclareDocumentCommand{\allocation}{o}%
{%
	\mathOrText{\allocationSymbol\IfNoValueF{#1}{\left(#1\right)}}%
}
\DeclareDocumentCommand{\invAllocation}{o}%
{%
	\mathOrText{\allocationSymbol^{-1}\IfNoValueF{#1}{\left(#1\right)}}%
}

\newcommand*{\allocationProb}{\mathOrText{p}}

\DeclareDocumentCommand{\middleFunction}{o}%
{%
	\mathOrText{f\IfNoValueF{#1}{\left(#1\right)}}%
}

\newcommand*{\genericError}{\mathOrText{\varepsilon}}

\newcommand*{\genericFailingProb}{\mathOrText{\delta}}

\newcommand*{\randomVariable}{\mathOrText{X}}
\newcommand*{\modifiedRandomVariable}{\mathOrText{Y}}

\newcommand*{\samplingError}{\mathOrText{\varepsilon_{\text{S}}}}

\newcommand*{\approxError}{\mathOrText{\varepsilon_{\text{A}}}}

\newcommand*{\fugacityGap}{\mathOrText{\alpha}}

\title{Algorithms for hard-constraint point processes via discretization}

 \author{Tobias Friedrich$^{*}$ \and Andreas Göbel$^{*}$ \and Maximilian Katzmann$^{*}$ \and Martin~S. Krejca$^{\dagger}$ \and Marcus Pappik$^{*}$}

 \publishers
 {%

 	$^{*}$ \normalsize Hasso Plattner Institute, University of Potsdam, Potsdam, Germany\\
 	\{tobias.friedrich, maximilian.katzmann, andreas.goebel, marcus.pappik\}@hpi.de \\[4 ex]

 	$^{\dagger}$ Sorbonne University, CNRS, LIP6, Paris, France \\
 	martin.krejca@lip6.fr \\
 }

\begin{document}

\maketitle
\vspace*{-4 ex}
\begin{abstract}
We study the algorithmic applications of a natural discretization for the hard-sphere model and the Widom--Rowlinson model in a region of $\dimension$-dimensional Euclidean space $\volume \subset \R^{\dimension}$.
These continuous models are frequently used in statistical physics to describe mixtures of one or multiple particle types subjected to hard-core interactions.
For each type, particles are distributed according to a Poisson point process with a type specific activity parameter, called \emph{fugacity}.
The \emph{Gibbs distribution} over all possible system states is characterized by the mixture of these point processes conditioned that no two particles being closer than some type-dependent distance threshold.
A key part in better understanding the Gibbs distribution is its normalizing constant, called \emph{partition function}.

We give sufficient conditions that the partition function of a discrete hard-core model on a geometric graph based on a point set $X \subset \volume$ closely approximates those of such continuous models.
Previously, such a result was only shown for the hard-sphere model on cubic regions $\volume = [0, \sidelength)^{\dimension}$ when the size of $X$ is exponential in the volume of the region $\vol{\volume}$, limiting algorithmic applications.
In the same setting, our refined analysis only requires a quadratic number of points, which we argue to be tight.

We use our improved discretization results to approximate the partition functions of the hard-sphere model and the Widom-Rowlinson efficiently in $\vol{\volume}$.
For simplicity, we state our algorithmic results for cubic regions although most of our technical arguments extend to broader classes of regions.
For the hard-sphere model, we obtain the first quasi-polynomial deterministic approximation algorithm for the entire fugacity regime for which, so far, only randomized approximations are known.
Furthermore, we simplify a recently introduced fully polynomial randomized approximation algorithm.
For high dimensions, the fugacity regime we achieve is consistent with the best known bound for uniqueness of the Gibbs measure.
Similarly, we obtain the best known approximation bounds for the Widom--Rowlinson model, both randomized and deterministically.
Moreover, we obtain efficient approximate samplers for the Gibbs distributions of the respective spin systems within the same fugacity regimes.

% For our model we develop a discretization scheme reducing it to an instance of the hard-core model on a graph whose number of vertices is polynomial to the volume of the cube. % The normalising factor of the Gibbs distribution is the partition function.

% Using our discretisation scheme we obtain new algorithmic results, namely sampling algorithms and approximation algorithms for the partition function that are polynomial in the volume of the cube. This leads to randomised algorithmic results for the hard-sphere model that match the known uniqueness region and new quasi-polynomial approximation algorithms.
%{\footnotesize
%\textcolor{red!60!black}{
%\begin{itemize}
% \item The submission should be typeset using 11-point or larger fonts, in a single-column, single-space (between lines) format with ample spacing throughout and 1-inch margins all around, on letter-size (8.5 x 11 inch) paper.
% \item Submissions should not reveal the identity of the authors in any way. In particular, authors’ names, affiliations, and email addresses should not appear at the beginning or in the body of the submission. Authors should not include obvious references that reveal their own identity, and should ensure that any references to their own related work are in the third person
%\end{itemize}
%}}

\end{abstract}

\vspace*{3 ex}
\hspace*{1.3 em}
\begin{minipage}{0.85\textwidth}
\end{minipage}

\newpage

\section{Introduction}

Statistical physics models complex systems of interacting particles as probability distributions.
The main goal is to explain the macroscopic properties of such a \emph{spin system} when it is only described by the microscopic interactions among its particles.
Two classical models in this area are the hard-sphere model, central in the analysis of thermodynamics of liquids and liquid mixtures~\cite{LiquidsBook,boublik1980statistical}, and the Widom--Rowlinson model, which explains the evaporation of liquids~\cite{WR70}.
These models have in common, that each of their states consists of a finite set of points in Euclidean space, distributed according to a Poisson point process.
The points correspond to centers of spherical particles and are constrained by hard-core interactions, i.e., two particles are not allowed to occupy the same space.
The probability distribution characterizing the valid states of a model is called its \emph{Gibbs distribution}.
A key part in better understanding this distribution is its normalizing factor: the \emph{partition function} of the model.
%Thus to describe a spin systems we require its \emph{Gibbs distribution} and the \emph{partition function}, the normalizing factor of this distribution\todo{refine}.

Statistical physics studies spin systems in terms of their \emph{phase transitions}, i.e., the change of the system's macroscopic behavior as one of its parameters crosses a critical value.
% Statistical physics studies spin system in terms of parameter regimes that result in different macroscopic behavior.\todo{This sentence feels a bit too specific and wrong.}
% If there exists a critical value at which the system completely changes its behavior, this is called a \emph{phase transition}.\todo{The definition is of a sharp phase transition, and it might be that some systems admit a coarse one.}
Rigorous proofs of phase transitions for \emph{continuous} spin systems are rare.
Concerning our previous examples, a phase transition has been proved to exist for the Widom--Rowlinson model~\cite{ruelle1971existence}, while it is still an open question for the hard-sphere model.
Phase transitions are more commonly shown for \emph{discrete} spin systems, where, instead of the Euclidean space, particles are placed on the vertices of a graph and interact only with particles on adjacent vertices.

One of the simplest and most well-studied discrete spin systems is the \emph{hard-core model}.
It is considered to be a universal model, as any discrete spin system can be mapped to a hard-core instance~\cite[Section~5.7]{simonbook}.
For an undirected graph $\graph=(V, E)$ and a function $\weight\colon \vertices \to \R_{\ge 0}$, the hard-core model is defined by the tuple~$(\graph, \weight)$.
Its set of valid states is the set of all independent sets $\independentSets{\graph}$ of $\graph$.
Each independent set $\independentSet \in \independentSets{\graph}$ is assigned the weight $\prod_{\vertex \in \independentSet} \weight[\vertex]$.
The partition function~$\hcPartitionFunction{\graph}{\weight}$ of the hard-core model is the sum of weights of all independent sets, and its Gibbs distribution~$\hcGibbs{\graph}{\weight}$ assigns each independent set a probability proportional to its weight, normalized over the partition function.
Formally,
\[
	\hcPartitionFunction{\graph}{\weight} = \sum_{\independentSet \in \independentSets{\graph}} \prod_{\vertex \in \independentSet} \weight[\vertex]
	\text{  and, for all $\independentSet \in \independentSets{\graph}$, we have }
	\hcGibbs{\graph}{\weight}[\independentSet] = \frac{\prod_{\vertex \in \independentSet} \weight[\vertex]}{\hcPartitionFunction{\graph}{\weight}} .
\]
In the uniform case, where~$\fugacity$ is constant, classical results in statistical physics establish the existence of a threshold for which the hard-core model undergoes a phase transition on various graph classes, such as two-dimensional lattices~\cite{LeeYang} and regular trees~\cite{kelly1985,BRIGHTWELL1999221}.
In a series of celebrated results, the tree threshold has been linked to a change in the computational behavior of the hard-core model.
That is, for values of~$\fugacity$ where the Gibbs measure is unique on a $\degree$-regular tree, there is a polynomial-time algorithm for approximating $\hcPartitionFunction{\graph}{\weight}$ on general graphs of maximum degree~$\degree$~\cite{Weitz2006Counting}, while for the remaining values of~$\fugacity$, this is \NP-hard~\cite{2010:Sly:computational_transition,2014:Galanis:inapproximability_independent_hard_core}.

As discrete models are better understood, the hard-core model on lattices is often used as a coarse-grained version of the continuous hard-sphere model.
\citeauthor{friedli_velenik_2017} write for this discretization:
\begin{quote}
    ``Although it might appear as a significant departure from reality, we will see that it leads to satisfactory results and allows a good qualitative understanding of the corresponding phenomena.''~\cite[Chapter~$4.1$]{friedli_velenik_2017}
\end{quote}

In this article, we show that a natural discretization based on the hard-core model leads not only to a qualitative but also to an \emph{algorithmic} understanding of the hard-sphere model and the Widom--Rowlinson model.
%In this article, we show  for a broad class of \emph{continuous} spin systems with hard-core interactions that an appropriately chosen discretization scheme based on the hard-core model leads not only to a qualitative but also to a \emph{quantitative} understanding of the continuous model.
Our main contributions are:
%A summary of our main contributions follows.

% Our method shows that the partition function of such\todo{Why \emph{such}?} a continuous spin system can be closely approximated by the partition function of the hard-core model on an appropriately constructed graph.
% This allows us to use known algorithmic results for hard-core model and obtain new sampling and approximation algorithms for this class of continuous spin systems.
\begin{itemize}
	\item We study a natural discretization of the continuous Widom--Rowlinson model and the hard-sphere model on a region $\volume \subset \R^\dimension$, resulting in an instance of the hard-core model $(\graph, \fugacity)$.
	More precisely, we determine conditions such that the hard-core partition function $\hcPartitionFunction{\graph}{\weight}$ closely approximates the partition function of the continuous point process and the number of vertices of $\graph$ is polynomial in the volume of $\volume$.
	\item We obtain new algorithmic results for these continuous spin systems by combining our discretization result with known algorithms for the hard-core model.
	We present sufficient conditions for randomized approximation algorithms of the partition function and sampling algorithms for the Gibbs distribution with running time polynomial in the volume of $\volume$.
	For the Widom--Rowlinson model, this results in the best known parameter regime for approximation and sampling.
	Under the same conditions, we show the existence of deterministic approximation algorithms for the partition function with running time quasi-polynomial in the volume of $\volume$.
	As a result, we obtain the first efficient deterministic approximation algorithm for both models in the considered parameter regime.
    \item We show that, under suitable conditions, the partition functions of the hard-sphere model and the Widom--Rowlinson model can be seen as concentration points of hard-core partition functions on certain families of geometric random graphs.
    Further, we explore the asymptotic behavior of the concentration in terms of the number of vertices under different conditions on the region $\volume$.
% 	\item We show that discretization scheme where, for a given hard-constraint point process on a region $\volume \subset \R^\dimension$, we obtain an instance $(\graph, \weight)$ of the hard-core model.
%    We determine conditions such that the partition function $\hcPartitionFunction{\graph}{\weight}$ closely approximates the partition function of the hard-constraint point process and the number of vertices of $G$ is polynomial in the volume of $\volume$.
% 	\item We obtain new algorithmic results for continuous spin systems by combining our discretization scheme with known algorithms for the hard-core model.
% 	We present sufficient conditions for randomized approximation algorithms of the partition function and sampling algorithms for the Gibbs distribution of the hard-constraint point process with running time polynomial in the volume of $\volume$.
%    Under the same conditions, we show the existence of deterministic approximation algorithms for the partition function with running time quasi-polynomial in the volume of $\volume$.
%    \item We show that, under suitable conditions, the partition functions of hard-constrain point processes can be seen as concentration points of hard-core partition functions on certain families of geometric random graphs.
%    Further, we explore the asymptotic behavior of the concentration in terms of the number of vertices under different conditions on the region $\volume$.
\end{itemize}

%
% In this article we study the computational properties of very general continuous spin systems of spherical particles with hard-core interactions. Our techique reduces contninuous hard-constrained spin systems to discrete instances of the hard-core model of size polynomial in the volume. In a sense our method is a rigorous coarse graining~\cite[Chapter~8]{ProteinFoldingBook} approach for the hard-core model that allows us to use the extensive algorithmic results that are known for the hard-core model.

\subsection{Hard-constraint point processes}\label{sec:intro_model}
As discussed earlier, the hard-sphere model and the Widom--Rowlinson model have in common that they are characterized by a point process with hard-core interactions in a region $\volume$.
To remove redundancy, we define a slightly more general class of models that includes both of the above.
We refer to this class as \emph{hard-constraint point processes}.
Once we establish our results for this class of models, the corresponding statements for the hard-sphere model and the Widom--Rowlinson model follow immediately.

Let $\volume \subset \R^{\dimension}$ be bounded and measurable, and let $\numTypes \in \N_{\ge 1}$.
The model represents the distribution of particles of $\numTypes$ types, labeled by elements in $[\numTypes] \defeq [1, q] \cap \N$, on $\volume$.
Particles of the same types are assumed to be indistinguishable.
% This is, swapping particles of the same type with each other does not result in a different configuration.\todo{Is this sentence necessary?}
Let $\fugacity\colon [\numTypes] \to \R_{\ge 0}$ be a function that equips each particle type with a fugacity.
For each particle type $i \in [\numTypes]$, we assume that the positions of particles are distributed according to a (labeled) Poisson point process of intensity $\fugacity[i]$ on $\volume$.
To add the constraints to the model, let $\interactionMatrix \in \R_{\ge 0}^{\numTypes \times \numTypes}$ be a symmetric $\numTypes \times \numTypes$ matrix, called the \emph{interaction matrix}.
We condition the mixture of point processes by rejecting all configurations that contain particles at positions $x_1, x_2 \in \volume$ with corresponding particle types $\typeOf_1, \typeOf_2 \in [\numTypes]$ with distance $\dist{x_1}{x_2} < \interactionMatrix[\typeOf_1][\typeOf_2]$.
That is, the entries of $\interactionMatrix$  determine the minimum distance that particles of the respective types can have.
Especially note that $\interactionMatrix[\typeOf_1][\typeOf_2] = 0$ means that the particle types $\typeOf_1, \typeOf_2 \in [\numTypes]$ are not subjected to any pairwise constraints.
Thus, if all entries of $\interactionMatrix$ are set to $0$, we obtain a simple mixture of Poisson point processes.

For any instance $(\volume, \interactionMatrix, \fugacity)$ of a hard-constraint point process, the above characterizes a \emph{Gibbs distribution} as follows.
For all $\numParticles \in \N$ and all type assignments $\typeOf\colon [\numParticles] \to [\numTypes]$, let $\valid{\typeOf}{\interactionMatrix}\colon \left(\R^{\dimension}\right)^{\numParticles} \to \{0, 1\}$ be the function that indicates for a tuple of particles positions $\mathbold{x} = (x_i)_{i \in [\numParticles]} \in \left(\R^{\dimension}\right)^{\numParticles}$ whether it forms a valid configuration, assuming that for each $i \in [\numParticles]$, the particle at $x_i$ is of type $\typeOf[i]$.
Formally,
\[
	\valid{\typeOf}{\interactionMatrix}[\mathbold{x}] = \prod_{\substack{i, j \in [\numParticles]\colon\\ i<j}} \indicator{\dist{x_i}{x_j} \ge \interactionMatrix[\typeOf[i]][\typeOf[j]]} ,
\]
noting that $\valid{\typeOf}{\interactionMatrix}[\mathbold{x}] = 1$ for the case that $\numParticles = 0$.
For each tuple $(\mathbold{x}, \typeOf)$ as above, the Gibbs distribution of $(\volume, \interactionMatrix, \fugacity)$ is defined via the probability density
\[
	\generalGibbs{\volume}{\interactionMatrix}{\fugacity}[\mathbold{x}, \typeOf] = \frac{\frac{1}{\numParticles!} \left(\prod_{i \in [\numParticles]} \fugacity[\typeOf[i]] \right) \valid{\typeOf}{\interactionMatrix}[\mathbold{x}]}{\generalPartitionFunction{\volume}{\interactionMatrix}{\fugacity}} ,
\]
where the normalizing constant $\generalPartitionFunction{\volume}{\interactionMatrix}{\fugacity}$ is called the \emph{(grand canonical) partition function:}
\begin{align*}
%	\label{eq:partition_function}
	\generalPartitionFunction{\volume}{\interactionMatrix}{\fugacity}
	= 1 + \sum_{\numParticles \in \N_{\ge 1}} \frac{1}{\numParticles!} \sum_{\typeOf\colon [\numParticles] \to [\numTypes]} \left(\prod_{i \in [\numParticles]} \fugacity[\typeOf[i]] \right) \int_{\volume^{\numParticles}} \valid{\typeOf}{\interactionMatrix}[\mathbold{x}] \, \text{d} \lebesgue{\dimension \times \numParticles} ,
\end{align*}
with $\lebesgue{\dimension \times \numParticles}$ denoting the product of $\numParticles$ Lebesgue measures, each on $\R^{\dimension}$.
We proceed by showing how the hard-sphere model and the Widom--Rowlinson model are recovered from this framework.

\paragraph{Hard-sphere model:}
The hard-sphere model describes the distribution of particles of a single type subjected to hard-core interactions.
It is parameterized by a radius $\radius \in \R_{\ge 0}$ and a single fugacity $\fugacity \in \R_{\ge 0}$.
Informally speaking, it is a simple Poisson point process of intensity $\fugacity$, conditioned on no two points being closer than $2 \radius$ (i.e., when placing balls of radius $\radius$ at each of the points, they must be non-overlapping).
Using our framework, this is equivalent to a hard-constraint point process $(\volume, \interactionMatrixHS, \fugacity)$ with $\numTypes=1$ particle types, where we slightly abuse notation and treat $\fugacity$ as a constant function, and set $\interactionMatrixHS$ to be a $1 \times 1$ matrix containing only the entry $2 \radius$.

\paragraph{Widom--Rowlinson model:}
We consider the most general version of the Widom--Rowlinson model, although we might impose certain restrictions for some of our algorithmic results.
The Widom--Rowlinson model describes the interaction of particles of $\numTypes \ge 1$ types, each equipped with a radius $\radius_i \in \R_{\ge 0}$ and a fugacity $\fugacity_i \in \R_{\ge 0}$ for $i \in [\numTypes]$.
Informally speaking, the resulting distribution is a mixture of $\numTypes$ Poisson point processes, each with its own intensity $\fugacity_{i}$, with the condition that particles of the same type can be arbitrarily close to each other, but particles of different types $i, j \in [\numTypes]$ need to have a distance of at least $\radius_{i} + \radius_{j}$ (i.e., when placing a ball of radius $\radius_i$ at each point of type $i \in [\numTypes]$, balls of different types must be non-overlapping).
This is equivalent to a hard-constraint point process $(\volume, \interactionMatrixWR, \fugacity)$ with $\numTypes$ particle types, where we set $\fugacity[i] = \fugacity_i$ for all $i \in [\numTypes]$, and, for all $i, j \in [\numTypes]$,
\[
	\interactionMatrixWR[i][j] = \begin{cases}
		0 &\text{if $i=j$,} \\
		\radius_{i} + \radius_{j} &\text{otherwise.}
	\end{cases}
\]

\subsection{Reduction to a discrete hard-core model}\label{sec:intro_discretization}
We investigate the following natural discretization to turn a (continuous) hard-constraint point process into a (discrete) hard-core model.
Given an instance of a hard-constraint point process~$(\volume, \interactionMatrix, \fugacity)$ with $\numTypes \in \N_{\ge 1}$ particle types and a finite non-empty set of points $X \subset \volume$, we construct a simple undirected graph $\hcPPGraph{X} = (\hcPPVertices{X}, \hcPPEdges{X})$ as follows.
For each point in $X$ and each type, we have a vertex in $\hcPPVertices{X}$.
Two distinct vertices are connected by an edge if no two particles of the corresponding types are allowed to occupy the respective positions in a valid configuration.
%having particles of their corresponding types\todo{That sounds strange, doesn't it? \emph{if their corresponding particles have the same type}} and if their positions is\todo{are (?)} forbidden by the constraints encoded in $\interactionMatrix$.
Consequently, independent sets of the graph correspond to valid particle configurations.
We aim to `simulate' the original continuous model by a hard-core model on that graph.

Formally, we get the following construction:
\begin{itemize}
	\item For each point $x \in X$ and each type $i \in [\numTypes]$, we construct a vertex $\hcPPVertex{x}{i}$.
	Furthermore, for each $i \in [\numTypes]$, we set $\hcPPVertices{X}[i] = \big\{ \hcPPVertex{x}{i} \ \big\vert\ x \in X \big\}$, and we define $\hcPPVertices{X} = \bigcup_{i \in [\numTypes]} \hcPPVertices{X}[i]$.
	\item For each $i, j \in [\numTypes]$ and $x, y \in X$, we connect $\hcPPVertex{x}{i}, \hcPPVertex{y}{j} \in \hcPPVertices{X}$ with an edge in $\hcPPEdges{X}$ if
	\[
		\neg(x=y \text{ and } i=j) \text{  and  } \dist{x}{y} < \interactionMatrix[i][j] ,
	\]
	where the condition $\neg(x=y \text{ and } i=j)$ prevents self-loops.
\end{itemize}
Additionally, we define a function $\hcPPWeight{X}\colon \hcPPVertices{X} \to \R_{\ge 0}$ such that, for all $i \in [\numTypes]$ and all $x \in X$, $\hcPPWeight{X}[\hcPPVertex{x}{i}][\big] = \frac{\vol{\volume}}{\size{X}} \fugacity[i]$.
Note that for all $i \in [\numTypes]$ and all $x, y \in X$, this means that $\hcPPWeight{X}[\hcPPVertex{x}{i}][\big] = \hcPPWeight{X}[\hcPPVertex{y}{i}][\big]$, and we sometimes abuse notation and write $\hcPPWeight{X}[i]$ instead.
We call $(\hcPPGraph{X}, \hcPPWeight{X})$ the \emph{hard-core  representation} of $(\volume, \interactionMatrix, \fugacity)$ based on $X$.
Our goal is to find conditions, such that the hard-core partition function $\hcPartitionFunction{\hcPPGraph{X}}{\hcPPWeight{X}}$ closely approximates the partition function of original point process $\generalPartitionFunction{\volume}{\interactionMatrix}{\fugacity}$.

To this end, we focus on sets~$X$ that are the result of a function $\allocation\colon \volume \to X$, which we call an \emph{allocation.}
More precisely, we study allocations $\allocation\colon\volume\rightarrow X$ such that for each $x \in X$, the set $\invAllocation(x)$ has roughly the same volume and a bounded diameter.
This is expressed by the following definition.
\begin{restatable}[$\volumeError$-$\distanceError$-allocation]{definition}{edallocation} \label{def:allocation}
	Let $\volume \subset \R^{\dimension}$ be bounded and measurable, and assume $\vol{\volume} > 0$.
	For finite $X \subset \volume$, $\volumeError \in [0, 1)$, and $\distanceError \in \R_{>0}$, we call a function $\allocation\colon \volume \to X$ a \emph{$\volumeError$-$\distanceError$-allocation} for $X$ if, for all $x \in X$,
	\begin{enumerate}[(1)]
		\item $\invAllocation[x]$ is measurable and
		\[
		\left(1 - \volumeError\right) \frac{\vol{\volume}}{\size{X}} \le \vol{\invAllocation[x]} \le \left(1 + \volumeError\right) \frac{\vol{\volume}}{\size{X}} \textrm{, and,}
		\]
		\item for all $y \in \invAllocation[x]$, it holds that $\dist{x}{y} \le \distanceError$. \qedhere
	\end{enumerate}
\end{restatable}

Our main technical result is that, if the considered region $\volume$ is star-convex,\footnote[-0.1]{A region $\volume$ is star convex if there is a center $x \in \volume$ such that for every $y \in \volume$ the straight line from $x$ to $y$ is contained in $\volume$.} the partition function of the hard-core representation based on a point set $X \subset \volume$ closely approximates the partition function of the respective hard-constraint point process, whenever there is a $\volumeError$-$\distanceError$-allocation for $X$ with sufficiently small $\volumeError$ and $\distanceError$.
\begin{restatable*}{theorem}{DiscretisationError}
	\label{thm:discretization_error}
	Let $(\volume, \interactionMatrix, \fugacity)$ be a hard-constraint point process with $\numTypes \in \N_{\ge 1}$ particle types, and assume $\volume \subset \R^{\dimension}$ is star-convex.
	Furthermore, set $\interactionMatrix_{\min} = \inf_{i,j \in [\numTypes]} \{\interactionMatrix[i][j] \mid \interactionMatrix[i][j] > 0\}$,\footnote[-0.15]{We use the convention that $\inf \emptyset = \infty$, which in this setting means that $\interactionMatrix_{\min}$ can be chosen arbitrarily large.} and set $\fugacity_{\max} = \max_{i \in [\numTypes]} \fugacity[i]$.
	Let $X \subseteq \volume$ with $4 \fugacity_{\max} \vol{\volume} \le \size{X} < \infty$, and let $(\hcPPGraph{X}, \hcPPWeight{X})$ be the hard-core representation of $(\volume, \interactionMatrix, \fugacity)$ based on $X$.
	Finally, let $\volumeError \in \left[0, \frac{1}{2}\right]$ and $\distanceError \in \big[0, \frac{\interactionMatrix_{\min}}{2}\big]$, and assume that~$\volume$ has a $\volumeError$-$\distanceError$-allocation for $X$.
    Then
	\[
		\absolute{\hcPartitionFunction{\hcPPGraph{X}}{\hcPPWeight{X}} - \generalPartitionFunction{\volume}{\interactionMatrix}{\fugacity}}
		\le \Big(\eulerE^{\frac{8}{\size{X}} \sum_{i \in [\numTypes]} \fugacity[i]^2 \vol{\volume}^2} \eulerE^{\left( 2 \volumeError + \left(\frac{4 \distanceError }{\interactionMatrix_{\min}} \right)^{\dimension} \right) \sum_{i \in [\numTypes]} \fugacity[i] \vol{\volume}} - 1\Big) \generalPartitionFunction{\volume}{\interactionMatrix}{\fugacity}.
		\qedhere
	\]
\end{restatable*}

Since \Cref{thm:discretization_error} is at the core of our algorithmic results, we briefly overview its proof.
Given $X \subset \volume$ and a $\volumeError$-$\distanceError$-allocation~\allocation, for each $x\in X$, the total weight that~$\invAllocation[x]$ contributes to $\generalPartitionFunction{\volume}{\interactionMatrix}{\fugacity}$ is assigned to the total weight that the vertex of $\hcPPGraph{X}$ that corresponds to $x$ contributes to $\hcPartitionFunction{\hcPPGraph{X}}{\hcPPWeight{X}}$.
This weight allocation introduces three types of errors between the two partition functions, which we all bound.
The first one is due to the Poisson point process of the continuous model potentially generating more than one point in~\invAllocation[x], for some $x\in X$, which cannot be represented as part of the hard-core partition function.
To bound this error, we reduce the continuous model to an intermediate multiset version of the hard-core model, where valid configurations are allowed to contain multiple copies of the same vertex.
Then, we bound the difference between the partition function of the multiset hard-core model and $\hcPartitionFunction{\hcPPGraph{X}}{\hcPPWeight{X}}$.
The second error type is due to the regions $\invAllocation[x]$ not necessarily having the same volume.
However, this error is small, due to~\volumeError.
The third error type is due to configurations that are valid for the hard-constraint point process potentially being allocated to configurations that are invalid for the hard-core model and vice versa.
This happens because, e.g., two points $x_1,x_2\in\volume$ of type $\typeOf_1,\typeOf_2$, respectively, with $\dist{x_1}{x_2} \geq \interactionMatrix[\typeOf_1][\typeOf_2]$ get allocated to $\allocation(x_1),\allocation(x_2)\in X$ with $\dist{\allocation(x_1)}{\allocation(x_2)} < \interactionMatrix[\typeOf_1][\typeOf_2]$.
The parameter \distanceError ensures that each point $x\in \volume$ is allocated to a close by point $\allocation(x)\in X$.
This allows us to bound this error by the difference of two partition functions with slightly different values in their exclusion matrices~$\interactionMatrix$.
In Euclidean space, we can express this as a difference between partition functions with the original exclusion matrix $\interactionMatrix$ but on differently scaled versions of the regions $(1 \pm \alpha) \volume$ for some small $\alpha$.
Assuming star-convexity, we bound this difference by using the fact that a copy of $(1 - \alpha) \volume$ can be mapped into $(1 + \alpha) \volume$ while preserving distance and volume (\Cref{lemma:scaled_difference}), which proves \Cref{thm:discretization_error}.

Under a choice of $X$, with $\size{X}$ polynomial in $(\vol{\volume},\volumeError^{-1},\distanceError^{-1})$, \Cref{thm:discretization_error} is an approximation-preserving reduction from $\generalPartitionFunction{\volume}{\interactionMatrix}{\fugacity}$ to $\hcPartitionFunction{\hcPPGraph{X}}{\hcPPWeight{X}}$ that is polynomial in $\vol{\volume}$ and produces a graph~$\hcPPGraph{X}$ of size polynomial in $\vol{\volume}$.
In \Cref{sec:intro_approx}, we discuss how to choose the point set $X$ with $|X|\in O((\vol{\volume})^2)$ in the case of cubic regions $\volume = [0, \sidelength)^{\dimension}$ for some $\sidelength \in \R_{>0}$.
Applying \Cref{thm:discretization_error} allows us to use known algorithmic results on the resulting hard-core model to efficiently approximate $\generalPartitionFunction{\volume}{\interactionMatrix}{\fugacity}$.
Our bound on the required number of points in this setting is asymptotically tight in $\vol{\volume}$ for general hard-constraint point processes, as we show in \Cref{subsec:tightness}.

For cubic regions, we show that there is a $\volumeError$-$\distanceError$-allocation with~$\volumeError=0$.
However, allowing for a small margin of volume error~$\volumeError>0$ in \Cref{thm:discretization_error} results in interesting connections between the partition function of the hard-core model on random geometric graphs and the continuous partition function of the hard-constraint point process, as we discuss in \Cref{sec:intro_random_discretization}.
Namely, we show that the hard-core partition function on such geometric random graphs concentrates around that of a hard-constraint point process.
In fact, this even holds without the assumption of star-convexity, required in \Cref{thm:discretization_error}.

%We note that a discretization scheme similar to ours was used by \textcite{friedrich2021spectral} to obtain an approximation algorithm for the hard-sphere model with uniform fugacity.
For the case of the hard-sphere model on cubic regions, the same discretization scheme was used before to obtain a randomized approximation algorithm \cite{friedrich2021spectral}.
However, as this result requires the size of the graph to be super-exponential in $\vol{\volume}$, the previous algorithm needs to utilize the succinct representation and structural information of the produced graph $G$, which was highly specific to the hard-sphere model.
Since our refined bound yields much smaller graphs, we are able to use the results on the hard-core model as out-of-the-box algorithms.
Besides the aforementioned algorithmic simplifications, our result extends the applicability to other hard-constraint point processes including the Widom--Rowlinson model.
Further, it allows for the first efficient deterministic approximation within the same parameter regimes for both, the hard-sphere model and the Widom--Rowlinson model.
% \todo{already (?) This far in the sentence, it is not that clear to me that we are not the first ones to come up with this idea -- well, you know what I mean. Better now?}

% Therefore, the majority of algorithmic results for the hard-core model are not applicable, as their running time is polynomial in the number of vertices of the input graph.\todo{This sounds like bashing. Saying something like a \emph{direct} application or so would make it already far more tame. Reverse the order of the two sentences and merge.}

\subsection{Approximation algorithms via canonical discretization}\label{sec:intro_approx}

We show how \Cref{thm:discretization_error} can be used to obtain approximations for the partition functions of the hard-sphere model and the Widom--Rowlinson model.
To this end, we develop our result in the setting of general hard-constraint point processes and obtain the model-specific results as corollaries.
Although \Cref{thm:discretization_error} applies to a variety of regions $\volume \subset \R^{\dimension}$,
we focus on cubic regions $\volume = [0, \sidelength)^{\dimension}$ with  $\sidelength \in \R_{> 0}$ for stating our algorithmic results, as such regions exhibit a straightforward way to construct point sets with small allocations.
We suggest the following choice, which we call \emph{canonical discretization}.
The point set $X \subset \volume$ consists of a $d$-dimensional grid of side length $n$, and the allocation $\allocation$ maps each point $y=\big(y^{(i)}\big)_{i \in [\dimension]} \in \volume$ to the closest point in $x = \big(x^{(i)}\big)_{i \in [\dimension]} \in X$ such that, for all $i \in [\dimension]$, it holds that $x^{(i)} \le y^{(i)}$.
Since, for each $x\in X$, the set $\invAllocation(x)$ has the same volume, this results in a $0$-$\distanceError$-allocation.
By \Cref{thm:discretization_error}, in order to approximate the partition function of the hard-constraint point process~$\generalPartitionFunction{\volume}{\interactionMatrix}{\fugacity}$, it suffices to do so for the partition function of the hard-core model on $(\hcPPGraph{X},\hcPPWeight{X})$.

We first study the case where the considered hard-constraint point process only contains a single particle type (e.g., the hard-sphere model) or all particle types have the same fugacity (e.g., the Widom--Rowlinson model with uniform fugacities).
In this case, the hard-core representation is a uniform hard-core model, i.e., there is a constant $\weight$ such that, for each~$\vertex$, it holds that $\weight[\vertex]=\weight$.
A selection of algorithmic results are applicable in this setting.
Usually, approximation algorithms for the partition function of the uniform hard-core model apply to instances $(\graph,\weight)$ where $\graph$ has bounded maximum degree $\degree$ and $\weight<\criticalWeight{\degree}=\frac{\left(\degree - 1 \right)^{\degree - 1}}{\left(\degree - 2\right)^{\degree}} \approx \frac{\eulerE}{\degree}$.
On general graphs, this parameter regime is tight, as approximating the partition function for a hard-core model is known to be \NP-hard for instances $(\graph,\weight)$ with $\weight>\criticalWeight{\degree}$ \cite{2010:Sly:computational_transition,2014:Galanis:inapproximability_independent_hard_core}.
We obtain our algorithms by bounding the maximum degree of $\hcPPGraph{X}$ and using the known approximation algorithms for the hard-core model.
Namely, we use a Monte-Carlo Markov chain algorithm, for which polynomial running time was recently proven for all $\degree$ in a sequence of papers by \textcite{ALOG20,Chen2020rapid,ChenLV21,chen2021alldegrees}.
Furthermore, using the deterministic approximation algorithm of~\textcite{Weitz2006Counting} with running time in $n^{\bigO{\log\degree}}$, we obtain a deterministic algorithm with running time quasi-polynomial in the volume of \volume\!\!.\footnote[-0.1]{We get the same result by combining the results of \textcite{peters2019conjecture} with the method of \textcite{barvinok2016combinatorics}.}

In order to formally state our algorithmic results for the class of hard-constraint point processes, we introduce the \emph{volume exclusion matrix} $\volumeMatrix \in \R_{\ge 0}^{\numTypes \times \numTypes}$ of a hard-constraint point process $(\volume, \interactionMatrix, \fugacity)$ with $\numTypes \in \N_{\ge 1}$ particles types.
For all $i, j \in [\numTypes]$, it holds that $\volumeMatrix[i][j] = \vol{\ball{\interactionMatrix[i][j]}}$, where $\ball{r}$ denotes a $d$-dimensional ball of radius~$r$.
Intuitively, the entry $\volumeMatrix[i][j]$ gives an upper bound on the volume of the region around a particle of type $i \in [\numTypes]$ in which no particle of type $j \in [\numTypes]$ can be placed.

Our main algorithmic result for general hard-constraint point processes is the following.
\begin{restatable*}{theorem}{uniformapprox}\label{thm:uniform_approx}
	Let $(\volume, \interactionMatrix, \fugacity)$ be a hard-constraint point process with $\numTypes \in \N_{\ge 1}$ particle types, $\volume = [0, \sidelength)^{\dimension}$ for some $\sidelength \in \R_{>0}$, and let $\fugacity$ be a constant.
	Let $\volumeMatrix$ be the corresponding volume exclusion matrix and denote its $L_{1}$-norm by $\lOneNorm{\volumeMatrix}$.
	If
	\[
		\fugacity < \frac{\eulerE}{\lOneNorm{\volumeMatrix}},
	\]
	then for all $\approxError \in (0, 1]$, there is a randomized $\approxError$-approximation algorithm for $\generalPartitionFunction{\volume}{\interactionMatrix}{\fugacity}$ with running time in $\poly{\frac{\vol{\volume}}{\approxError}}[\big]$ and a quasi-polynomial deterministic approximation algorithm for $\generalPartitionFunction{\volume}{\interactionMatrix}{\fugacity}$ with running time $\big(\frac{\vol{\volume}}{\approxError}\big)^{\bigTheta{\ln \left(\vol{\volume}/\approxError\right)}}$.
\end{restatable*}

We derive our algorithmic results for the hard-core model and the Widom--Rowlinson model with uniform fugacities and radii from \Cref{thm:uniform_approx}.
For the hard-sphere model, we obtain the following corollary.
\begin{restatable*}{corollary}{monoatomichs}\label{cor:monoatomic_hs}
	Let $\volume = [0, \sidelength)^{\dimension}$ for some $\sidelength \in \R_{>0}$.
	Further, let $\radius \in \R_{>0}$ and $\fugacity \in \R_{\ge 0}$.
	Denote by $\hsPartitionFunction{\volume}{\radius}{\fugacity}$ the hard-sphere partition function on $\volume$ with particles of radius $\radius$ and fugacity~$\fugacity$.
	If
	\[
		\fugacity < \frac{\eulerE}{2^{\dimension} \vol{\ball{\radius}}},
	\]
	 then for all $\approxError \in (0, 1]$, there is a randomized $\approxError$-approximation algorithm for $\hsPartitionFunction{\volume}{\radius}{\fugacity}$ with running time in $\poly{\frac{\vol{\volume}}{\approxError}}[\big]$ and a quasi-polynomial deterministic $\approxError$-approximation algorithm for $\hsPartitionFunction{\volume}{\radius}{\fugacity}$ with running time $\big(\frac{\vol{\volume}}{\approxError}\big)^{\bigTheta{\ln \left(\vol{\volume}/\approxError\right)}}$.
\end{restatable*}
\Cref{cor:monoatomic_hs} achieves the best known algorithmic bound in terms of $\fugacity$~\cite{friedrich2021spectral,perkinsAlgorithm}.
%, which corresponds to the known bound for the absence of a phase transition~\cite{perkinsHardSpheres2020}.
Furthermore, the quasi-polynomial deterministic approximation algorithm is the first efficient deterministic approximation for this parameter regime and partially answers an open question of~\textcite{perkinsAlgorithm}.

For the Widom--Rowlinson model, \Cref{thm:uniform_approx} yields the following rigorous approximation result.
\begin{restatable*}{corollary}{uniformwr}\label{cor:uniform_wr}
	Let $\volume = [0, \sidelength)^{\dimension}$, $\sidelength \in \R_{>0}$, let and $\radius \in \R_{>0}$ and $\fugacity \in \R_{\ge 0}$.
	Let $\wrPartitionFunction{\volume}{\radius}{\fugacity}$ denote the Widom--Rowlinson partition function on $\volume$ with $\numTypes \in \N_{\ge 1}$ particle types, each of radius $\radius$ and fugacity $\fugacity$.
	If
	\[
		\fugacity < \frac{\eulerE}{(\numTypes - 1) 2^{\dimension} \vol{\ball{\radius}}},
	\]
	then for all $\approxError \in (0, 1]$, there is a randomized $\approxError$-approximation algorithm for $\wrPartitionFunction{\volume}{\radius}{\fugacity}$ with running time in $\poly{\frac{\vol{\volume}}{\approxError}}[\big]$ and a quasi-polynomial deterministic $\approxError$-approximation algorithm for $\wrPartitionFunction{\volume}{\radius}{\fugacity}$ with running time $\big(\frac{\vol{\volume}}{\approxError}\big)^{\bigTheta{\ln \left(\vol{\volume}/\approxError\right)}}$.
\end{restatable*}
To the best of our knowledge, no other efficient approximation algorithm is known for this parameter regime.
In fact, \Cref{cor:uniform_wr} also applies to the non-uniform Widom-Rowlinson model when replacing $\fugacity$ and $\radius$ with the maximum fugacity and maximum radius among all particle types (see \Cref{remark:univariate_to_multivariate}).

For non-uniform hard-constraint point processes, i.e., each particle type has its own fugacity, the resulting instance of the hard-core model is also non-uniform.
There are some algorithmic results for the non-uniform hard-core model in the literature, but none of them is known to establish a tight condition.
To our knowledge, the broadest parameter regime for the non-uniform hard-core model is achieved by a randomized approximation algorithm \cite{friedrich2020polymer}.
Applying this approximation algorithm to the instance $(\hcPPGraph{X},\hcPPWeight{X})$, we obtain the following algorithmic condition for the hard-constraint point process.

\begin{restatable*}{theorem}{nonUniformApprox}
	\label{thm:non_uniform_approx}
	Let $(\volume, \interactionMatrix, \fugacity)$ be a hard-constraint point process with $\numTypes \in \N_{\ge 1}$ particle types and $\volume = [0, \sidelength)^{\dimension}$ for some $\sidelength \in \R_{>0}$.
	Further, let $\volumeMatrix$ be the corresponding volume exclusion matrix.
	If there is a function $f\colon [\numTypes] \to \R_{>0}$ such that for all $i \in [\numTypes]$, it holds that
	\[
		f(i) > \sum_{j \in [\numTypes]} \volumeMatrix[i][j] f(j) \fugacity[j] ,
	\]
	then for all $\approxError \in (0, 1]$, there is a randomized $\approxError$-approximation algorithm for $\generalPartitionFunction{\volume}{\interactionMatrix}{\fugacity}$ with running time in $\poly{\frac{\vol{\volume}}{\approxError}, \ln\left(\frac{\max_{i \in [\numTypes]} f(i)}{\min_{i \in [\numTypes]} f(i)}\right)}$.
\end{restatable*}

For the non-uniform hard-core model, faster randomized algorithms~\cite{CGGPSV19} as well as deterministic algorithms~\cite{JKP19} have appeared in the literature.
These results are also applicable to our setting.
The stated results are chosen to yield the broadest parameter regime possible.

A particularly interesting corollary of~\Cref{thm:non_uniform_approx} is for the two-particle Widom--Rowlinson model.
\begin{restatable*}{corollary}{unbalancedWR}
	\label{cor:unbalanced_WR}
	Let $\volume = [0, \sidelength)^{\dimension}$ for some $\sidelength \in \R_{>0}$.
	Further, let $\radius \in \R_{>0}$, $\fugacity_1, \fugacity_2 \in \R_{\ge 0}$.
	Denote by $\wrPartitionFunction{\volume}{\radius}{\fugacity_1, \fugacity_2}$ the Widom--Rowlinson partition function on $\volume$ with $2$ particle types, both with the same radius $\radius$ but (possibly) different fugacities $\fugacity_1, \fugacity_2$.
	If
	\[
		\fugacity_1 \fugacity_2 < \frac{1}{4^{\dimension} \vol{\ball{\radius}}^{2}},
	\]
	then for all $\approxError \in (0, 1]$, there is a randomized $\approxError$-approximation algorithm for $\wrPartitionFunction{\volume}{\radius}{\fugacity_1, \fugacity_2}$ with running time in $\poly{\frac{\vol{\volume}}{\approxError}}[\big]$.
\end{restatable*}
Note that this condition only requires the product of the fugacities of both particle types to be bounded.
Thus, one fugacity can be arbitrarily large as long as the other is sufficiently small, leading to an approximation whenever the model is sufficiently unbalanced.
One key observation to explain this behavior is that the resulting hard-core representation used in \Cref{cor:unbalanced_WR} is a bipartite graph with different fugacities for each vertex partition.
Such hard-core models on unbalanced bipartite graphs are known to exhibit efficient approximations for broader parameter regimes, as shown by \textcite{CP20}.
However, their results do not appear to give feasible parameter regimes in our setting.

\subsection{Sampling via random perturbations}\label{sec:intro_sampling}

So far, we only discussed approximation results for the partition functions of the hard-sphere model and the Widom--Rowlinson model, although, in discrete spin systems, approximation algorithms often go hand in hand with sampling algorithms~\cite{JVV86}.
Unfortunately, for continuous spin systems, a natural barrier is that outputting a sample, i.e., a tuple of points that represents a valid configuration, requires infinite floating-point precision.
Thus, assuming a discrete computational model, as is common in computer science, no (approximate) sampling algorithm with meaningful error bounds in terms of total-variation distance can be obtained.

Assuming a computational model that performs arithmetic operations of floating-point values with arbitrary precision and can uniformly sample a random floating-point number from an interval (for more details see \Cref{assumption:floating_pointer}), we use our discretization $\hcPPGraph{X}$ to recover an approximate sampler for $\generalGibbs{\volume}{\interactionMatrix}{\fugacity}$.
Such an assumption is not satisfied by an ordinary computer.
However, for practical considerations, the floating-point precision of common discrete computational models might be seen as sufficient for applying our sampling approach.

Given a $\volumeError$-$\distanceError$-allocation~\allocation, our sampling algorithm (\Cref{alg:perturbation_algorithm}) first samples an independent set~$I$ from the Gibbs distribution of the hard-core model on $(\hcPPGraph{X},\hcPPWeight{X})$, using a known Markov chain method (see, e.g., \cite{ALOG20,Chen2020rapid,chen2021alldegrees}).
Recall that each vertex $\hcPPVertex{x}{i}$ of $\hcPPGraph{X}$ corresponds to a point $x\in X$ of a particle of type $i$.
Then, for each vertex $\hcPPVertex{x}{i}$ in~$I$, our sampler chooses a position in $\invAllocation(x)$ uniformly at random and places a particle of type~$i$ at this position.
This way, we obtain the following sampling  analogue of \Cref{thm:uniform_approx}.
\begin{restatable*}{theorem}{uniformSampling}
 	\label{thm:uniform_sampling}
	Suppose a computational model that satisfies \Cref{assumption:floating_pointer}.
	Let $(\volume, \interactionMatrix, \fugacity)$ be a hard-constraint point process with $\numTypes \in \N_{\ge 1}$ particle types, $\volume = [0, \sidelength)^{\dimension}$ for some $\sidelength \in \R_{>0}$, and let~$\fugacity$ be a constant.
	Let $\volumeMatrix$ be the corresponding volume exclusion matrix and denote its $L_{1}$-norm by~$\lOneNorm{\volumeMatrix}$.
	If
	\[
		\fugacity < \frac{\eulerE}{\lOneNorm{\volumeMatrix}},
	\]
	then for all $\samplingError \in (0, 1]$, there is an $\samplingError$-approximate sampler for $\generalGibbs{\volume}{\interactionMatrix}{\fugacity}$ with running time in $\poly{\frac{\vol{\volume}}{\samplingError}}[\big]$.
\end{restatable*}
Our sampling procedure does not always result in a valid configuration for the hard-constraint point process.
However, by bounding the total-variation distance between the two distributions, $\generalGibbs{\volume}{\interactionMatrix}{\fugacity}$ and the distribution of our sampler, we bound the probability of the sampler returning an invalid configuration by some small constant.
In this case, we reject this configuration and repeat the sampling procedure.

Observe that the sampling analogues of \Cref{cor:monoatomic_hs,cor:uniform_wr} immediately follow from the above theorem.
With \Cref{alg:perturbation_algorithm}, we also obtain the following theorem for the non-uniform case.
\begin{restatable*}{theorem}{nonUniformSampling}
	\label{thm:non_uniform_sampling}
	Suppose a computational model that satisfies \Cref{assumption:floating_pointer}.
	Let $(\volume, \interactionMatrix, \fugacity)$ be a hard-constraint point process with $\numTypes \in \N_{\ge 1}$ particle types and $\volume = [0, \sidelength)^{\dimension}$, $\sidelength \in \R_{>0}$.
	Further, let $\volumeMatrix$ be the corresponding volume exclusion matrix.
	If there is a function $f\colon [\numTypes] \to \R_{>0}$ such that for all $i \in [\numTypes]$,
	\[
		f(i) > \sum_{j \in [\numTypes]} \volumeMatrix[i][j] f(j) \fugacity[j] ,
	\]
	then for all $\samplingError \in (0, 1]$, there is an $\samplingError$-approximate sampler for $\generalGibbs{\volume}{\interactionMatrix}{\fugacity}$ with running time in $\poly{\frac{\vol{\volume}}{\approxError}}$ and $\poly{\ln\left(\frac{\max_{i \in [\numTypes]} f(i)}{\min_{i \in [\numTypes]} f(i)}\right)}$.
\end{restatable*}
As with the uniform case, the sampling analogue of \Cref{cor:unbalanced_WR} follows from the above theorem.

\subsection{Concentration of random discretizations}\label{sec:intro_random_discretization}

As seen in \Cref{sec:intro_approx,sec:intro_sampling}, our algorithmic results use a $0$-$\distanceError$-allocation.
However, \Cref{thm:discretization_error} allows for a more general set of allocations, where $\volumeError>0$.
Using this generality, we show sufficient conditions for the hard-core partition function of a discretization based on a uniformly random set of points $X$ to concentrate around a the partition function of the respective hard-contraint point process.
Central to this condition is the notion of a $\volumeScaling$-$\distanceError$-partitioning.

\begin{restatable}[$\volumeScaling$-$\distanceError$-partitioning]{definition}{edpartitioning} \label{def:partitioning}
	Let $\volume \subset \R^{\dimension}$ be bounded and measurable and assume $\vol{\volume} > 0$.
	For $\volumeScaling \in (0, 1]$ and $\distanceError \in \R_{>0}$, we say that a finite partitioning $(\volume_{k})_{k \in [\numPartitions]}$ of $\volume$ is a \emph{$\volumeScaling$-$\distanceError$-partitioning} of size $\numPartitions$ if and only if for all $k \in [\numPartitions]$ it holds that
	\begin{enumerate}[(1)]
		\item $\volume_{k}$ is measurable and
		\[
		\volumeScaling \frac{\vol{\volume}}{\numPartitions} \le \vol{\volume_{k}}
		\]
		\item for all $x, y \in \volume_{k}$ it holds that
		\[
		\dist{x}{y} \le \distanceError . \qedhere
		\]
	\end{enumerate}
\end{restatable}

Our main result in this direction essentially says that, for star-convex regions $\volume$, the existence of a polynomial size $\volumeScaling$-$\distanceError$-partitioning implies strong concentration of hard-core partition functions for random discretizations of polynomial size.

\begin{restatable*}{corollary}{randomdiscretizationconcentration}
	\label{cor:random_discretization_concentration}
	Let $(\volume, \interactionMatrix, \fugacity)$ be a hard-constraint point process with $\numTypes \in \N_{\ge 1}$ particle types, and assume $\volume \subset \R^{\dimension}$ is star-convex.
	Further, assume there is some $\volumeScaling \in (0, 1]$ such that for all $\distanceError \in \R_{>0}$ there is a $\volumeScaling$-$\distanceError$-partitioning of $\volume$ of size $\poly{\frac{\vol{\volume}}{\distanceError}}$.
	For all $\discretizationError \in (0, 1]$ and $\allocationProb \in (0, 1]$ there is some  $\numPoints_{\discretizationError, \allocationProb} \in \poly{\frac{\vol{\volume}}{\discretizationError} \ln \left(\frac{\vol{\volume}}{\discretizationError \allocationProb}\right)}$ such that for all $\numPoints \ge \numPoints_{\discretizationError, \allocationProb}$ it holds for $X \subset \volume$ with $\size{X} = \numPoints$ uniformly at random that
	\[
		\eulerE^{-\discretizationError}\generalPartitionFunction{\volume}{\interactionMatrix}{\fugacity}
		\le \hcPartitionFunction{\hcPPGraph{X}}{\hcPPWeight{X}}
		\le \eulerE^{\discretizationError}\generalPartitionFunction{\volume}{\interactionMatrix}{\fugacity}
	\]
	with probability at least $1-\allocationProb$.
\end{restatable*}
An immediate consequence of \Cref{cor:random_discretization_concentration} is that, for uniform random geometric graphs on $\volume = [0, \sidelength)^{\dimension}$ with edge connection threshold $2 \radius$ and appropriately chosen fugacities, the hard-core partition function concentrates with increasing number of vertices around the partition function of a hard-sphere model with particle radius $\radius$ on $\volume$.
There are ongoing efforts to prove concentration of different properties of geometric random graphs \cite{muller2008two,bachmann2016concentration,bachmann2018concentration}.
To the best of our knowledge, our result is the first such concentration bound for hard-core partition functions.
Thus, we believe this connection to be interesting in its own right.

As the proof of \Cref{cor:random_discretization_concentration} uses \Cref{thm:discretization_error}, it inherits the assumption that $\volume$ is star-convex.
However, without assuming star-convexity, a slightly weaker concentration holds.
\begin{restatable*}{theorem}{concentrationnonstarconvex}
	\label{thm:prob_bound_non_starconvex}
	Let $(\volume, \interactionMatrix, \fugacity)$ be a hard-constraint point process with $\numTypes \in \N_{\ge 1}$ particle types.
	Further, assume there is some $\volumeScaling \in (0, 1]$ such that for all $\distanceError \in \R_{>0}$ there is a $\volumeScaling$-$\distanceError$-partitioning of $\volume$ of size $\poly{\frac{\vol{\volume}}{\distanceError}}$.
	For all $\discretizationError \in (0, 1]$ and $\allocationProb \in (0, 1]$ there is some  $\numPoints_{\discretizationError, \allocationProb} \in \poly{\frac{\vol{\volume}}{\discretizationError \allocationProb}}$ such that for all $\numPoints \ge \numPoints_{\discretizationError, \allocationProb}$ it holds for $X \subset \volume$ with $\size{X} = \numPoints$ uniformly at random that
	\[
		\eulerE^{-\discretizationError}\generalPartitionFunction{\volume}{\interactionMatrix}{\fugacity}
		\le \hcPartitionFunction{\hcPPGraph{X}}{\hcPPWeight{X}}
		\le  \eulerE^{\discretizationError}\generalPartitionFunction{\volume}{\interactionMatrix}{\fugacity}
	\]
	with probability at least $1-\allocationProb$.
\end{restatable*}

The main idea behind the proof of \Cref{thm:prob_bound_non_starconvex} is that, even without star-convexity, it still holds that $\hcPartitionFunction{\hcPPGraph{X}}{\hcPPWeight{X}}$ is not much smaller than $\generalPartitionFunction{\volume}{\interactionMatrix}{\fugacity}$ if there is a suitable allocation for $X$.
It follows that, given~$\volume$ has a suitable partitioning, $\hcPartitionFunction{\hcPPGraph{X}}{\hcPPWeight{X}}$ concentrates from below.
We then establish a simple lemma that derives concentration from above, given sufficiently strong concentration from below for non-negative random variables.
We believe this tool to be of independent interest.

\begin{restatable*}{lemma}{modifiedmarkov}
	\label{lemma:modified_markov}
	Let $\randomVariable$ be a non-negative random variable with finite expectation.
	If there are $\genericError \in \R_{>0}$ and $\genericFailingProb \in [0, 1]$ such that $\Pr{\randomVariable < (1-\genericError) \E{\randomVariable}} \le \genericFailingProb$, then it holds that $\Pr{\randomVariable \ge (1 + c \genericError) \E{\randomVariable}} \le \frac{1}{c+1} \left(1 + \frac{\genericFailingProb (1 - \genericError)}{\genericError}\right)$ for all $c \in \R_{\ge 0}$.
\end{restatable*}

\subsection{Discussion and future directions}\label{sec:intro_outlook}
Our algorithmic results for the Widom--Rowlinson model and the hard-sphere model are centered around discretizations based on the hard-core model.
For the hard-sphere model, similar algorithmic results have recently been obtained without discretization by directly applying a Markov chain Monte Carlo method to the continuous problem, assuming a continuous model of computation \cite{perkinsAlgorithm}.
However, this approach is inherently probabilistic, whereas we have shown that discretization also yields efficient deterministic algorithms for the same parameter regime.
It would be interesting to see how to reproduce such results directly using the continuous model.
One approach could be to approximate the logarithm of the partition function instead, for example by using the cluster expansion.
This has successfully been done in the setting of discrete models \cite{patel2017deterministic,HPR19}.
However, in contrast to the discrete setting, we are not aware of any method to compute the coefficients of the expansion efficiently for continuous problems, which would be required for algorithmic application.

A major question for future research is if the parameter bounds for approximation that are presented in this paper can be improved further.
An obvious idea would be to use more detailed insights about the properties of the graphs that result from the discretization.
A candidate for that could be the connective constant, which was used before to improve algorithmic results for the hard-core model on certain graph classes \cite{sinclair2017spatial}.
However, arguments were made that for the canonical discretization of the hard-sphere model, the connective constant is asymptotically equivalent to the maximum degree of the graph (see~\cite{penrose1994self} and~\cite[Section~1.4]{friedrich2021spectral}).
On the other hand, \textcite{michelen2021potential} recently introduced the notion of a potential-weighted connective constant in the continuous setting to increase the regime of uniqueness of the Gibbs measure for the hard-sphere model in low dimensions by using structural properties of Euclidean space.
This raises the question if the potential-weighted connective constant can be translated to the discrete setting, leading to a similar improvement for approximating the partition function for low dimensions.

Another interesting algorithmic question is whether the quasi-polynomial running time of the deterministic approximation algorithm can be improved to a polynomial.
The two techniques yielding deterministic algorithms for the hard-core model have a running time of $n^{\bigO{\log\degree}}$ for graphs of maximum degree~$\degree$, which corresponds in our setting to the quasi-polynomial running time.
In the correlation decay method of Weitz~\cite{Weitz2006Counting}, this running time comes from computing the self-avoiding walk tree of $G$.
In Barvinok's interpolation method, this comes from enumerating induced connected subgraphs~\cite{patel2017deterministic}.
Constructing a deterministic algorithm for general graphs of maximum degree~$\degree$ and better running-time dependency on~$\degree$, if possible, seems to require significant conceptual insight.
However, one could hope that the symmetric structure of~$\hcPPGraph{X}$, produced by our discretization, leads to faster running times for the computational tasks used in any of these two algorithmic techniques.

Finally, it is worth pointing out that there are a variety of models in statistical physics that do not fit into the framework we considered.
For instance, in \cite{perkinsAlgorithm} the problem of obtaining a randomized approximating for the partition function of single type Gibbs point processes with general repulsive potentials is considered.
It would be interesting to see to what extend discretization based approaches can also be applied in this setting.

\section{Preliminaries} \label{sec:preliminaries}
Throughout the paper, denote by $\N$ the set of all non-negative integers, including $0$ and write $\N_{\ge 1}$ for the set $\N \setminus \{0\}$.
For all $n \in $ denote by $[n]$ the set $[1, n] \cap \N = \{1, 2, \dots, n\}$.
Furthermore, we use the convention that $\inf \emptyset = \infty$.

Let $S$ be some set and $f\colon S \to \R$.
For any $\alpha \in \R$ we write $\alpha f$ for the function $s \mapsto \alpha f(s)$ for all $s \in S$.
Within this work, we occasionally encounter functions whose domain is the empty set.
Recall that, following the definition of functions as subsets of the Cartesian product, there is exactly one function $f\colon \emptyset \to S$ for all sets $S$.

For any subset of $\dimension$-dimensional Euclidean space $A \subseteq \R^{\dimension}$ we write $\scaleUp A$ for the set $\{\scaleUp x \mid x \in A\}$.
Let $\dist{x}{y}$ denote the Euclidean distance between points $x,y \in \R^{\dimension}$.
For $x \in \R^{\dimension}$ and $\radius \in \R_{> 0}$, we write $\ball{\radius}[x]$ for a the open ball of radius $\radius$ around $x$, formally $\ball{\radius}[x] = \{y \in \R^{\dimension} \mid \dist{x}{y} < \radius\}$.
When the center $x$ of the ball is not important, we simply omit it and write $\ball{\radius}$ instead.
Further, we denote the Lebesgue measure on $\R^{\dimension}$ by $\lebesgue{\dimension}$.
We extend this notation by writing $\lebesgue{\dimension \times \numParticles}$ for the product of $\numParticles$ Lebesgue measures, each on $\R^{\dimension}$.
Whenever $\dimension$ and $\numParticles$ are clear from the context, we simply write $\vol{\cdot}$ to simplify notation.
This is, for all measurable sets $A \subset \R^{\dimension}$ we write $\vol{A} = \lebesgue{\dimension}[A]$, and for all measurable sets $B \subset \left(\R^{\dimension}\right)^{\numParticles}$ we write $\vol{B} = \lebesgue{\dimension \times \numParticles}[B]$.
Finally, a set $A \subseteq \R^{\dimension}$ is called \emph{star-convex} if there is a point $x \in A$ such that for all $y \in A$ and all $\alpha \in [0, 1]$ it holds that $\alpha x + (1-\alpha) y \in A$.
We say that any such point $x \in A$ is a center of $A$.

\subsection{Notions of approximation}
% Another relevant concept are $\error$-approximations.
For $x \in \R_{>0}$ and some $\error \in \R_{>0}$ we call $x' \in \R_{>0}$ an \emph{$\error$-approximation for~$x$} if $\eulerE^{-\error} x \le x' \le \eulerE^{\error} x$.
Further, we call a real-valued random variable $X$ a \emph{randomized $\error$-approximation for~$x$} if
$\Pr{\eulerE^{-\error} x \le X \le \eulerE^{\error} x} \ge 3/4$.

For approximate sampling we use the following formal notion.
Let $\xi$  be a probability distribution on a state space $\Omega$.
For $\error \in (0, 1]$, we say that a distribution $\xi'$ on $\Omega$ is an \emph{$\error$-approximation of $\xi$} if $\dtv{\xi}{\xi'} \le \error$, where $\dtv{\cdot}{\cdot}$ denotes the total-variation distance.
Further, we say that we can \emph{$\error$-approximately sample from $\xi$} if we can sample from any distribution $\xi'$ such that $\xi'$ is an $\error$-approximation of $\xi$.

\subsection{The hard-core model}
We recall the definition of the hard-core model from the introduction.
For an undirected graph $G=(V, E)$ the multivariate hard-core model is defined as a tuple $(\graph, \weight)$ where $\weight\colon \vertices \to \R_{\ge 0}$.
Let $\independentSets{\graph}$ denote the set of all independent sets in $\graph$.
We assign each independent set $\independentSet \in \independentSets{\graph}$ a weight of $\prod_{\vertex \in \independentSet} \weight[\vertex]$.
The multivariate hard-core partition function is now defined as the sum of weights of all independent sets.
Further, the Gibbs distribution of the multivariate hard-core model assigns each independent set a probability proportional to its weight.
Formally, that is
\[
	\hcPartitionFunction{\graph}{\weight} = \sum_{\independentSet \in \independentSets{\graph}} \prod_{\vertex \in \independentSet} \weight[\vertex]
	\text{  and  }
	\hcGibbs{\graph}{\weight}[\independentSet] = \frac{\prod_{\vertex \in \independentSet} \weight[\vertex]}{\hcPartitionFunction{\graph}{\weight}}
	\text{ for all $\independentSet \in \independentSets{\graph}$.}
\]

If $\weight$ is a constant, we call the hard-core model univariate (or uniform).
Note that in most of the existing literature the term hard-core model refers to the univariate version.
Further, the multivariate case appears sometimes as polymer models in the algorithmic context.
However, in fact there are some conceptual differences between both terms, such as that polymer models are used to encode other computational problems and tend to be exponentially large in the original problem instance, which means that no explicit representation of the polymer graph is assumed.
We stick to this conceptual difference, and thus, we use the term hard-core model.
Unless we explicitly state that any of our statements is limited to the univariate case, we will use the term hard-core model for the multivariate version.
Further, whenever we deal with the univariate case, we will abuse notation and simply write $\weight$ as a constant instead of $\weight[\vertex]$.
Note that this simplifies the weight of an independent $\independentSet \in \independentSets{\graph}$ to $\weight^{\size{\independentSet}}$, which directly carries over to the definitions of the partition function and the Gibbs distribution.

The following paragraphs summarize known algorithmic results for approximate sampling from the Gibbs distribution and approximating the partition function.
For obtaining a deterministic approximation of the partition function, two main algorithmic approaches are known, namely the computational tree method, introduces by \textcite{Weitz2006Counting}, and the polynomial interpolation method of \textcite{patel2017deterministic,barvinok2016combinatorics}.
To get a randomized approximation, usually a Markov chain Monte Carlo algorithm based on self-reducibility is used.
For this approach, it is necessary to approximately sample from the Gibbs distribution.
This is often done via \emph{Glauber dynamics}, which is an ergodic Markov chain with state space $\independentSets{\graph}$ and stationary distribution $\hcGibbs{\graph}{\weight}$.
In the multivariate setting, the update rule of this Markov chain works as follows.
Let $\independentSet_{t}$ be the state of the chain at time $t$.
The configuration at time $t+1$ is obtained by first choosing a vertex $\vertex \in \vertices$ uniformly at random.
With probability $\frac{1}{1 + \weight[\vertex]}$, we set $\independentSet_{t+1} = \independentSet_{t} \setminus \{\vertex\}$ (note that nothing changes if $\vertex \notin \independentSet_{t}$).
Otherwise, with probability $\frac{\weight[\vertex]}{1 + \weight[\vertex]}$, we set $\independentSet_{t+1} = \independentSet_{t} \cup \{\vertex\}$ if $\independentSet_{t} \cup \{\vertex\} \in \independentSets{\graph}$.

\paragraph*{Univariate approximation:}
For the univariate hard-core model, the ability to efficiently approximate the partition function is closely related to the maximum degree of the graph.
For a graph of maximum degree $\degree \in \N_{\ge 2}$, the tree threshold is defined as
\[
	\criticalWeight{\degree} = \frac{\left(\degree - 1 \right)^{\degree - 1}}{\left(\degree - 2\right)^{\degree}} .
\]

The following theorems show some of the algorithmic results that are known below this threshold.
The first statement deals with approximately sampling from the Gibbs distribution of the univariate hard-core model via Glauber dynamics.

\begin{theorem}[{\cite[Theorem $1.3$]{chen2021alldegrees}}]
	\label{thm:hc_sampling_univariate}
	Let $\graph = (\vertices, \edges)$ be an undirected graph with maximum vertex degree bounded by $\degree \in \N_{\ge 2}$.
	For all $\weight \in \R_{\ge 0}$ with $\weight < \criticalWeight{\degree}$
	the Glauber dynamics on the univariate hard-core model $(\graph, \weight)$ are mixing in time $\poly{\size{\vertices}}$, and for all $\samplingError \in (0, 1]$ we can $\samplingError$-approximately sample from the Gibbs measure $\hcGibbs{\graph}{\weight}$ in time $\bigO{n\left(n \ln(\degree) + \ln\left(\frac{1}{\samplingError}\right)\right)}$.
\end{theorem}

It is well known that, based on such a sampler, a Markov chain Monte Carlo approximation for the partition function of the hard-core model can be obtained using self-reducibility.
Thus, \Cref{thm:hc_sampling_univariate} leads to the following approximation result.

\begin{theorem}
	\label{thm:hc_fpras_univariate}
	Let $\graph = (\vertices, \edges)$ be an undirected graph with maximum vertex degree bounded by $\degree \in \N_{\ge 2}$ and let $\weight \in \R_{\ge 0}$ with $\weight < \criticalWeight{\degree}$.
	For all $\approxError \in (0, 1]$ there is a randomized $\approxError$-approximation algorithm for the univariate hard-core partition function $\hcPartitionFunction{\graph}{\weight}$ with running time $\poly{\frac{n}{\approxError}}$.
\end{theorem}

Besides this randomized algorithm, the following result can be obtained deterministically.
\begin{theorem}[{\cite[Theorem $2.7$]{Weitz2006Counting}}]
	\label{thm:hc_fptas_univariate}
	Let $\graph = (\vertices, \edges)$ be an undirected graph with maximum vertex degree bounded by $\degree \in \N_{\ge 2}$ and let $\weight \in \R_{\ge 0}$ with $\weight < \criticalWeight{\degree}$.
	For all $\approxError \in (0, 1]$ there is a deterministic $\approxError$-approximation algorithm for the univariate hard-core partition function $\hcPartitionFunction{\graph}{\weight}$ with running time $\left(\frac{\size{\vertices}}{\approxError}\right)^{\bigTheta{\ln (\degree)}}$.
\end{theorem}

Note that \Cref{thm:hc_fpras_univariate} is based on the computational tree method by \textcite{Weitz2006Counting}.
A similar result can be obtained by combining the zero-freeness result by \textcite{peters2019conjecture} and the polynomial interpolation method as proposed by \textcite{patel2017deterministic}.

\begin{remark} \label{remark:univariate_hc}
	In fact, \Cref{thm:hc_sampling_univariate,thm:hc_fpras_univariate,thm:hc_fptas_univariate} can also be applied in the non-uniform setting when the condition is replaced by
	\[
		\weight_{\max} < \criticalWeight{\degree} ,
	\]
	where $\weight_{\max} = \max_{\vertex \in \vertices} \weight[\vertex]$.
\end{remark}

\paragraph*{Multivariate approximation:}
One of the least restrictive conditions that appeared in the literature for efficient approximation of multivariate hard-core partition functions is the clique dynamics condition, introduced by \textcite{friedrich2020polymer}.
Initially, the condition was used for proving rapid mixing of clique dynamics for abstract polymer models.
However Glauber dynamics can, in fact, be seen as a special case of clique dynamics for cliques of size $1$.
The result for Glauber dynamics is summarized in the following statement.

\begin{theorem}[{\cite[Theorem $2$]{friedrich2020polymer}}]
	\label{thm:clique_dynamics_condition_sampling}
	Let $(\graph, \weight)$ be a multivariate hard-core model.
	If there is a function $f\colon \vertices \to \R_{>0}$ such that for all $\vertex \in \vertices$ it holds that
	\[
		f(v) \ge \sum_{w \in \neighbors{\vertex}} f(w) \frac{\weight[w]}{1 + \weight[w]} ,
	\]
	where $\neighbors{\vertex}$ denotes the neighborhood of $\vertex$ in $\graph$, then the Glauber dynamics on $(\graph, \weight)$ are mixing in time $\poly{\size{\vertices}, \ln\left(\frac{\max_{\vertex \in \vertices} f(\vertex)}{\min_{\vertex \in \vertices} f(\vertex)}\right)}$.
	Further, for all $\samplingError \in (0, 1]$ we can sample $\samplingError$-approximately from the Gibbs distribution $\hcGibbs{\graph}{\weight}$ in time $\poly{\frac{\size{\vertices}}{\samplingError}, \ln\left(\frac{\max_{\vertex \in \vertices} f(\vertex)}{\min_{\vertex \in \vertices} f(\vertex)}\right)}$.
\end{theorem}

Further, the following approximation result is obtained via a Markov chain Monte Carlo approach and \Cref{thm:clique_dynamics_condition_sampling}.
\begin{theorem}[{\cite[Theorem $3$]{friedrich2020polymer}}]
	\label{thm:clique_dynamics_condition_approx}
	Let $(\graph, \weight)$ be a multivariate hard-core model.
	If there is a function $f\colon \vertices \to \R_{>0}$ such that for all $\vertex \in \vertices$ it holds that
	\[
		f(v) \ge \sum_{w \in \neighbors{\vertex}} f(w) \frac{\weight[w]}{1 + \weight[w]} ,
	\]
	then for all $\approxError \in (0, 1]$ there is a randomized $\approxError$-approximation algorithm for the partition function $\hcPartitionFunction{\graph}{\weight}$ with running time in $\poly{\frac{\size{\vertices}}{\approxError}, \ln\left(\frac{\max_{\vertex \in \vertices} f(\vertex)}{\min_{\vertex \in \vertices} f(\vertex)}\right)}$.
\end{theorem}

\section{Point processes with general hard constraints}\label{sec:model}
We recall the definition of a hard-constraint process from the introduction and give some alternative forms of the partition function.
Let $\volume \subset \R^{\dimension}$ be bounded and measurable, and let $\numTypes \in \N_{\ge 1}$.
The model represents the distribution of particles of $\numTypes$ types, labeled by elements in $[\numTypes] \defeq [1, q] \cap \N$, on $\volume$.
Particles of the same types are assumed to be indistinguishable.
% This is, swapping particles of the same type with each other does not result in a different configuration.\todo{Is this sentence necessary?}
Let $\fugacity\colon [\numTypes] \to \R_{\ge 0}$ be a function that equips each particle type with a fugacity.
For each particle type $i \in [\numTypes]$, we assume that the positions of particles are distributed according to a (labeled) Poisson point process of intensity $\fugacity[i]$ on $\volume$.
To add the constraints to the model, let $\interactionMatrix \in \R_{\ge 0}^{\numTypes \times \numTypes}$ be a symmetric $\numTypes \times \numTypes$ matrix, called the \emph{interaction matrix}.
We condition the mixture of point processes by rejecting all configurations that contain particles at positions $x_1, x_2 \in \volume$ with corresponding particle types $\typeOf_1, \typeOf_2 \in [\numTypes]$ whenever $\dist{x_1}{x_2} < \interactionMatrix[\typeOf_1][\typeOf_2]$.
That is, the entries of $\interactionMatrix$  determine the minimum distance that particles of the respective types can have.
Especially note that $\interactionMatrix[\typeOf_1][\typeOf_2] = 0$ means that the particle types $\typeOf_1, \typeOf_2 \in [\numTypes]$ are not subjected to any pairwise constraints.

For any instance $(\volume, \interactionMatrix, \fugacity)$ of a hard-constraint point process that satisfies the requirements above, this description characterizes a \emph{Gibbs distribution} as follows.
For all $\numParticles \in \N$ and all type assignments $\typeOf\colon [\numParticles] \to [\numTypes]$ let $\valid{\typeOf}{\interactionMatrix}\colon \left(\R^{\dimension}\right)^{\numParticles} \to \{0, 1\}$ be the function that indicates for a tuple of particles positions $\mathbold{x} = (x_1, \dots, x_{\numParticles}) \in \left(\R^{\dimension}\right)^{\numParticles}$ if it forms a valid configuration, assuming that for each $i \in [\numParticles]$ the particle at $x_i$ is of type $\typeOf[i]$.
Formally, this is
\[
	\valid{\typeOf}{\interactionMatrix}[\mathbold{x}] =
	\valid{\typeOf}{\interactionMatrix}[x_1, \dots, x_{\numParticles}] = \prod_{\substack{i, j \in [\numParticles]\colon\\ i<j}} \indicator{\dist{x_i}{x_j} \ge \interactionMatrix[\typeOf[i]][\typeOf[j]]} ,
\]
where $\valid{\typeOf}{\interactionMatrix}[\mathbold{x}] = 1$ for the case that $\numParticles = 0$.
For each tuple $(\mathbold{x}, \typeOf)$ as above, the Gibbs distribution of $(\volume, \interactionMatrix, \fugacity)$ is defined via the probability density
\[
	\generalGibbs{\volume}{\interactionMatrix}{\fugacity}[\mathbold{x}, \typeOf] = \frac{\frac{1}{\numParticles!} \left(\prod_{i \in [\numParticles]} \fugacity[\typeOf[i]] \right) \valid{\typeOf}{\interactionMatrix}[\mathbold{x}]}{\generalPartitionFunction{\volume}{\interactionMatrix}{\fugacity}} ,
\]
where the normalizing constant $\generalPartitionFunction{\volume}{\interactionMatrix}{\fugacity}$ is called the \emph{partition function}.
This is
\begin{align}
	\label{eq:partition_function}
	\generalPartitionFunction{\volume}{\interactionMatrix}{\fugacity}
	= 1 + \sum_{\numParticles \in \N_{\ge 1}} \frac{1}{\numParticles!} \sum_{\typeOf\colon [\numParticles] \to [\numTypes]} \left(\prod_{i \in [\numParticles]} \fugacity[\typeOf[i]] \right) \int_{\volume^{\numParticles}} \valid{\typeOf}{\interactionMatrix}[\mathbold{x}] \, \text{d} \lebesgue{\dimension \times \numParticles} .
\end{align}
The main goal of this paper is to investigate condition for an efficient (randomized) approximation of $\generalPartitionFunction{\volume}{\interactionMatrix}{\fugacity}$.
Occasionally it will be convenient to include the case $\numParticles = 0$ in the sum in \cref{eq:partition_function}.
To this end, we use the convention $\int_{\volume^{0}} \valid{\typeOf}{\interactionMatrix}[\mathbold{x}] \, \text{d} \lebesgue{\dimension \times 0} = 1$ to rewrite \cref{eq:partition_function} as
\[
	\generalPartitionFunction{\volume}{\interactionMatrix}{\fugacity}
	= \sum_{\numParticles \in \N} \frac{1}{\numParticles!} \sum_{\typeOf\colon [\numParticles] \to [\numTypes]} \left(\prod_{i \in [\numParticles]} \fugacity[\typeOf[i]] \right) \int_{\volume^{\numParticles}} \valid{\typeOf}{\interactionMatrix}[\mathbold{x}] \, \text{d} \lebesgue{\dimension \times \numParticles} .
\]

We can further simplify this notation by defining a weight
\[
	\generalWeight{\interactionMatrix}{\fugacity}[\mathbold{x}, \typeOf] = \frac{1}{\numParticles!} \left(\prod_{i \in [\numParticles]} \fugacity[\typeOf[i]] \right) \valid{\typeOf}{\interactionMatrix}[\mathbold{x}]
\]
for each $\mathbold{x} = (x_1, \dots, x_{\numParticles}) \in \volume^{\numParticles}$ and $\typeOf\colon [\numParticles] \to [\numTypes]$.
Using this weight function, we can rewrite the Gibbs density and the partition function slightly more compact as
\[
	\generalGibbs{\volume}{\interactionMatrix}{\fugacity}[\mathbold{x}, \typeOf] = \frac{\generalWeight{\interactionMatrix}{\fugacity}[\mathbold{x}, \typeOf]}{\generalPartitionFunction{\volume}{\interactionMatrix}{\fugacity}}
	\text{  and  }
	\generalPartitionFunction{\volume}{\interactionMatrix}{\fugacity} = \sum_{\numParticles \in \N} \sum_{\typeOf\colon [\numParticles] \to [\numTypes]} \int_{\volume^{\numParticles}} \generalWeight{\interactionMatrix}{\fugacity}[\mathbold{x}, \typeOf] \, \text{d} \lebesgue{\dimension \times \numParticles} .
\]

For formalizing our conditions for efficient approximation of the partition function of a hard-constraint point process $(\volume, \interactionMatrix, \fugacity)$ it will come in handy to introduce the \emph{volume exclusion matrix} $\volumeMatrix \in \R_{\ge 0}^{\numTypes \times \numTypes}$ the model.
It is defined by the entries $\volumeMatrix[i][j] = \vol{\ball{\interactionMatrix[i][j]}}$ for all $i, j \in [\numTypes]$.
Intuitively speaking, this means that the entry $\volumeMatrix[i][j]$ gives an upper bound on the volume of the region around a point $x \in \volume$ in which no particle of type $j \in [\numTypes]$ can be placed, given we place a particle of type $i \in [\numTypes]$ at $x$. 

For our algorithmic results, we investigate asymptotic behavior in the volume~$\vol{\volume}$, while we assume the parameters $\numTypes$, $\dimension$, and the entries of $\interactionMatrix$ as well as the range of $\fugacity$ to be constant.

\subsection{Scaling, bounds and monotonicity}
\label{subsec:model}
We start by collecting some useful properties of the partition function of general hard-constraint point processes, which will come in handy when bounding the speed of convergence of our discretization.
Our first lemma gives a trivial upper bound on the partition function.
\begin{lemma}
	\label{lemma:bound_partition_function}
	For every hard-constraint point process $(\volume, \interactionMatrix, \fugacity)$ with $\numTypes \in \N_{\ge 1}$ particle types it holds that
	\[
			\generalPartitionFunction{\volume}{\interactionMatrix}{\fugacity} \le \eulerE^{\sum_{i \in [\numTypes]} \fugacity[i] \vol{\volume}} .
		\qedhere
	\]
\end{lemma}
\begin{proof}
	Fix some $\numParticles \in \N_{\ge 0}$ and observe that for any $\typeOf\colon [\numParticles] \to [\numTypes]$ it holds that
	\[
		\int_{\volume^{\numParticles}} \valid{\typeOf}{\interactionMatrix}[\mathbold{x}] \, \text{d} \lebesgue{\dimension \times \numParticles} \le \vol{\volume^{\numParticles}} = \vol{\volume}^{\numParticles} .
	\]
	Now, fix some $(\numParticles[1], \dots, \numParticles[\numTypes]) \in \N^{\numTypes}$ such that $\sum_{i \in [\numTypes]} \numParticles[i] = \numParticles$.
	Observe that there are exactly $\frac{\numParticles!}{\prod_{i \in [\numTypes]} \numParticles[i] !}$ type assignments $\typeOf\colon [\numParticles] \to [\numTypes]$ such that $\size{\invTypeOf[i]} = \numParticles[i]$ for all $i \in [\numTypes]$.
	Further, for each such type assignment, it holds that
	\[
		\prod_{i \in [\numParticles]} \fugacity[\typeOf[i]] = \prod_{i \in [\numTypes]} \fugacity[i]^{\numParticles[i]} .
	\]
	Thus, we obtain
	\begin{align*}
		\generalPartitionFunction{\volume}{\interactionMatrix}{\fugacity}
		&\le 1 + \sum_{\numParticles \in \N_{\ge 1}} \frac{1}{\numParticles!} \sum_{\typeOf\colon [\numParticles] \to [\numTypes]} \left(\prod_{i \in [\numParticles]} \fugacity[\typeOf[i]] \right) \vol{\volume}^{\numParticles} \\
		&= 1 + \sum_{\numParticles \in \N_{\ge 1}} \vol{\volume}^{\numParticles} \sum_{\substack{(\numParticles[1], \dots, \numParticles[\numTypes]) \in \N^{\numTypes}\colon\\ \numParticles[1] + \dots + \numParticles[\numTypes] = \numParticles}} \prod_{i \in [\numTypes]} \frac{\fugacity[i]^{\numParticles[i]}}{\numParticles[i] !} \\
		&= 1 + \sum_{\numParticles \in \N_{\ge 1}} \vol{\volume}^{\numParticles} \frac{\left(\sum_{i \in [\numTypes]} \fugacity[i] \right)^{\numParticles}}{\numParticles !},
	\end{align*}
	where the last equality is obtained via the multinomial theorem.
	The claim follows from the Taylor expansion of the exponential function around $0$.
\end{proof}

Our next observation is concerned with the monotonicity of the partition function of hard-constraint point processes with respect to $\volume$, $\interactionMatrix$ and $\fugacity$.
\begin{observation}
	\label{obs:monotonicity_model}
	Fix some $\numTypes \in \N_{\ge 1}$.
	Let $\volume, \volume_{1}, \volume_{2} \subset \R^{\dimension}$, let $\interactionMatrix, \interactionMatrix_1, \interactionMatrix_2 \in \R_{\ge 0}^{\numTypes \times \numTypes}$ be symmetric and let $\fugacity\colon [\numTypes] \to \R_{\ge 0}, \fugacity_1\colon [\numTypes] \to \R_{\ge 0}, \fugacity_2\colon [\numTypes] \to \R_{\ge 0}$.
	It holds that:
	\begin{enumerate}[(a)]
		\item If $\volume_1 \subseteq \volume_2$ then $\generalPartitionFunction{\volume_1}{\interactionMatrix}{\fugacity} \le \generalPartitionFunction{\volume_2}{\interactionMatrix}{\fugacity}$.
		\item If $\interactionMatrix_{1}(i, j) \le \interactionMatrix_{2}(i, j)$ for all $i, j \in [\numTypes]$ then $\generalPartitionFunction{\volume}{\interactionMatrix_1}{\fugacity} \ge \generalPartitionFunction{\volume}{\interactionMatrix_2}{\fugacity}$.
		\item If $\fugacity_{1}(i) \le \fugacity_{2}(i)$ for all $i \in [\numTypes]$ then $\generalPartitionFunction{\volume}{\interactionMatrix}{\fugacity_1} \le \generalPartitionFunction{\volume}{\interactionMatrix}{\fugacity_2}$.
		\qedhere
	\end{enumerate}
\end{observation}

It is a folklore that the partition function the hard-sphere model and the Widom--Rowlinson model are log-submodular in the considered region.
The following lemma formalizes this idea for general hard-constraint point processes.
\begin{lemma}
	\label{lemma:submulti_model}
	Let $\volume_1 \subset \R^{\dimension}$ and $\volume_2 \subset \R^{\dimension}$ be bounded and measurable, and let $\numTypes \in \N_{\ge 1}$.
	For symmetric $\interactionMatrix \in \R_{\ge 0}^{\numTypes \times \numTypes}$ and all functions $\fugacity\colon [\numTypes] \to \R_{\ge 0}$ it holds that
	\[
		\generalPartitionFunction{\volume_1 \cup \volume_2}{\interactionMatrix}{\fugacity} \le \generalPartitionFunction{\volume_1}{\interactionMatrix}{\fugacity} \cdot  \generalPartitionFunction{\volume_2}{\interactionMatrix}{\fugacity}
		\qedhere
	\]
\end{lemma}
\begin{proof}
	First, note that it is sufficient to prove the claim for disjoint sets $\volume_1$ and $\volume_2$.
	For non-disjoint $\volume_1, \volume_2$ the claim then follows from partitioning $\volume_1 \cup \volume_2$ into the disjoint sets $\volume_1$ and $\volume_2 \setminus \volume_1$, and by noting that $\volume_2 \setminus \volume_1 \subseteq \volume_2$, which by \Cref{obs:monotonicity_model} implies $\generalPartitionFunction{\volume_2 \setminus \volume_1}{\interactionMatrix}{\fugacity} \le \generalPartitionFunction{\volume_2}{\interactionMatrix}{\fugacity}$.

	Thus, assume $\volume_1 \cap \volume_2 = \emptyset$ and set $\volume = \volume_1 \cup \volume_2$.
	Fix some $\numParticles \in \N$.
	For $\mathbold{x} = (x_1, \dots, x_{\numParticles}) \in \volume^{\numParticles}$, let $\numParticlesFrom{1}{\mathbold{x}}$ denote the number from $\volume_1$.
	Formally, this is $\size{\{i \in [\numParticles] \mid x_i \in \volume_1\}}$.
	As a first step, we rewrite the partition function on $\volume$ as
	\begin{align*}
		\generalPartitionFunction{\volume}{\interactionMatrix}{\fugacity}
		&= \sum_{\numParticles \in \N} \frac{1}{\numParticles!} \sum_{\typeOf\colon [\numParticles] \to [\numTypes]} \left(\prod_{i \in [\numParticles]} \fugacity[\typeOf[i]] \right) \int_{\volume^{\numParticles}} \valid{\typeOf}{\interactionMatrix}[\mathbold{x}] \, \text{d} \lebesgue{\dimension \times \numParticles} \\
		&= \sum_{\numParticles \in \N} \sum_{\numParticles[1] \in \N} \frac{1}{\numParticles!} \sum_{\typeOf\colon [\numParticles] \to [\numTypes]} \left(\prod_{i \in [\numParticles]} \fugacity[\typeOf[i]] \right) \int_{\volume^{\numParticles}} \indicator{\numParticlesFrom{1}{\mathbold{x}} = \numParticles[1]}\valid{\typeOf}{\interactionMatrix}[\mathbold{x}] \, \text{d} \lebesgue{\dimension \times \numParticles} \\
		&= \sum_{\numParticles \in \N} \sum_{\numParticles[1] \in \N} \frac{1}{\numParticles!} \int_{\volume^{\numParticles}} \indicator{\numParticlesFrom{1}{\mathbold{x}} = \numParticles[1]} \sum_{\typeOf\colon [\numParticles] \to [\numTypes]} \left(\prod_{i \in [\numParticles]} \fugacity[\typeOf[i]] \right) \valid{\typeOf}{\interactionMatrix}[\mathbold{x}] \, \text{d} \lebesgue{\dimension \times \numParticles} .
	\end{align*}

	Fix some $\numParticles \in \N$ and let $\permutation\colon [\numParticles] \to [\numParticles]$ be a bijective (i.e., $\permutation$ a permutation of $[\numParticles]$).
	For every type assignment $\typeOf\colon [\numParticles] \to [\numTypes]$ we write $\permutation \composition \typeOf$ for the composition of both functions.
	That is, for all $i \in [\numParticles]$ we have $(\permutation \composition \typeOf)(i) = \typeOf[\permutation[i]]$.
	Further, for all $\mathbold{x} = (x_1, \dots, x_{\numParticles}) \in \volume^{\numParticles}$ we abuse notation and write $\permutation[\mathbold{x}] = (x_{\permutation[1]}, \dots, x_{\permutation[\numParticles]})$.
	Observe that for all permutations $\permutation\colon [\numParticles] \to [\numParticles]$ it holds that
	\[
		\valid{\typeOf}{\interactionMatrix}[\mathbold{x}] = \valid{\permutation \composition \typeOf}{\interactionMatrix}[\permutation[\mathbold{x}]] .
	\]

	Next, observe that for all $\mathbold{x} \in \volume^{\numParticles}$ there is exactly one permutation $\permutationOf{\mathbold{x}}\colon [\numParticles] \to [\numParticles]$ such that $\permutationOf{\mathbold{x}}[\mathbold{x}] \in \volume_1^{\numParticlesFrom{1}{\mathbold{x}}} \times \volume_2^{\numParticles - \numParticlesFrom{1}{\mathbold{x}}}$ and the order among particle centers from $\volume_1$ and from $\volume_2$ is preserved.
	Formally, that means $\permutationOf{\mathbold{x}}[\mathbold{x}] \in \volume_1^{\numParticlesFrom{1}{\mathbold{x}}} \times \volume_2^{\numParticles - \numParticlesFrom{1}{\mathbold{x}}}$, for all $i, j \in [\numParticlesFrom{1}{\mathbold{x}}]$ it holds that $\permutationOf{\mathbold{x}}[i] < \permutationOf{\mathbold{x}}[j]$ if and only if $i < j$, and for all $i, j \in [\numParticles] \setminus [\numParticlesFrom{1}{\mathbold{x}}]$ it also holds that $\permutationOf{\mathbold{x}}[i] < \permutationOf{\mathbold{x}}[j]$ if and only if $i < j$.
	We can now rewrite
	\begin{align*}
		\sum_{\typeOf\colon [\numParticles] \to [\numTypes]} \left(\prod_{i \in [\numParticles]} \fugacity[\typeOf[i]] \right) \valid{\typeOf}{\interactionMatrix}[\mathbold{x}]
		&= \sum_{\typeOf\colon [\numParticles] \to [\numTypes]} \left(\prod_{i \in [\numParticles]} \fugacity[\typeOf[i]] \right) \valid{\permutationOf{\mathbold{x}} \composition \typeOf}{\interactionMatrix}[\permutationOf{\mathbold{x}}[\mathbold{x}]] \\
		&= \sum_{\typeOf\colon [\numParticles] \to [\numTypes]} \left(\prod_{i \in [\numParticles]} \fugacity[\permutationOf{\mathbold{x}}[\typeOf[i]]] \right) \valid{\permutationOf{\mathbold{x}} \composition \typeOf}{\interactionMatrix}[\permutationOf{\mathbold{x}}[\mathbold{x}]] ,
	\end{align*}
	where the last equality is due to the fact that applying the permutation $\permutationOf{\mathbold{x}}$ does only change the order of the finite product.

	For the next step, note that for every fixed permutation $\permutationOf{\mathbold{x}}$ and every fixed type assignment $\typeOf$ there is exactly one type assignment $\typeOf'$ such that $\typeOf' = (\permutationOf{\mathbold{x}} \composition \typeOf)$.
	This means, applying $\permutationOf{\mathbold{x}}$ to $\typeOf$ does only change the order of summation.
	We obtain
	\[
		\sum_{\typeOf\colon [\numParticles] \to [\numTypes]} \left(\prod_{i \in [\numParticles]} \fugacity[\typeOf[i]] \right) \valid{\typeOf}{\interactionMatrix}[\mathbold{x}]
		= \sum_{\typeOf\colon [\numParticles] \to [\numTypes]} \left(\prod_{i \in [\numParticles]} \fugacity[\typeOf[i]] \right) \valid{\typeOf}{\interactionMatrix}[\permutationOf{\mathbold{x}}[\mathbold{x}]] .
	\]

	To proceed, we need to introduce an additional bit of notation.
	Let $\numParticles[1], \numParticles[2] \in \N$, and let $\typeOf_1\colon [\numParticles[1]] \to [\numTypes]$ and $\typeOf_2\colon [\numParticles[2]] \to [\numTypes]$.
	We write $\typeOf_1 \combineAssignments \typeOf_2$ to denote the function $(\typeOf_1 \combineAssignments \typeOf_2)\colon [\numParticles[1] + \numParticles[2]] \to [\numTypes]$ such that for all $i \in [\numParticles[1] + \numParticles[2]]$ we have
	\[
		(\typeOf_1 \combineAssignments \typeOf_2) (i) = \begin{cases}
			\typeOf_1 (i) \text{ if } i \le \numParticles[1] \\
			\typeOf_2 (i - \numParticles[1]) \text{ otherwise. }
		\end{cases}
	\]
	Note that for each type assignment $\typeOf\colon [\numParticles] \to [\numTypes]$ and every fixed $\numParticles[1], \numParticles[2] \in \N$ such that $\numParticles[1] + \numParticles[2] = \numParticles$, there is exactly one pair of assignments $\typeOf_1\colon [\numParticles[1]] \to [\numTypes]$ and $\typeOf_2\colon [\numParticles[2]] \to [\numTypes]$ such that $\typeOf = (\typeOf_1 \combineAssignments \typeOf_2)$.
	Further, for each such decomposition and all $(x_1, \dots, x_{\numParticles}) \in \volume^{\numParticles}$ it holds that
	\[
		\valid{\typeOf}{\interactionMatrix}[x_1, \dots, x_{\numParticles}]
		= \valid{\typeOf_1 \combineAssignments \typeOf_2}{\interactionMatrix}[x_1, \dots, x_{\numParticles}]
		\le \valid{\typeOf_1}{\interactionMatrix}[x_1, \dots, x_{\numParticles[1]}] \cdot \valid{\typeOf_2}{\interactionMatrix}[x_{\numParticles[1]+1}, \dots, x_{\numParticles}] .
	\]

	For $\mathbold{x} = (x_1, \dots, x_{\numParticles}) \in \volume^{\numParticles}$, let $\projection{1}[\mathbold{x}] = (x_{\permutationOf{\mathbold{x}}[1]}, \dots, x_{\permutationOf{\mathbold{x}}[\numParticlesFrom{1}{\mathbold{x}}]})$, and let $\projection{2}[\mathbold{x}] = (x_{\permutationOf{\mathbold{x}}[\numParticlesFrom{1}{\mathbold{x}} + 1]}, \dots, x_{\permutationOf{\mathbold{x}}[\numParticles]})$.
	This is, $\projection{1}$ and $\projection{2}$ project $\mathbold{x}$ to those component that come from $\volume_1$ and $\volume_2$ respectively while preserving the order among them.
	We obtain
	\begin{align*}
		\sum_{\typeOf\colon [\numParticles] \to [\numTypes]} &\left(\prod_{i \in [\numParticles]} \fugacity[\typeOf[i]] \right) \valid{\typeOf}{\interactionMatrix}[\permutationOf{\mathbold{x}}[\mathbold{x}]] \\
		&\le \sum_{\typeOf_1\colon [\numParticlesFrom{1}{\mathbold{x}}] \to [\numTypes]}
		\left(\prod_{i \in [\numParticlesFrom{1}{\mathbold{x}}]} \fugacity[\typeOf_1(i)] \right)
		\sum_{\typeOf_2\colon [\numParticles - \numParticlesFrom{1}{\mathbold{x}}] \to [\numTypes]} \left(\prod_{i \in [\numParticles - \numParticlesFrom{1}{\mathbold{x}}]} \fugacity[\typeOf_2(i)] \right) \valid{\typeOf_1}{\interactionMatrix}[\projection{1}[\mathbold{x}]] \cdot \valid{\typeOf_2}{\interactionMatrix}[\projection{2}[\mathbold{x}]]
	\end{align*}
	and by substituting this into the definition of the partition function we get
	\begin{align*}
		\generalPartitionFunction{\volume}{\interactionMatrix}{\fugacity}
		\le \sum_{\numParticles \in \N}
		&\frac{1}{\numParticles!} \sum_{\numParticles[1] \in [\numParticles]} \sum_{\typeOf_1\colon [\numParticles[1]] \to [\numTypes]}
		\left(\prod_{i \in [\numParticles[1]]} \fugacity[\typeOf_1(i)] \right)
		\sum_{\typeOf_2\colon [\numParticles - \numParticles[1]] \to [\numTypes]} \left(\prod_{i \in [\numParticles - \numParticles[1]]} \fugacity[\typeOf_2(i)] \right) \\
		&\cdot \int_{\volume^{\numParticles}} \indicator{\numParticlesFrom{1}{\mathbold{x}} = \numParticles[1]} \valid{\typeOf_1}{\interactionMatrix}[\projection{1}[\mathbold{x}]] \cdot \valid{\typeOf_2}{\interactionMatrix}[\projection{2}[\mathbold{x}]] \, \text{d} \lebesgue{\dimension \times \numParticles} .
	\end{align*}

	Now, note that for all $\numParticles \in \N$ and $\numParticles[1] \in [\numParticles]$ and for all $\mathbold{y} \in \volume_1^{\numParticles[1]}$ and $\mathbold{z} \in \volume_2^{\numParticles - \numParticles[1]}$ it holds that there are exactly $\frac{\numParticles !}{\numParticles[1]! (\numParticles - \numParticles[1])!}$ configurations $\mathbold{x} \in \volume^{\numParticles}$ such that $\projection{1}[\mathbold{x}] = \mathbold{y}$ and $\projection{2}[\mathbold{x}] = \mathbold{z}$.
	Thus, it holds that
	\begin{align*}
		\generalPartitionFunction{\volume}{\interactionMatrix}{\fugacity}
		&\le \sum_{\numParticles \in \N}
		\sum_{\numParticles[1] \in [\numParticles]} \frac{1}{\numParticles[1]!} \sum_{\typeOf_1\colon [\numParticles[1]] \to [\numTypes]}
		\left(\prod_{i \in [\numParticles[1]]} \fugacity[\typeOf_1(i)] \right) \frac{1}{(\numParticles - \numParticles[1])!}
		\sum_{\typeOf_2\colon [\numParticles - \numParticles[1]] \to [\numTypes]} \left(\prod_{i \in [\numParticles - \numParticles[1]]} \fugacity[\typeOf_2(i)] \right) \\
			&\hspace{3em} \cdot \int_{\volume_1^{\numParticles[1]}} \int_{\volume_2^{\numParticles - \numParticles[1]}} \valid{\typeOf_1}{\interactionMatrix}[\mathbold{y}] \cdot \valid{\typeOf_2}{\interactionMatrix}[\mathbold{z}] \, \text{d} \lebesgue{\dimension \times (\numParticles - \numParticles[1])} \text{d} \lebesgue{\dimension \times \numParticles[1]} \\
		&= \sum_{\numParticles \in \N} \sum_{\numParticles[1] \in [\numParticles]} \left(\frac{1}{\numParticles[1]!} \sum_{\typeOf_1\colon [\numParticles[1]] \to [\numTypes]} \left(\prod_{i \in [\numParticles[1]]} \fugacity[\typeOf_1(i)] \right) \int_{\volume_1^{\numParticles[1]}} \valid{\typeOf_1}{\interactionMatrix}[\mathbold{y}] \, \text{d} \lebesgue{\dimension \times \numParticles[1]} \right) \\
			&\hspace{3em} \cdot \left(\frac{1}{(\numParticles - \numParticles[1])!} \sum_{\typeOf_2\colon [\numParticles - \numParticles[1]] \to [\numTypes]} \left(\prod_{i \in [\numParticles - \numParticles[1]]} \fugacity[\typeOf_2(i)] \right) \int_{\volume_2^{\numParticles - \numParticles[1]}} \valid{\typeOf_2}{\interactionMatrix}[\mathbold{z}] \, \text{d} \lebesgue{\dimension \times (\numParticles - \numParticles[1])} \right) \\
		&= \left(\sum_{\numParticles[1] \in \N} \frac{1}{\numParticles[1]!} \sum_{\typeOf_1\colon [\numParticles[1]] \to [\numTypes]}
		\left(\prod_{i \in [\numParticles[1]]} \fugacity[\typeOf_1(i)] \right) \int_{\volume_1^{\numParticles[1]}} \valid{\typeOf_1}{\interactionMatrix}[\mathbold{y}] \, \text{d} \lebesgue{\dimension \times \numParticles[1]} \right) \\
			&\hspace{3em} \cdot \left( \sum_{\numParticles[2] \in \N} \frac{1}{\numParticles[2]!} \sum_{\typeOf_2\colon [\numParticles[2]] \to [\numTypes]} \left(\prod_{i \in [\numParticles[2]]} \fugacity[\typeOf_2(i)] \right) \int_{\volume_2^{\numParticles[2]}} \valid{\typeOf_2}{\interactionMatrix}[\mathbold{z}] \, \text{d} \lebesgue{\dimension \times \numParticles[2]} \right) \\
		&= \generalPartitionFunction{\volume_1}{\interactionMatrix}{\fugacity} \cdot \generalPartitionFunction{\volume_2}{\interactionMatrix}{\fugacity} ,
	\end{align*}
	which concludes the proof.
	\end{proof}

To provide a sufficient error bound for our discretization, it will be helpful to translate between scaling the particle radii $\interactionMatrix$, the fugacities $\fugacity$ and the space itself $\volume$.
In Euclidean space, this can be done as stated in the following lemma.
\begin{lemma}
	\label{lemma:scaled_radius}
	Let $(\volume, \interactionMatrix, \fugacity)$ be a hard-constraint point process with $\numTypes \in \N_{\ge 1}$ particle types in $\dimension \in \N_{\ge 1}$ dimensions.
	For all $\scaleUp \in \R_{>0}$ it holds that
	\[
	\generalPartitionFunction{\volume}{\frac{1}{\scaleUp}\interactionMatrix}{\fugacity} = \generalPartitionFunction{\scaleUp \volume}{\interactionMatrix}{\frac{1}{\scaleUp^{\dimension}}\fugacity} .
	\qedhere
	\]
\end{lemma}
\begin{proof}
	The key observation to prove this claim is that for all pairs of points $x, y \in \volume$ and types $i, j \in [\numTypes]$ it holds that $\dist{x}{y} \ge \frac{1}{\scaleUp} \interactionMatrix[i][j]$ if and only if $\dist{\scaleUp x}{\scaleUp y} \ge \interactionMatrix[i][j]$.
	Thus, we get for all $\numParticles \in \N_{\ge 1}$ and all type assignments $\typeOf\colon [\numParticles] \to [\numTypes]$ that
	\begin{align*}
		\int_{\volume^{\numParticles}} \valid{\typeOf}{\frac{1}{\scaleUp}\interactionMatrix}[x_1, \dots, x_{\numParticles}] \, \text{d} \lebesgue{\dimension \times \numParticles}
		&= \int_{\volume^{\numParticles}} \valid{\typeOf}{\interactionMatrix}[\scaleUp x_1, \dots, \scaleUp x_{\numParticles}] \,\d \lebesgue{\dimension \times \numParticles} \\
		&= \frac{1}{\scaleUp^{\dimension \numParticles}} \int_{(\scaleUp \volume)^{\numParticles}} \valid{\typeOf}{\interactionMatrix}[x_1, \dots, x_{\numParticles}] \,\d \lebesgue{\dimension \times \numParticles} ,
	\end{align*}
	where the additional factor of $\frac{1}{\scaleUp^{\dimension \numParticles}}$ compensates for the scaling of the Lebesgue measure under linear transformation.

	Substituting this back into the definition of the partition function we obtain
	\begin{align*}
		\generalPartitionFunction{\volume}{\frac{1}{\scaleUp}\interactionMatrix}{\fugacity}
		&= 1 + \sum_{\numParticles \in \N_{\ge 1}} \frac{1}{\numParticles!} \sum_{\typeOf\colon [\numParticles] \to [\numTypes]} \left(\prod_{i \in [\numParticles]} \fugacity[\typeOf[i]] \right) \frac{1}{\scaleUp^{\dimension \numParticles}} \int_{(\scaleUp \volume)^{\numParticles}} \valid{\typeOf}{\interactionMatrix}[\mathbold{x}] \, \text{d} \lebesgue{\dimension \times \numParticles} \\
		&= 1 + \sum_{\numParticles \in \N_{\ge 1}} \frac{1}{\numParticles!} \sum_{\typeOf\colon [\numParticles] \to [\numTypes]} \left(\prod_{i \in [\numParticles]} \frac{\fugacity[\typeOf[i]]}{\scaleUp^{\dimension}} \right) \int_{(\scaleUp \volume)^{\numParticles}} \valid{\typeOf}{\interactionMatrix}[\mathbold{x}] \, \text{d} \lebesgue{\dimension \times \numParticles} \\
		&= \generalPartitionFunction{\scaleUp \volume}{\interactionMatrix}{\frac{1}{\scaleUp^{\dimension}}\fugacity} ,
	\end{align*}
	which proves the claim.
\end{proof}

Based on \Cref{lemma:bound_partition_function,obs:monotonicity_model} and \Cref{lemma:submulti_model,lemma:scaled_radius} we can now prove the following bound on the difference of partition functions on star-convex regions $\volume$ with differently scaled particle interaction matrices $\interactionMatrix$, which will be the very core of our convergence result.
\begin{lemma}
	\label{lemma:scaled_difference}
	Let $\scaleUp \in \left[0, 1\right]$, and let $\volume \subset \R^{\dimension}$ be bounded, measurable and star-convex.
	For any hard-constraint point process $(\volume, \interactionMatrix, \fugacity)$ with $\numTypes \in \N_{\ge 1}$ particle types it holds that
	\[
		\generalPartitionFunction{\volume}{(1 - \scaleUp)\interactionMatrix}{\fugacity} - \generalPartitionFunction{\volume}{(1 + \scaleUp)\interactionMatrix}{\fugacity}
		\le \left(\eulerE^{\left(2 \scaleUp\right)^{\dimension} \sum_{i \in [\numTypes]} \fugacity[i] \vol{\volume}} - 1 \right) \generalPartitionFunction{\volume}{\interactionMatrix}{\fugacity} .
		\qedhere
	\]
\end{lemma}
\begin{proof}
	Let $x_0 \in \volume$ denote the center of $\volume$.
	Note that the partition function is invariant under translation.
	Thus, we can assume that $x_0 = 0$.
	Note that this especially implies for all $\beta_1 \in [0, 1]$ and $\beta_2 \in [1, \infty)$ that $\beta_1 \volume \subseteq \volume \subseteq \beta_2 \volume$.

	We proceed by rewriting the difference based on \Cref{lemma:scaled_radius} as
	\begin{align}
		\notag
		\generalPartitionFunction{\volume}{(1 - \scaleUp)\interactionMatrix}{\fugacity}
		&- \generalPartitionFunction{\volume}{(1 + \scaleUp)\interactionMatrix}{\fugacity} \\
		&= \generalPartitionFunction{\frac{1}{1 - \scaleUp}\volume}{\interactionMatrix}{(1 - \scaleUp)^{\dimension} \fugacity} - \generalPartitionFunction{\frac{1}{1 + \scaleUp} \volume}{\interactionMatrix}{(1 + \scaleUp)^{\dimension} \fugacity} .
	\label{lemma:scaled_difference:eq1}
	\end{align}

	Using the fact that $\frac{1}{1 + \scaleUp}\volume \subseteq \volume \subseteq \frac{1}{1 - \scaleUp} \volume$ and \Cref{lemma:submulti_model} we obtain
	\[
		\generalPartitionFunction{\frac{1}{1 - \scaleUp}\volume}{\interactionMatrix}{(1 - \scaleUp)^{\dimension} \fugacity}
		\le \generalPartitionFunction{\frac{1}{1 - \scaleUp}\volume \setminus \frac{1}{1 + \scaleUp} \volume}{\interactionMatrix}{(1 - \scaleUp)^{\dimension} \fugacity} \generalPartitionFunction{\frac{1}{1 + \scaleUp} \volume}{\interactionMatrix}{(1 - \scaleUp)^{\dimension} \fugacity} .
	\]
	Because $(1-\scaleUp)^{\dimension} \fugacity[i] \le (1+\scaleUp)^{\dimension} \fugacity[i]$ for all $i \in [\numTypes]$ we can further bound this by \Cref{obs:monotonicity_model} to get
	\begin{align}
		\generalPartitionFunction{\frac{1}{1 - \scaleUp}\volume}{\interactionMatrix}{(1 - \scaleUp)^{\dimension} \fugacity}
		\le \generalPartitionFunction{\frac{1}{1 - \scaleUp}\volume \setminus \frac{1}{1 + \scaleUp} \volume}{\interactionMatrix}{(1 - \scaleUp)^{\dimension} \fugacity} \generalPartitionFunction{\frac{1}{1 + \scaleUp} \volume}{\interactionMatrix}{(1 + \scaleUp)^{\dimension} \fugacity}
		\label[ineq]{lemma:scaled_difference:eq2} .
	\end{align}

	Combining \cref{lemma:scaled_difference:eq1,lemma:scaled_difference:eq2} yields
	\begin{align}
		\notag
		\generalPartitionFunction{\volume}{(1 - \scaleUp)\interactionMatrix}{\fugacity} - &\generalPartitionFunction{\volume}{(1 + \scaleUp)\interactionMatrix}{\fugacity} \\
		\notag
		&\le \generalPartitionFunction{\frac{1}{1 - \scaleUp}\volume \setminus \frac{1}{1 + \scaleUp} \volume}{\interactionMatrix}{(1 - \scaleUp)^{\dimension} \fugacity} \generalPartitionFunction{\frac{1}{1 + \scaleUp} \volume}{\interactionMatrix}{(1 + \scaleUp)^{\dimension} \fugacity} \\
		\notag
			&\hspace{3em} - \generalPartitionFunction{\frac{1}{1 + \scaleUp} \volume}{\interactionMatrix}{(1 + \scaleUp)^{\dimension} \fugacity} \\
		&= \left(\generalPartitionFunction{\frac{1}{1 - \scaleUp}\volume \setminus \frac{1}{1 + \scaleUp} \volume}{\interactionMatrix}{(1 - \scaleUp)^{\dimension} \fugacity} - 1 \right) \generalPartitionFunction{\frac{1}{1 + \scaleUp} \volume}{\interactionMatrix}{(1 + \scaleUp)^{\dimension} \fugacity}
		\label[ineq]{lemma:scaled_difference:eq3} .
	\end{align}

	Using \Cref{lemma:bound_partition_function} and the facts that $\frac{1}{1 + \scaleUp}\volume  \subseteq \frac{1}{1 - \scaleUp} \volume$, $\vol{\frac{1}{1 + \scaleUp}\volume} = \left(\frac{1}{1+\scaleUp}\right)^{\dimension} \vol{\volume}$ and $\vol{\frac{1}{1 - \scaleUp}\volume} = \left(\frac{1}{1-\scaleUp}\right)^{\dimension} \vol{\volume}$ we further get
	\begin{align}
		\notag
		\generalPartitionFunction{\frac{1}{1 - \scaleUp}\volume \setminus \frac{1}{1 + \scaleUp} \volume}{\interactionMatrix}{(1 - \scaleUp)^{\dimension} \fugacity}
		&\le \eulerE^{(1-\scaleUp)^{\dimension} \sum_{i \in [\numTypes]}\fugacity[i] \vol{\frac{1}{1 - \scaleUp}\volume \setminus \frac{1}{1 + \scaleUp}\volume}} \\
		\notag
		&= \eulerE^{(1-\scaleUp)^{\dimension} \left(\frac{1}{1 - \scaleUp} - \frac{1}{1 + \scaleUp}\right)^{\dimension} \sum_{i \in [\numTypes]}\fugacity[i] \vol{\volume}} \\
		\notag
		&= \eulerE^{\left(1 - \frac{1-\scaleUp}{1 + \scaleUp}\right)^{\dimension} \sum_{i \in [\numTypes]}\fugacity[i] \vol{\volume}} .
	\end{align}
	Further, for $0 \le \scaleUp$, it holds that
	\[
		\left(1 - \frac{1-\scaleUp}{1 + \scaleUp}\right)^{\dimension}
		\le \left(1 - (1 - 2 \scaleUp)\right)^{\dimension}
		= \left(2 \scaleUp\right)^{\dimension} .
	\]
	Thus, we obtain
	\begin{align}
		\generalPartitionFunction{\frac{1}{1 - \scaleUp}\volume \setminus \frac{1}{1 + \scaleUp} \volume}{\interactionMatrix}{(1 - \scaleUp)^{\dimension} \fugacity}
		\le \eulerE^{\left(2 \scaleUp\right)^{\dimension} \sum_{i \in [\numTypes]}\fugacity[i] \vol{\volume}} .
		\label[ineq]{lemma:scaled_difference:eq4}
	\end{align}

	Substituting \cref{lemma:scaled_difference:eq4} into \cref{lemma:scaled_difference:eq3} now gives us
	\[
		\generalPartitionFunction{\volume}{(1 - \scaleUp)\interactionMatrix}{\fugacity} - \generalPartitionFunction{\volume}{(1 + \scaleUp)\interactionMatrix}{\fugacity}
		\le \left(\eulerE^{\left(2 \scaleUp\right)^{\dimension} \sum_{i \in [\numTypes]}\fugacity[i] \vol{\volume}} - 1 \right) \generalPartitionFunction{\frac{1}{1 + \scaleUp} \volume}{\interactionMatrix}{(1 + \scaleUp)^{\dimension} \fugacity} .
	\]
	Finally, because \Cref{lemma:scaled_radius}, \Cref{obs:monotonicity_model} and the fact that $(1 + \scaleUp) \interactionMatrix[i][j] \ge \interactionMatrix[i][j]$ for all $i, j \in [\numTypes]$ we have
	\[
		\generalPartitionFunction{\frac{1}{1 + \scaleUp} \volume}{\interactionMatrix}{(1 + \scaleUp)^{\dimension} \fugacity}
		= \generalPartitionFunction{\volume}{(1 + \scaleUp)\interactionMatrix}{\fugacity}
		\le \generalPartitionFunction{\volume}{\interactionMatrix}{\fugacity} .
	\]
	This gives us
	\[
		\generalPartitionFunction{\volume}{(1 - \scaleUp)\interactionMatrix}{\fugacity} - \generalPartitionFunction{\volume}{(1 + \scaleUp)\interactionMatrix}{\fugacity}
		\le \left(\eulerE^{\left(2 \scaleUp\right)^{\dimension} \sum_{i \in [\numTypes]}\fugacity[i] \vol{\volume}} - 1 \right) \generalPartitionFunction{\volume}{\interactionMatrix}{\fugacity} ,
	\]
	and concludes the proof.
\end{proof}

\section{Discretization and convergence} \label{sec:discretization}
We recall the construction of the graph $\hcPPGraph{X}$ from the introduction.
Let $(\volume, \interactionMatrix, \fugacity)$ be an instance of a hard-constraint point process with $\numTypes \in \N_{\ge 1}$ particle types.
For any finite and non-empty set of points $X \subseteq \volume$, we construct an undirected graph $\hcPPGraph{X} = (\hcPPVertices{X}, \hcPPEdges{X})$ and a function $\hcPPWeight{X}\colon \hcPPVertices{X} \to \R_{\ge 0}$ as follows:
\begin{itemize}
	\item For each point $x \in X$ and each type $i \in [\numTypes]$, we construct a vertex $\hcPPVertex{x}{i}$.
	Further, for each $i \in [\numTypes]$, we set $\hcPPVertices{X}[i] = \left\{ \hcPPVertex{x}{i} \ \middle\vert\ x \in X \right\}$, and we set $\hcPPVertices{X} = \bigcup_{i \in [\numTypes]} \hcPPVertices{X}[i]$.
	\item For each $i, j \in [\numTypes]$ and $x, y \in X$, we connect $\hcPPVertex{x}{i}, \hcPPVertex{y}{j} \in \hcPPVertices{X}$ with an edge in $\hcPPEdges{X}$ if and only if
	\[
		\neg(x=y \text{ and } i=j) \text{  and  } \dist{x}{y} < \interactionMatrix[i][j] ,
	\]
	where the condition $\neg(x=y \text{ and } i=j)$ prevents self-loops.
	\item For each $i \in [\numTypes]$ and each $x \in X$, we set $\hcPPWeight{X}[\hcPPVertex{x}{i}] = \frac{\vol{\volume}}{\size{X}} \fugacity[i]$.
	Note that for all $x, y \in X$ this means that $\hcPPWeight{X}[\hcPPVertex{x}{i}] = \hcPPWeight{X}[\hcPPVertex{y}{i}]$, and we sometimes abuse notation and write $\hcPPWeight{X}[i]$ instead.
\end{itemize}

We note that the multivariate hard-core representation $\left(\hcPPGraph{X}, \hcPPWeight{X}\right)$ does not only depend on the set of points~$X$ but also on the hard-constraint point process $(\volume, \interactionMatrix, \fugacity)$.
Since we fix the hard-constraint point process in advance, our notation is unambiguous in each context.

\subsection{Properties of the hard-core model and independent multisets}
We prove important properties of the multivariate hard-core model that we use in our main proofs.
Our first statement gives a trivial upper bound for the partition function.
\begin{lemma}
	\label{lemma:bound_hardcore}
	Let $(\graph, \weight)$ be a multivariate hard-core instance.
	It holds that
	\[
		\hcPartitionFunction{\graph}{\weight} \le \eulerE^{\sum_{\vertex \in \vertices} \weight[\vertex]}.
		\qedhere
	\]
\end{lemma}
\begin{proof}
	The claim immediately follows by observing that
	\[
		\hcPartitionFunction{\graph}{\weight}
		\le \sum_{S \subseteq \vertices} \prod_{\vertex \in S} \weight[\vertex]
		= \prod_{\vertex \in \vertices} (1 + \weight[\vertex])
		\le \prod_{\vertex \in \vertices} \eulerE^{\weight[\vertex]} .\qedhere
	\]
\end{proof}

Next, we show that the multivariate hard-core partition function is log-subadditive in $\weight$.
\begin{lemma}
	\label{lemma:log_subadditive_hardcore}
	Let $(\graph, \weight_1)$ and $(\graph, \weight_2)$ be two multivariate hard-core models on the same graph~$\graph$.
	Further, let us write $\weight_1 + \weight_2$ for the function $\weight_1 + \weight_2 \colon \vertex \mapsto \weight_1 (\vertex) + \weight_2 (\vertex)$.
	It holds that
	\[
		\hcPartitionFunction{\graph}{\weight_1 + \weight_2} \le \hcPartitionFunction{\graph}{\weight_1} \cdot \hcPartitionFunction{\graph}{\weight_2} .
		\qedhere
	\]
\end{lemma}
\begin{proof}
	First observe that
	\[
		\hcPartitionFunction{\graph}{\weight_1 + \weight_2}
		= \sum_{\independentSet \in \independentSets{\graph}} \prod_{\vertex \in \independentSet} (\weight_1 (\vertex) + \weight_2 (\vertex))
		= \sum_{\independentSet \in \independentSets{\graph}} \sum_{\independentSet_1 \subseteq \independentSet} \prod_{\vertex \in \independentSet_1} \weight_1 (\vertex) \prod_{\vertex \in \independentSet \setminus \independentSet_1} \weight_2 (\vertex) .
	\]

	For every independent set $\independentSet \in \independentSets{\graph}$, it holds that all subsets of vertices $\independentSet_1 \subseteq \independentSet$ are also independent sets $\independentSet_1 \in \independentSets{\graph}$.
	Note that this also implies that $\independentSet \setminus \independentSet_1 \in \independentSets{\graph}$.
	Further, note that for every pair of independent sets $\independentSet_1, \independentSet_2 \in \independentSets{\graph}$, there is at most one $\independentSet \in \independentSets{\graph}$ such that $\independentSet_1 \cup \independentSet_2 = \independentSet$.
	Thus, we get
	\begin{align*}
		\sum_{\independentSet \in \independentSets{\graph}} \sum_{\independentSet_1 \subseteq \independentSet} \prod_{\vertex \in \independentSet_1} \weight_1 (\vertex) \prod_{\vertex \in \independentSet \setminus \independentSet_1} \weight_2 (\vertex)
		& \le \sum_{\independentSet_1 \in \independentSets{\graph}} \sum_{\independentSet_2 \in \independentSets{\graph}} \prod_{\vertex \in \independentSet_1} \weight_1 (\vertex) \prod_{\vertex \in \independentSet_2} \weight_2 (\vertex) \\
		& = \left(\sum_{\independentSet \in \independentSets{\graph}} \prod_{\vertex \in \independentSet} \weight_1 (\vertex) \right) \left( \sum_{\independentSet \in \independentSets{\graph}} \prod_{\vertex \in \independentSet} \weight_2 (\vertex) \right) \\
		& = \hcPartitionFunction{\graph}{\weight_1} \cdot \hcPartitionFunction{\graph}{\weight_2} ,
	\end{align*}
	which proves the claim.
\end{proof}

To bound the error between the partition function of a hard-constraint point process and its hard-core representation, we need to deal with the following problem:
The continuous process is capable to place particles very close to each other, as long as the respective particle types as not subjected to interactions according to the interaction matrix $\interactions$.
On the other hand, in the hard-core representation, each vertex cannot be multiple times in the same independent set.
This means each combination of a point $x \in X$ and a type $i \in [\numTypes]$ can be represented at most once in every valid configuration.

In order  to overcome the difference between both configuration spaces, we consider a modified version of the multivariate hard-core model, which assigns weights to independent multisets, that is, multisets of vertices that have an independent set as support.
Formally, we describe such a multiset by a tuple $(\independentSet, \countingFunc)$ where $\independentSet \in \independentSets{\graph}$ and $\countingFunc\colon \independentSet \to \N_{\ge 1}$.
Given a hard-core instance $(\graph, \weight)$ where for all $\vertex \in V$, we have $\weight[\vertex] < 1$, we define the \emph{multiset hard-core partition function} as
\[
	\mhcPartitionFunction{\graph}{\weight} = \sum_{\independentSet \in \independentSets{\graph}} \sum_{\countingFunc\colon \independentSet \to \N_{\ge 1}} \prod_{\vertex \in \independentSet} \weight[\vertex]^{\countingFunc[\vertex]} .
\]
We note that we require for all $\vertex \in V$ that $\weight[\vertex] < 1$ in order for the partition function to converge.
Further, note that there is always exactly one function $\countingFunc\colon \emptyset \to \N_{\ge 1}$ and that the empty product evaluates to $1$.
This ensures that the empty set contributes exactly $1$ to the partition function.
%Analogously to the regular hard-core model, we now define the multiset Gibbs distribution to assign the probability
%\[
%	\mhcGibbs{\graph}{\weight}[\independentSet, \countingFunc] = \frac{\prod_{\vertex \in \independentSet} \weight[\vertex]^{\countingFunc[\vertex]}}{\mhcPartitionFunction{\graph}{\weight}}
%\]
%to each multiset $(\independentSet, \countingFunc)$ with $\independentSet \in \independentSets{\graph}$ and $\countingFunc: \independentSet \to \N_{\ge 1}$.
%\todo{Not sure if we need the Gibbs distribution}

We obtain the following upper and lower bound on $\mhcPartitionFunction{\graph}{\weight}$ in terms of $\hcPartitionFunction{\graph}{\weight}$.
\begin{lemma}
	\label{lemma:multiset_rewritten}
	Let $(\graph, \weight)$ be a multivariate hard-core instance, such that for all $\vertex \in \vertices$, it holds that $\weight[\vertex] \le \frac{1}{2}$.
    Further, let $\frac{\weight}{1-\weight}\colon \vertex \mapsto \frac{\weight[\vertex]}{1-\weight[\vertex]}$ and $2 \weight^2\colon \vertex \mapsto 2 \weight[\vertex]^2$.
	Then
	\[
		\hcPartitionFunction{\graph}{\weight}
		\le \mhcPartitionFunction{\graph}{\weight}
		= \hcPartitionFunction{\graph}{\frac{\weight}{1-\weight}}
		\le \hcPartitionFunction{\graph}{2 \weight^2} \cdot \hcPartitionFunction{\graph}{\weight}
		\le \eulerE^{2 \sum_{\vertex \in \vertices} \weight[\vertex]^2} \hcPartitionFunction{\graph}{\weight} .\qedhere
	\]
\end{lemma}
\begin{proof}
	First, note that the first inequality is trivial, as the multiset partition function especially includes for each independent set $\independentSet \in \independentSets{\graph}$ the multiset $(\independentSet, \countingFunc)$ where $\countingFunc$ is constantly $1$.
	The only exception to this is $\independentSet = \emptyset$, in which case the contribution to both partition functions is exactly $1$.

	To derive the identity $\mhcPartitionFunction{\graph}{\weight} = \hcPartitionFunction{\graph}{\frac{\weight}{1-\weight}}$, we start by rewriting
	\[
		\mhcPartitionFunction{\graph}{\weight}
		= \sum_{\independentSet \in \independentSets{\graph}} \sum_{\countingFunc\colon \independentSet \to \N_{\ge 1}} \prod_{\vertex \in \independentSet} \weight[\vertex]^{\countingFunc[\vertex]}
		=  \sum_{\independentSet \in \independentSets{\graph}} \prod_{\vertex \in \independentSet} \weight[\vertex] \sum_{\countingFunc\colon \independentSet \to \N} \prod_{\vertex \in \independentSet} \weight[\vertex]^{\countingFunc[\vertex]} .
	\]
	As, for all $\vertex \in \vertices$, we have $\weight[\vertex] \le \frac{1}{2} < 1$, we further get
	\[
		\sum_{\independentSet \in \independentSets{\graph}} \prod_{\vertex \in \independentSet} \weight[\vertex] \sum_{\countingFunc\colon \independentSet \to \N} \prod_{\vertex \in \independentSet} \weight[\vertex]^{\countingFunc[\vertex]}
		= \sum_{\independentSet \in \independentSets{\graph}} \prod_{\vertex \in \independentSet} \weight[\vertex] \prod_{\vertex \in \independentSet} \sum_{k \in \N} \weight[\vertex]^{k}
		= \sum_{\independentSet \in \independentSets{\graph}} \prod_{\vertex \in \independentSet} \frac{\weight[\vertex]}{1 - \weight[\vertex]} ,
	\]
	which is exactly the desired result.

	Next, for all $\vertex \in \vertices$, observe that $\weight[\vertex] \le \frac{1}{2}$ implies $\frac{\weight[\vertex]}{1-\weight[\vertex]} \le \left(1 + 2 \weight[\vertex] \right) \weight[\vertex] = 2 \weight[\vertex]^2 + \weight[\vertex]$.
	Using \Cref{lemma:log_subadditive_hardcore}, we obtain
	\[
		\hcPartitionFunction{\graph}{\frac{\weight}{1 - \weight}} \le \hcPartitionFunction{\graph}{2 \weight^2} \cdot \hcPartitionFunction{\graph}{\weight}.
	\]

	Finally, applying \Cref{lemma:bound_hardcore} to $\hcPartitionFunction{\graph}{2 \weight^2}$ yields the last inequality.
\end{proof}
We note that while we primarily make use of the upper bound $\eulerE^{2 \sum_{\vertex \in \vertices} \weight[\vertex]^2} \hcPartitionFunction{\graph}{\weight}$, we believe that the intermediate bounds are of independent interest.
Note that the choice of $\frac{1}{2}$ as an upper bound of $\fugacity$ in \Cref{lemma:multiset_rewritten} is rather arbitrary.
In fact, for any $a > 0$, the upper bound $1 - \frac{1}{a}$ yields
\[
	\mhcPartitionFunction{\graph}{\weight}
	\le \hcPartitionFunction{\graph}{a \weight^2} \cdot \hcPartitionFunction{\graph}{\weight}
	\le \eulerE^{a \sum_{\vertex \in \vertices} \weight[\vertex]^2} \hcPartitionFunction{\graph}{\weight} .
\]
For our purpose, the bound of $\frac{1}{2}$ is sufficient.

\subsection{Bounding the discretization error}
Our bound for the discretization error for any given finite non-empty point set $X \subseteq \volume$ is closely related to the existence of a $\volumeError$-$\distanceError$-allocation, defined as follows.

\edallocation*
% \begin{definition}[$\volumeError$-$\distanceError$-allocation] \label{def:allocation}
% 	Let $\volume \subset \R^{\dimension}$ be bounded and measurable and assume $\vol{\volume} > 0$.
% 	For finite $X \subset \volume$, $\volumeError \in [0, 1)$ and $\distanceError \in \R_{>0}$ we call a function $\allocation: \volume \to X$ a \emph{$\volumeError$-$\distanceError$-allocation} for $X$ if and only if for all $x \in X$ it holds that
% 	\begin{enumerate}[(1)]
% 		\item $\invAllocation[x]$ is measurable and
% 		\[
% 			\left(1 - \volumeError\right) \frac{\vol{\volume}}{\size{X}} \le \vol{\invAllocation[x]} \le \left(1 + \volumeError\right) \frac{\vol{\volume}}{\size{X}}
% 		\]
% 		\item for all $y \in \invAllocation[x]$ it holds that
% 		\[
% 			\dist{x}{y} \le \distanceError . \qedhere
% 		\]
% 	\end{enumerate}
% \end{definition}

For $X \subset \volume$ and a corresponding allocation $\allocation$, we relate the hard-core representation $(\hcPPGraph{X}, \hcPPWeight{X})$ to the continuous model $(\volume, \interactionMatrix, \fugacity)$.
To this end, we consider two intermediate continuous models.

For $\numParticles \in \N$ and a tuple $\mathbold{x} = (x_i)_{i \in [\numParticles]} \in \volume^{\numParticles}$ we abuse notation and write $\allocation[\mathbold{x}]$ for the tuple $(\allocation[x_i])_{i \in [\numParticles]} \in X^{\numParticles}$.
The first intermediate partition function we consider is
\[
	\allocationPartitionFunction{\volume}{\interactionMatrix}{\fugacity}{\allocation} = \sum_{\numParticles \in \N} \sum_{\typeOf\colon [\numParticles] \to [\numTypes]} \int_{\volume^{\numParticles}} \generalWeight{\interactionMatrix}{\fugacity}[\allocation[\mathbold{x}], \typeOf] \, \text{d} \lebesgue{\dimension \times \numParticles}.
\] 
Intuitively, this is the partition function that results from mapping each point in $\volume$ to the point in $X$ that it is allocated to, before checking if a configuration is valid.

For the second intermediate model, let again $\numParticles \in \N$.
For all $\mathbold{x} = (x_i)_{i \in [\numParticles]} \in \volume^{\numParticles}$ and $\typeOf\colon [\numParticles] \to [\numTypes]$, let
\begin{align*}
    \hcPPValidContinuous{\typeOf}{\hcPPGraph{X}}{\allocation}[\mathbold{x}] &= \prod_{\substack{i, j \in [\numParticles]\colon\\ i < j}} \indicator{\left(\hcPPVertex{\allocation[x_i]}{\typeOf[i]}, \hcPPVertex{\allocation[x_j]}{\typeOf[j]}\right) \notin \hcPPEdges{X} \text{ and } \hcPPVertex{\allocation[x_i]}{\typeOf[i]} \neq \hcPPVertex{\allocation[x_j]}{\typeOf[j]}} \textrm{, and let}\\
	\hcPPWeightContinuous{\hcPPGraph{X}}{\hcPPWeight{X}}{\allocation}[\mathbold{x}, \typeOf]
	&= \frac{1}{\numParticles!} \left(\prod_{i \in [\numParticles]}  \frac{\hcPPWeight{X}[\hcPPVertex{x_i}{\typeOf[i]}]}{\vol{\invAllocation[\allocation[x_i]]}}\right) \hcPPValidContinuous{\typeOf}{\hcPPGraph{X}}{\allocation}[\mathbold{x}] .\\
%	&= \frac{1}{\numParticles!} \left(\prod_{i \in [\numParticles]}  \fugacity[\typeOf[i]]\frac{\vol{\volume}/\size{X}}{\vol{\invAllocation[\allocation[x_i]]}}\right) \hcPPValidContinuous{\typeOf}{\hcPPGraph{X}}{\allocation}
\end{align*}
Further, define the partition function and Gibbs density corresponding to $\hcPPWeightContinuous{\hcPPGraph{X}}{\hcPPWeight{X}}{\allocation}$, respectively, as
\begin{align*}
	\hcPPPartitionFunctionContinuous{\volume}{\hcPPGraph{X}}{\hcPPWeight{X}}{\allocation}
	&= \sum_{\numParticles \in \N} \sum_{\typeOf\colon [\numParticles] \to [\numTypes]} \int_{\volume^{\numParticles}} \hcPPWeightContinuous{\hcPPGraph{X}}{\hcPPWeight{X}}{\allocation}[\mathbold{x}, \typeOf] \, \text{d} \lebesgue{\dimension \times \numParticles} \textrm{ and}\\
    \hcPPGibbsContinuous{\volume}{\hcPPGraph{X}}{\hcPPWeight{X}}{\allocation}[\mathbold{x}, \typeOf] &= \frac{\hcPPWeightContinuous{\hcPPGraph{X}}{\hcPPWeight{X}}{\allocation}[\mathbold{x}, \typeOf]}{\hcPPPartitionFunctionContinuous{\volume}{\hcPPGraph{X}}{\hcPPWeight{X}}{\allocation}} .
\end{align*}
In \Cref{sec:sampling}, we give a detailed probabilistic interpretation of $\hcPPGibbsContinuous{\volume}{\hcPPGraph{X}}{\hcPPWeight{X}}{\allocation}$.
For this section and the approximation results in \Cref{sec:algo}, we are mostly interested in $\hcPPPartitionFunctionContinuous{\volume}{\hcPPGraph{X}}{\hcPPWeight{X}}{\allocation}$.

The following bound is central to almost all results in this paper.
\begin{lemma}
	\label{lemma:discretization_error}
	Let $(\volume, \interactionMatrix, \fugacity)$ be a hard-constraint point process with $\numTypes \in \N_{\ge 1}$ particle types, and assume $\volume \subset \R^{\dimension}$ is star-convex.
	Further, set $\interactionMatrix_{\min} = \inf_{i,j \in [\numTypes]} \{\interactionMatrix[i][j] \mid \interactionMatrix[i][j] > 0\}$, and set $\fugacity_{\max} = \max_{i \in [\numTypes]} \fugacity[i]$.
	Let $X \subseteq \volume$ with $4 \fugacity_{\max} \vol{\volume} \le \size{X} < \infty$, and let $(\hcPPGraph{X}, \hcPPWeight{X})$ be the hard-core representation of $(\volume, \interactionMatrix, \fugacity)$ based on $X$.
	Last, let $\volumeError \in \left[0, \frac{1}{2}\right]$ and $\distanceError \in \left[0, \frac{\interactionMatrix_{\min}}{2}\right]$, and assume that $\allocation$ is a $\volumeError$-$\distanceError$-allocation for $X$.
	Then
	\begin{align*}
		\sum_{\numParticles \in \N} \sum_{\typeOf\colon [\numParticles] \to [\numTypes]} \int_{\volume^{\numParticles}}
		& \absolute{\generalWeight{\interactionMatrix}{\fugacity}[\mathbold{x}, \typeOf] - \hcPPWeightContinuous{\hcPPGraph{X}}{\hcPPWeight{X}}{\allocation}[\mathbold{x}, \typeOf]} \, \text{d} \lebesgue{\dimension \times \numParticles} \\
		&\le \left(\eulerE^{\frac{8}{\size{X}} \sum_{i \in [\numTypes]} \fugacity[i]^2 \vol{\volume}^2} \eulerE^{\left( 2 \volumeError + \left(\frac{4 \distanceError }{\interactionMatrix_{\min}} \right)^{\dimension} \right) \sum_{i \in [\numTypes]} \fugacity[i] \vol{\volume}} - 1\right) \generalPartitionFunction{\volume}{\interactionMatrix}{\fugacity}.
		\qedhere
	\end{align*}
\end{lemma}

The proof of this statement is divided into two parts, which we state separately.
The first part essentially bounds the error that is introduced into the partition function of the point process if we map each point $x \in \volume$ to $\allocation[x]$, the point it is allocated to, before checking if a configuration is valid or not. 
This essentially corresponds to relating the original point process to the first intermediate model.
\begin{lemma}
	\label{lemma:discretization_error:1}
	Let $(\volume, \interactionMatrix, \fugacity)$ be a hard-constraint point process with $\numTypes \in \N_{\ge 1}$ particle types, and assume $\volume \subset \R^{\dimension}$ is star-convex.
	Further, set $\interactionMatrix_{\min} = \inf_{i,j \in [\numTypes]} \{\interactionMatrix[i][j] \mid \interactionMatrix[i][j] > 0\}$, and set $\fugacity_{\max} = \max_{i \in [\numTypes]} \fugacity[i]$.
	Let $X \subseteq \volume$ be finite, and let $(\hcPPGraph{X}, \hcPPWeight{X})$ be the hard-core representation of $(\volume, \interactionMatrix, \fugacity)$ based on $X$.
	Last, let $\volumeError \in \left[0, 1\right]$ and $\distanceError \in \left[0, \frac{\interactionMatrix_{\min}}{2}\right]$, and assume that $\allocation$ is a $\volumeError$-$\distanceError$-allocation for $X$.
	For all $\mathbold{x} = (x_i)_{i \in [\numParticles]} \in \volume^{\numParticles}$ write $\allocation[\mathbold{x}]$ for $\left(\allocation[x_i]\right)_{i \in [\numParticles]} \in X^{\numParticles}$.
	Then
	\begin{align*}
		\sum_{\numParticles \in \N} \sum_{\typeOf\colon [\numParticles] \to [\numTypes]} \int_{\volume^{\numParticles}} \absolute{\generalWeight{\interactionMatrix}{\fugacity}[\mathbold{x}, \typeOf] - \generalWeight{\interactionMatrix}{\fugacity}[\allocation[\mathbold{x}], \typeOf]} \, \text{d} \lebesgue{\dimension \times \numParticles}
		&\le \left(\eulerE^{\left(\frac{4 \distanceError}{\interactionMatrix_{\min}}\right)^{\dimension} \sum_{i \in [\numTypes]} \fugacity[i] \vol{\volume}} - 1\right)\generalPartitionFunction{\volume}{\interactionMatrix}{\fugacity} .
		\qedhere
	\end{align*}
\end{lemma} 

\begin{proof}
	The main idea of the proof is to show that
	\[
		\sum_{\numParticles \in \N} \sum_{\typeOf\colon [\numParticles] \to [\numTypes]} \int_{\volume^{\numParticles}} \absolute{\generalWeight{\interactionMatrix}{\fugacity}[\mathbold{x}, \typeOf] - \generalWeight{\interactionMatrix}{\fugacity}[\allocation[\mathbold{x}], \typeOf]} \, \text{d} \lebesgue{\dimension \times \numParticles}
		\le
		\generalPartitionFunction{\volume}{\interactionMatrix_{-}}{\fugacity} - \generalPartitionFunction{\volume}{\interactionMatrix_{+}}{\fugacity} ,
	\]
	where $\interactionMatrix_{-}$ and $\interactionMatrix_{+}$ are slightly scaled versions of $\interactionMatrix$, and then applying \Cref{lemma:scaled_difference}.
	
	First, note that for $\numParticles \le 1$, it holds for all $\mathbold{x} \in \volume^{\numParticles}$ and $\typeOf\colon [\numParticles] \to [\numTypes]$ that $\generalWeight{\interactionMatrix}{\fugacity}[\mathbold{x}, \typeOf] = \generalWeight{\interactionMatrix}{\fugacity}[\allocation[\mathbold{x}], \typeOf]$.
	Fix $\numParticles \ge 2$ as well as $\typeOf\colon [\numParticles] \to [\numTypes]$.
	Let $N_{\typeOf} \subseteq \volume^{\numParticles}$ be the set of all points $\mathbold{x} = (x_i)_{i \in [\numParticles]} \in \volume^{\numParticles}$ such that $\valid{\typeOf}{\interactionMatrix}[\allocation[\mathbold{x}]] \neq \valid{\typeOf}{\interactionMatrix}[\mathbold{x}]$.
	Observe that
	\[
	\absolute{\generalWeight{\interactionMatrix}{\fugacity}[\mathbold{x}, \typeOf] - \generalWeight{\interactionMatrix}{\fugacity}[\allocation[\mathbold{x}], \typeOf]}
	\le \frac{1}{\numParticles!} \left(\prod_{i \in [\numParticles]} \fugacity[\typeOf[i]]\right) \indicator{\mathbold{x} \in N_{\typeOf}} .
	\]
	
	We proceed by characterizing a superset of $N_{\typeOf}$ that is easier to analyze.
	Observe that, as $\allocation$ is a $\volumeError$-$\distanceError$-allocation, it holds for all $x \in \volume$ that $\dist{x}{\allocation[x]} \le \distanceError$.
	Let $\interactionMatrix_{-} = \left(1 - \frac{2 \distanceError}{\interactionMatrix_{\min}}\right) \interactionMatrix$ and $\interactionMatrix_{+} = \left(1 + \frac{2 \distanceError}{\interactionMatrix_{\min}}\right) \interactionMatrix$.
	Thus, for all $i, j \in [\numTypes]$ with $\interactionMatrix[i][j] > 0$, it holds that $\interactionMatrix_{-}(i, j) \le \interactionMatrix[i][j] - 2 \distanceError$ and $\interactionMatrix_{+}(i, j) \ge \interactionMatrix[i][j] + 2 \distanceError$.
	
	Let $\mathbold{x} = (x_i)_{i \in [\numParticles]} \in \volume^{\numParticles}$ such that $\valid{\typeOf}{\interactionMatrix_{-}}[\mathbold{x}] = 0$.
	Note that this implies that there are $i, j \in [\numParticles]$ with $i \neq j$ such that $\interactionMatrix[\typeOf[i]][\typeOf[j]] > 0$ and
	\[
	\dist{x_i}{x_j}
	< \interactionMatrix_{-}(\typeOf[i], \typeOf[j])
	\le \interactionMatrix[\typeOf[i]][\typeOf[j]] - 2 \distanceError.
	\]
	We conclude that $\dist{x_i}{x_j} < \interactionMatrix[\typeOf[i]][\typeOf[j]]$ and by triangle inequality $\dist{\allocation[x_i]}{\allocation[x_j]} < \interactionMatrix[\typeOf[i]][\typeOf[j]]$.
	Thus, we have $\valid{\typeOf}{\interactionMatrix}[\mathbold{x}] = \valid{\typeOf}{\interactionMatrix}[\allocation[\mathbold{x}]] = 0$ and consequently $\mathbold{x} \notin N_{\typeOf}$.
	
	Next, let $\mathbold{x} = (x_i)_{i \in [\numParticles]} \in \volume^{\numParticles}$ with $\valid{\typeOf}{\interactionMatrix_{+}}[\mathbold{x}] = 1$.
	Note that this implies for all $i, j \in [\numParticles]$ with $i \neq j$ and $\interactionMatrix[\typeOf[i]][\typeOf[j]] > 0$ that
	\[
	\dist{x_i}{x_j}
	\ge \interactionMatrix_{+}(\typeOf[i], \typeOf[j])
	\ge \interactionMatrix[\typeOf[i]][\typeOf[j]] + 2 \distanceError.
	\]
	We conclude that $\dist{x_i}{x_j} \ge \interactionMatrix[\typeOf[i]][\typeOf[j]]$ and by triangle inequality $\dist{\allocation[x_i]}{\allocation[x_j]} \ge \interactionMatrix[\typeOf[i]][\typeOf[j]]$.
	Thus, we have $\valid{\typeOf}{\interactionMatrix}[\mathbold{x}] = \valid{\typeOf}{\interactionMatrix}[\allocation[\mathbold{x}]] = 1$ and consequently $\mathbold{x} \notin N_{\typeOf}$.
	
	Finally, observe that, because for all $i, j \in [\numTypes]$, it holds that $\interactionMatrix_{-}(i,j) \le \interactionMatrix_{+}(i,j)$, we have that $\valid{\typeOf}{\interactionMatrix_{-}}[\mathbold{x}] = 0$ implies $\valid{\typeOf}{\interactionMatrix_{+}}[\mathbold{x}] = 0$, and $\valid{\typeOf}{\interactionMatrix_{+}}[\mathbold{x}] = 1$ implies $\valid{\typeOf}{\interactionMatrix_{-}}[\mathbold{x}] = 1$.
	Thus, we have
	\[
	\indicator{\mathbold{x} \in N_{\typeOf}}
	\le \valid{\typeOf}{\interactionMatrix_{-}}[\mathbold{x}] - \valid{\typeOf}{\interactionMatrix_{+}}[\mathbold{x}] .
	\]
	Substituting this into \cref{lemma:discretization_error:eq_part1} yields
	\begin{align*}
		\sum_{\numParticles \in \N} \sum_{\typeOf\colon [\numParticles] \to [\numTypes]}
		&\int_{\volume^{\numParticles}} \absolute{\generalWeight{\interactionMatrix}{\fugacity}[\mathbold{x}, \typeOf] - \generalWeight{\interactionMatrix}{\fugacity}[\allocation[\mathbold{x}], \typeOf]} \, \text{d} \lebesgue{\dimension \times \numParticles} \\
		&\le \sum_{\numParticles \in \N_{\ge 2}} \frac{1}{\numParticles !} \sum_{\typeOf\colon [\numParticles] \to [\numTypes]} \left(\prod_{i \in [\numParticles]} \fugacity[\typeOf[i]] \right)
		\int_{\volume^{\numParticles}} \indicator{\mathbold{x} \in N_{\typeOf}}\, \text{d} \lebesgue{\dimension \times \numParticles} \\
		&\le \sum_{\numParticles \in \N_{\ge 2}} \frac{1}{\numParticles !} \sum_{\typeOf\colon [\numParticles] \to [\numTypes]} \left(\prod_{i \in [\numParticles]} \fugacity[\typeOf[i]]\right) \int_{\volume^{\numParticles}} \valid{\typeOf}{\interactionMatrix_{-}}[\mathbold{x}] \, \text{d} \lebesgue{\dimension \times \numParticles} \\
		&\hspace{2em} - \sum_{\numParticles \in \N_{\ge 2}} \frac{1}{\numParticles !} \sum_{\typeOf\colon [\numParticles] \to [\numTypes]} \left(\prod_{i \in [\numParticles]} \fugacity[\typeOf[i]]\right) \int_{\volume^{\numParticles}} \valid{\typeOf}{\interactionMatrix_{+}}[\mathbold{x}] \, \text{d} \lebesgue{\dimension \times \numParticles} \\
		&\le \generalPartitionFunction{\volume}{\interactionMatrix_{-}}{\fugacity} - \generalPartitionFunction{\volume}{\interactionMatrix_{+}}{\fugacity} .
	\end{align*}
	
	Recalling the definitions of $\interactionMatrix_{-}$ and $\interactionMatrix_{+}$ as well as that~$\volume$ is star-convex, we apply \Cref{lemma:scaled_difference} with $\scaleUp = \frac{2 \distanceError}{\interactionMatrix_{\min}} \in [0, 1]$ and obtain
	\begin{align}
		\label[ineq]{lemma:discretization_error:eq_bound1}
		\sum_{\numParticles \in \N} \sum_{\typeOf\colon [\numParticles] \to [\numTypes]}
		&\int_{\volume^{\numParticles}} \absolute{\generalWeight{\interactionMatrix}{\fugacity}[\mathbold{x}, \typeOf] - \generalWeight{\interactionMatrix}{\fugacity}[\allocation[\mathbold{x}], \typeOf]} \, \text{d} \lebesgue{\dimension \times \numParticles}
		\le \left(\eulerE^{\left(\frac{4 \distanceError}{\interactionMatrix_{\min}}\right)^{\dimension} \sum_{i \in [\numTypes]} \fugacity[i] \vol{\volume}} - 1\right)\generalPartitionFunction{\volume}{\interactionMatrix}{\fugacity} ,
	\end{align}
	which concludes the proof.
\end{proof}

The second lemma relates the partition function that results from mapping each point $x \in \volume$ to $\allocation[x]$ with $\hcPPPartitionFunctionContinuous{\volume}{\hcPPGraph{X}}{\hcPPWeight{X}}{\allocation}$.
This is, relating our first continuous intermediate model with the second one.

\begin{lemma}
	\label{lemma:discretization_error:2}
	Let $(\volume, \interactionMatrix, \fugacity)$ be a hard-constraint point process with $\numTypes \in \N_{\ge 1}$ particle types.
	Set $\fugacity_{\max} = \max_{i \in [\numTypes]} \fugacity[i]$.
	Let $X \subseteq \volume$ with $4 \fugacity_{\max} \vol{\volume} \le \size{X} < \infty$, and let $(\hcPPGraph{X}, \hcPPWeight{X})$ be the hard-core representation of $(\volume, \interactionMatrix, \fugacity)$ based on $X$.
	Last, let $\volumeError \in \left[0, \frac{1}{2}\right]$ and $\distanceError \in \left[0, 1\right]$, and assume that $\allocation$ is a $\volumeError$-$\distanceError$-allocation for $X$.
	Then
	\begin{align*}
		\sum_{\numParticles \in \N}
		&\sum_{\typeOf\colon [\numParticles] \to [\numTypes]} \int_{\volume^{\numParticles}}
		\absolute{\generalWeight{\interactionMatrix}{\fugacity}[\allocation[\mathbold{x}], \typeOf] - \hcPPWeightContinuous{\hcPPGraph{X}}{\hcPPWeight{X}}{\allocation}[\mathbold{x}, \typeOf]} \, \text{d} \lebesgue{\dimension \times \numParticles} \\
		&\le \left(\eulerE^{\frac{8}{\size{X}} \sum_{i \in [\numTypes]} \fugacity[i]^2 \vol{\volume}^2} \eulerE^{2 \volumeError \sum_{i \in [\numTypes]} \fugacity[i] \vol{\volume}} - 1\right)		\allocationPartitionFunction{\volume}{\interactionMatrix}{\fugacity}{\allocation}.
		\qedhere
	\end{align*}
\end{lemma}

\begin{proof}
	We prove the statement by arguing that
	\begin{align*}
		\sum_{\numParticles \in \N} \sum_{\typeOf\colon [\numParticles] \to [\numTypes]} \int_{\volume^{\numParticles}}
		&\absolute{\generalWeight{\interactionMatrix}{\fugacity}[\allocation[\mathbold{x}], \typeOf] - \hcPPWeightContinuous{\hcPPGraph{X}}{\hcPPWeight{X}}{\allocation}[\mathbold{x}, \typeOf]} \, \text{d} \lebesgue{\dimension \times \numParticles} \\
		&\le \mhcPartitionFunction{\hcPPGraph{X}}{(1+\volumeError) \hcPPWeight{X}} - \hcPartitionFunction{\hcPPGraph{X}}{(1-\volumeError) \hcPPWeight{X}}
	\end{align*}
	and applying \Cref{lemma:bound_hardcore,lemma:log_subadditive_hardcore,lemma:multiset_rewritten}.
	
	For $\mathbold{x} = (x_i)_{i \in [\numParticles]} \in \volume^{\numParticles}$ and $\typeOf\colon [\numParticles] \to [\numTypes]$, let
	\[
	\mhcPPValidContinuous{\typeOf}{\hcPPGraph{X}}{\allocation}[\mathbold{x}] = \prod_{\substack{i, j \in [\numParticles]\colon\\ i < j}}\indicator{\left(\hcPPVertex{\allocation[x_i]}{\typeOf[i]}, \hcPPVertex{\allocation[x_j]}{\typeOf[j]}\right) \notin \hcPPEdges{X} \text{ or } \hcPPVertex{\allocation[x_i]}{\typeOf[i]} = \hcPPVertex{\allocation[x_j]}{\typeOf[j]}} .
	\]
	Note that for all $\mathbold{x} \in \volume^{\numParticles}$ and all type assignments $\typeOf\colon [\numParticles] \to [\numTypes]$, it holds that
	\begin{align*}
		\hcPPValidContinuous{\typeOf}{\hcPPGraph{X}}{\allocation}[\mathbold{x}]
		\le \valid{\typeOf}{\interactionMatrix}[\allocation[\mathbold{x}]]
		\le \mhcPPValidContinuous{\typeOf}{\hcPPGraph{X}}{\allocation}[\mathbold{x}] .
	\end{align*}
	Further, as $\allocation$ is a $\volumeError$-$\distanceError$-allocation for $X$, it holds for all $x \in X$ that
	\[
	\left(1 - \volumeError\right) \frac{\vol{\volume}}{\size{X}} \le \vol{\invAllocation[x]} \le \left(1 + \volumeError\right) \frac{\vol{\volume}}{\size{X}} .
	\]
	For all $x \in X$ and $i \in [\numTypes]$, because $\hcPPWeight{X}[\hcPPVertex{x}{i}] = \fugacity[i] \frac{\vol{\volume}}{\size{X}}$, we have
	\[
	\left(1 - \volumeError\right) \frac{\hcPPWeight{X}[\hcPPVertex{x}{i}]}{\vol{\invAllocation[x]}}
	\le \fugacity[i]
	\le \left(1 - \volumeError\right) \frac{\hcPPWeight{X}[\hcPPVertex{x}{i}]}{\vol{\invAllocation[x]}} .
	\]
	This implies for all $\mathbold{x} \in \volume^{\numParticles}$ and $\typeOf\colon [\numParticles] \to [\numTypes]$ that
	\begin{align*}
		\frac{1}{\numParticles!} \left(\prod_{i \in [\numParticles]} \left(1 - \volumeError\right) \frac{\hcPPWeight{X}[\hcPPVertex{\allocation[x_i]}{\typeOf[i]}]}{\vol{\invAllocation[\allocation[x_i]]}}\right) \hcPPValidContinuous{\typeOf}{\hcPPGraph{X}}{\allocation}[\mathbold{x}]
		&\le \generalWeight{\interactionMatrix}{\fugacity}[\allocation[\mathbold{x}], \typeOf] \\
		&\le \frac{1}{\numParticles!} \left(\prod_{i \in [\numParticles]} \left(1 + \volumeError\right) \frac{\hcPPWeight{X}[\hcPPVertex{\allocation[x_i]}{\typeOf[i]}]}{\vol{\invAllocation[\allocation[x_i]]}}\right) \mhcPPValidContinuous{\typeOf}{\hcPPGraph{X}}{\allocation}[\mathbold{x}]
	\end{align*}
	and
	\begin{align}
		\notag
		\frac{1}{\numParticles!} \left(\prod_{i \in [\numParticles]} \left(1 - \volumeError\right) \frac{\hcPPWeight{X}[\hcPPVertex{\allocation[x_i]}{\typeOf[i]}]}{\vol{\invAllocation[\allocation[x_i]]}}\right) \hcPPValidContinuous{\typeOf}{\hcPPGraph{X}}{\allocation}[\mathbold{x}]
		&\le \hcPPWeightContinuous{\hcPPGraph{X}}{\hcPPWeight{X}}{\allocation}[\mathbold{x}, \typeOf] \\
		&\le \frac{1}{\numParticles!} \left(\prod_{i \in [\numParticles]} \left(1 + \volumeError\right) \frac{\hcPPWeight{X}[\hcPPVertex{\allocation[x_i]}{\typeOf[i]}]}{\vol{\invAllocation[\allocation[x_i]]}}\right) \mhcPPValidContinuous{\typeOf}{\hcPPGraph{X}}{\allocation}[\mathbold{x}] .
		\label[ineq]{lemma:discretization_error:eq_bound_weight}
	\end{align}
	Substituting this into \cref{lemma:discretization_error:eq_part2} yields
	\begin{align}
		\notag
		\sum_{\numParticles \in \N}
		&\sum_{\typeOf\colon [\numParticles] \to [\numTypes]} \int_{\volume^{\numParticles}} \absolute{\generalWeight{\interactionMatrix}{\fugacity}[\allocation[\mathbold{x}], \typeOf] - \hcPPWeightContinuous{\hcPPGraph{X}}{\hcPPWeight{X}}{\allocation}[\mathbold{x}, \typeOf]} \, \text{d} \lebesgue{\dimension \times \numParticles} \\
		\notag
		&\le \sum_{\numParticles \in \N} \sum_{\typeOf\colon [\numParticles] \to [\numTypes]} \int_{\volume^{\numParticles}} \frac{1}{\numParticles!} \left(\prod_{i \in [\numParticles]} \left(1 + \volumeError\right) \frac{\hcPPWeight{X}[\hcPPVertex{\allocation[x_i]}{\typeOf[i]}]}{\vol{\invAllocation[\allocation[x_i]]}}\right) \mhcPPValidContinuous{\typeOf}{\hcPPGraph{X}}{\allocation}[\mathbold{x}]\, \text{d} \lebesgue{\dimension \times \numParticles} \\
		\label[ineq]{lemma:discretization_error:eq_part2_first_bounds}
		& - \sum_{\numParticles \in \N} \sum_{\typeOf\colon [\numParticles] \to [\numTypes]} \int_{\volume^{\numParticles}} \frac{1}{\numParticles!} \left(\prod_{i \in [\numParticles]} \left(1 - \volumeError\right) \frac{\hcPPWeight{X}[\hcPPVertex{\allocation[x_i]}{\typeOf[i]}]}{\vol{\invAllocation[\allocation[x_i]]}}\right) \hcPPValidContinuous{\typeOf}{\hcPPGraph{X}}{\allocation}[\mathbold{x}]\, \text{d} \lebesgue{\dimension \times \numParticles} .
	\end{align}
	
	Next, note that due to $\volumeError \le 1$ and $4 \fugacity_{\text{max}} \vol{\volume} \le \size{X}$, it holds for all $\hcPPVertex{x}{i} \in \hcPPVertices{X}$ that
	\begin{align}
		\label{lemma:discretization_error:eq_fugacitybound}
		\left(1 + \volumeError\right)\hcPPWeight{X}[\hcPPVertex{x}{i}]
		\le 2 \frac{\vol{\volume}}{\size{X}} \fugacity[i]
		\le \frac{1}{2} .
	\end{align}
	Thus, $\mhcPartitionFunction{\hcPPGraph{X}}{\left(1 + \volumeError\right)\hcPPWeight{X}}$ is well defined.
	We proceed by rewriting both functions, $\hcPartitionFunction{\hcPPGraph{X}}{\left(1 - \volumeError\right)\hcPPWeight{X}}$ and $\mhcPartitionFunction{\hcPPGraph{X}}{\left(1 + \volumeError\right)\hcPPWeight{X}}$, by expressing each in terms of tuples of vertices.
	For $\hcPartitionFunction{\hcPPGraph{X}}{\left(1 - \volumeError\right)\hcPPWeight{X}}$, this yields
	\begin{align}
		\notag
		\hcPartitionFunction{\hcPPGraph{X}}{\left(1 - \volumeError\right)\hcPPWeight{X}}
		&= \sum_{\numParticles \in \N} \sum_{\typeOf\colon [\numParticles] \to [\numTypes]} \sum_{\mathbold{x} \in X^{\numParticles}} \frac{1}{\numParticles!} \left(\prod_{i \in [\numParticles]} \left(1 - \volumeError\right) \hcPPWeight{X}[\hcPPVertex{x_i}{\typeOf[i]}]\right) \hcPPValidContinuous{\typeOf}{\hcPPGraph{X}}{\allocation}[\mathbold{x}] \\
		&= \sum_{\numParticles \in \N} \sum_{\typeOf\colon [\numParticles] \to [\numTypes]} \int_{\volume^{\numParticles}} \frac{1}{\numParticles!} \left(\prod_{i \in [\numParticles]} \left(1 - \volumeError\right) \frac{\hcPPWeight{X}[\hcPPVertex{\allocation[x_i]}{\typeOf[i]}]}{\vol{\invAllocation[\allocation[x_i]]}}\right) \hcPPValidContinuous{\typeOf}{\hcPPGraph{X}}{\allocation}[\mathbold{x}]\, \text{d} \lebesgue{\dimension \times \numParticles}  .
		\label{lemma:discretization_error:eq_rewrite_hc}
	\end{align}
	Similarly, for $\mhcPartitionFunction{\hcPPGraph{X}}{\left(1 + \volumeError\right)\hcPPWeight{X}}$, we get
	\begin{align}
		\notag
		\mhcPartitionFunction{\hcPPGraph{X}}{\left(1 + \volumeError\right)\hcPPWeight{X}}
		&= \sum_{\numParticles \in \N} \sum_{\typeOf\colon [\numParticles] \to [\numTypes]} \sum_{\mathbold{x} \in X^{\numParticles}} \frac{1}{\numParticles!} \left(\prod_{i \in [\numParticles]} \left(1 + \volumeError\right) \hcPPWeight{X}[\hcPPVertex{x_i}{\typeOf[i]}]\right) \mhcPPValidContinuous{\typeOf}{\hcPPGraph{X}}{\allocation}[\mathbold{x}] \\
		&= \sum_{\numParticles \in \N} \sum_{\typeOf\colon [\numParticles] \to [\numTypes]} \int_{\volume^{\numParticles}} \frac{1}{\numParticles!} \left(\prod_{i \in [\numParticles]} \left(1 + \volumeError\right) \frac{\hcPPWeight{X}[\hcPPVertex{\allocation[x_i]}{\typeOf[i]}]}{\vol{\invAllocation[\allocation[x_i]]}}\right) \mhcPPValidContinuous{\typeOf}{\hcPPGraph{X}}{\allocation}[\mathbold{x}]\, \text{d} \lebesgue{\dimension \times \numParticles}  .
		\label{lemma:discretization_error:eq_rewrite_mhc}
	\end{align}
	Note that the multiset version always allows arbitrarily many copies of the same vertex, independent of $\interactionMatrix$, which is captured by the indicator function $\mhcPPValidContinuous{\typeOf}{\hcPPGraph{X}}{\allocation}$.
	
	Substituting \cref{lemma:discretization_error:eq_rewrite_hc,lemma:discretization_error:eq_rewrite_mhc} into \cref{lemma:discretization_error:eq_part2_first_bounds} yields
	\begin{align*}
		\sum_{\numParticles \in \N} \sum_{\typeOf\colon [\numParticles] \to [\numTypes]} \int_{\volume^{\numParticles}}
		& \absolute{\generalWeight{\interactionMatrix}{\fugacity}[\allocation[\mathbold{x}], \typeOf] - \hcPPWeightContinuous{\hcPPGraph{X}}{\hcPPWeight{X}}{\allocation}[\mathbold{x}, \typeOf]} \, \text{d} \lebesgue{\dimension \times \numParticles} \\
		&\le \mhcPartitionFunction{\hcPPGraph{X}}{\left(1 + \volumeError\right)\hcPPWeight{X}} - \hcPartitionFunction{\hcPPGraph{X}}{\left(1 - \volumeError\right)\hcPPWeight{X}} .
	\end{align*}
	We proceed by bounding both terms of the difference.
	To this end, we start by upper-bounding the partition function $\mhcPartitionFunction{\hcPPGraph{X}}{\left(1 + \volumeError\right)\hcPPWeight{X}}$ in terms of $\hcPartitionFunction{\hcPPGraph{X}}{\left(1 + \volumeError\right)\hcPPWeight{X}}$ using \Cref{lemma:multiset_rewritten} (note that we can do so by \cref{lemma:discretization_error:eq_fugacitybound}).
	We obtain
	\begin{align*}
		\mhcPartitionFunction{\hcPPGraph{X}}{\left(1 + \volumeError\right)\hcPPWeight{X}}
		&\le \eulerE^{2 \sum_{i \in [\numTypes]} \sum_{x \in X} \left(1 + \volumeError\right)^2 \hcPPWeight{X}[\hcPPVertex{x}{i}]^2} \hcPartitionFunction{\hcPPGraph{X}}{\left(1 + \volumeError\right)\hcPPWeight{X}} \\
		&\le \eulerE^{2 \left(1 + \volumeError\right)^2 \sum_{i \in [\numTypes]} \fugacity[i]^2 \sum_{x \in X} \frac{\vol{\volume}^2}{\size{X}^{2}}} \hcPartitionFunction{\hcPPGraph{X}}{\left(1 + \volumeError\right)\hcPPWeight{X}} \\
		&\le \eulerE^{\frac{8}{\size{X}} \sum_{i \in [\numTypes]} \fugacity[i]^2 \vol{\volume}^2} \hcPartitionFunction{\hcPPGraph{X}}{\left(1 + \volumeError\right)\hcPPWeight{X}} .
	\end{align*}
	Further, by \Cref{lemma:log_subadditive_hardcore,lemma:bound_hardcore} we obtain
	\begin{align*}
		\hcPartitionFunction{\hcPPGraph{X}}{\left(1 + \volumeError\right)\hcPPWeight{X}}
		&\le \hcPartitionFunction{\hcPPGraph{X}}{2 \volumeError \hcPPWeight{X}} \hcPartitionFunction{\hcPPGraph{X}}{\left(1 - \volumeError\right)\hcPPWeight{X}} \\
		&\le \eulerE^{\sum_{i \in [\numTypes]} \sum_{x \in X} 2 \volumeError \hcPPWeight{X}[\hcPPVertex{x}{i}]} \hcPartitionFunction{\hcPPGraph{X}}{\left(1 - \volumeError\right)\hcPPWeight{X}} \\
		&=   \eulerE^{2 \volumeError \sum_{i \in [\numTypes]} \fugacity[i] \vol{\volume}} \hcPartitionFunction{\hcPPGraph{X}}{\left(1 - \volumeError\right)\hcPPWeight{X}} .
	\end{align*}
	Thus, we have
	\begin{align*}
		\notag
		&\mhcPartitionFunction{\hcPPGraph{X}}{(1+\volumeError)\hcPPWeight{X}} - \hcPartitionFunction{\hcPPGraph{X}}{(1-\volumeError)\hcPPWeight{X}} \\
		&\hspace{1em} \le \left( \eulerE^{\frac{8}{\size{X}} \sum_{i \in [\numTypes]} \fugacity[i]^2 \vol{\volume}^2} \eulerE^{2 \volumeError \sum_{i \in [\numTypes]} \fugacity[i] \vol{\volume}} - 1 \right) \hcPartitionFunction{\hcPPGraph{X}}{\left(1 - \volumeError\right)\hcPPWeight{X}}.
	\end{align*}
	
	Next, by \cref{lemma:discretization_error:eq_bound_weight,lemma:discretization_error:eq_rewrite_hc}, we know that
	\begin{align*}
		\hcPartitionFunction{\hcPPGraph{X}}{\left(1 - \volumeError\right)\hcPPWeight{X}}
		&\le \sum_{\numParticles \in \N} \sum_{\typeOf\colon [\numParticles] \to [\numTypes]} \int_{\volume^{\numParticles}} \generalWeight{\interactionMatrix}{\fugacity}[\allocation[\mathbold{x}], \typeOf]\, \text{d} \lebesgue{\dimension \times \numParticles} \\
		&=	\allocationPartitionFunction{\volume}{\interactionMatrix}{\fugacity}{\allocation} .
	\end{align*}
	Thus, we have
	\begin{align*}
		\sum_{\numParticles \in \N} \sum_{\typeOf\colon [\numParticles] \to [\numTypes]} &\int_{\volume^{\numParticles}}
		 \absolute{\generalWeight{\interactionMatrix}{\fugacity}[\allocation[\mathbold{x}], \typeOf] - \hcPPWeightContinuous{\hcPPGraph{X}}{\hcPPWeight{X}}{\allocation}[\mathbold{x}, \typeOf]} \, \text{d} \lebesgue{\dimension \times \numParticles} \\
		&\le \left( \eulerE^{\frac{8}{\size{X}} \sum_{i \in [\numTypes]} \fugacity[i]^2 \vol{\volume}^2} \eulerE^{2 \volumeError \sum_{i \in [\numTypes]} \fugacity[i] \vol{\volume}} - 1 \right) \hcPartitionFunction{\hcPPGraph{X}}{\left(1 - \volumeError\right)\hcPPWeight{X}} \\
		&\le \left( \eulerE^{\frac{8}{\size{X}} \sum_{i \in [\numTypes]} \fugacity[i]^2 \vol{\volume}^2} \eulerE^{2 \volumeError \sum_{i \in [\numTypes]} \fugacity[i] \vol{\volume}} - 1 \right) \allocationPartitionFunction{\volume}{\interactionMatrix}{\fugacity}{\allocation} ,
	\end{align*}
	which concludes the proof.
\end{proof}

We now apply \Cref{lemma:discretization_error:1,lemma:discretization_error:2} to prove \Cref{lemma:discretization_error}.

\begin{proof}[Proof of \Cref{lemma:discretization_error}]
	By \Cref{lemma:discretization_error:1} we have
	\begin{align*}
		\allocationPartitionFunction{\volume}{\interactionMatrix}{\fugacity}{\allocation}
		&= \sum_{\numParticles \in \N} \sum_{\typeOf\colon  [\numParticles] \to [\numTypes]} \int_{\volume^{\numParticles}} \generalWeight{\interactionMatrix}{\fugacity}[\allocation[\mathbold{x}], \typeOf]\, \text{d} \lebesgue{\dimension \times \numParticles} \\
		&\le \eulerE^{\left(\frac{4  \distanceError}{\interactionMatrix_{\min}}\right)^{\dimension } \sum_{i \in [\numTypes]} \fugacity[i] \vol{\volume}} \generalPartitionFunction{\volume}{\interactionMatrix}{\fugacity} .
	\end{align*}
	Due to \Cref{lemma:discretization_error:2}, this implies
	\begin{align}
		\notag
		\sum_{\numParticles \in \N}
		&\sum_{\typeOf\colon [\numParticles] \to [\numTypes]} \int_{\volume^{\numParticles}} \absolute{\generalWeight{\interactionMatrix}{\fugacity}[\allocation[\mathbold{x}], \typeOf] - \hcPPWeightContinuous{\hcPPGraph{X}}{\hcPPWeight{X}}{\allocation}[\mathbold{x}, \typeOf]} \, \text{d} \lebesgue{\dimension \times \numParticles} \\
		\label[ineq]{lemma:discretization_error:eq_bound2}
		& \le \left(\eulerE^{\frac{8}{\size{X}} \sum_{i \in [\numTypes]} \fugacity[i]^2 \vol{\volume}^2} \eulerE^{2 \volumeError \sum_{i \in [\numTypes]} \fugacity[i] \vol{\volume}} - 1\right)\eulerE^{\left(\frac{4 \distanceError}{\interactionMatrix_{\min}}\right)^{\dimension} \sum_{i \in [\numTypes]} \fugacity[i] \vol{\volume}} \generalPartitionFunction{\volume}{\interactionMatrix}{\fugacity}.
	\end{align}
	
	Further, by the triangle inequality, we have that
	\begin{align}
		\notag
		\sum_{\numParticles \in \N}
		\sum_{\typeOf\colon [\numParticles] \to [\numTypes]} \int_{\volume^{\numParticles}}
		& \absolute{\generalWeight{\interactionMatrix}{\fugacity}[\mathbold{x}, \typeOf] - \hcPPWeightContinuous{\hcPPGraph{X}}{\hcPPWeight{X}}{\allocation}[\mathbold{x}, \typeOf]} \, \text{d} \lebesgue{\dimension \times \numParticles} \\
		\label[term]{lemma:discretization_error:eq_part1}
		&\le \sum_{\numParticles \in \N} \sum_{\typeOf\colon [\numParticles] \to [\numTypes]} \int_{\volume^{\numParticles}} \absolute{\generalWeight{\interactionMatrix}{\fugacity}[\mathbold{x}, \typeOf] - \generalWeight{\interactionMatrix}{\fugacity}[\allocation[\mathbold{x}], \typeOf]} \, \text{d} \lebesgue{\dimension \times \numParticles}  \\
		\label[term]{lemma:discretization_error:eq_part2}
		&\quad+ \sum_{\numParticles \in \N} \sum_{\typeOf\colon [\numParticles] \to [\numTypes]} \int_{\volume^{\numParticles}} \absolute{\generalWeight{\interactionMatrix}{\fugacity}[\allocation[\mathbold{x}], \typeOf] - \hcPPWeightContinuous{\hcPPGraph{X}}{\hcPPWeight{X}}{\allocation}[\mathbold{x}, \typeOf]} \, \text{d} \lebesgue{\dimension \times \numParticles} .
	\end{align}
	Combining \Cref{lemma:discretization_error:1} and \cref{lemma:discretization_error:eq_bound2} we obtain
	\begin{align*}
		\sum_{\numParticles \in \N} \sum_{\typeOf\colon [\numParticles] \to [\numTypes]} \int_{\volume^{\numParticles}}
		& \absolute{\generalWeight{\interactionMatrix}{\fugacity}[\mathbold{x}, \typeOf] - \hcPPWeightContinuous{\hcPPGraph{X}}{\hcPPWeight{X}}{\allocation}[\mathbold{x}, \typeOf]} \, \text{d} \lebesgue{\dimension \times \numParticles} \\
		&\le \left(\eulerE^{\frac{8}{\size{X}} \sum_{i \in [\numTypes]} \fugacity[i]^2 \vol{\volume}^2} \eulerE^{\left( 2 \volumeError + \left(\frac{4 \distanceError }{\interactionMatrix_{\min}} \right)^{\dimension} \right) \sum_{i \in [\numTypes]} \fugacity[i] \vol{\volume}} - 1\right) \generalPartitionFunction{\volume}{\interactionMatrix}{\fugacity} ,
	\end{align*}
	which proves the claim.
\end{proof}

As \Cref{lemma:discretization_error} immediately implies a bound on  $\absolute{\generalPartitionFunction{\volume}{\interactionMatrix}{\fugacity} - \hcPPPartitionFunctionContinuous{\volume}{\hcPPGraph{X}}{\hcPPWeight{X}}{\allocation}}$, this leads us to the question how to compute $\hcPPPartitionFunctionContinuous{\volume}{\hcPPGraph{X}}{\hcPPWeight{X}}{\allocation}$ efficiently.
To this end, the following identity will come in handy.

\begin{lemma}
	\label{lemma:hcPPPartitionFunction_equivalence}
	Let $(\volume, \interactionMatrix, \fugacity)$ be a hard-constraint point process with $\numTypes \in \N_{\ge 1}$ particle types and $\volume \subset \R^{\dimension}$.
	Further, let $X \subseteq \volume$ be finite and non-empty, and let $(\hcPPGraph{X}, \hcPPWeight{X})$ be the hard-core representation of $(\volume, \interactionMatrix, \fugacity)$ based on $X$.
	For all $\volumeError \in \left[0, 1\right)$ and $\distanceError \in \R_{>0}$, and all $\volumeError$-$\distanceError$-allocations $\allocation$ for~$X$, we have
	\[
		\hcPPPartitionFunctionContinuous{\volume}{\hcPPGraph{X}}{\hcPPWeight{X}}{\allocation} = \hcPartitionFunction{\hcPPGraph{X}}{\hcPPWeight{X}} .
		\qedhere
	\]
\end{lemma}
\begin{proof}
	We prove this claim by rewriting $\hcPartitionFunction{\hcPPGraph{X}}{\hcPPWeight{X}}$ in terms of tuples of vertices, yielding
	\begin{align*}
		\hcPartitionFunction{\hcPPGraph{X}}{\hcPPWeight{X}}
		&= \sum_{\numParticles \in \N} \sum_{\typeOf\colon [\numParticles] \to [\numTypes]} \sum_{\mathbold{x} \in X^{\numParticles}} \frac{1}{\numParticles!} \left(\prod_{i \in [\numParticles]}\hcPPWeight{X}[\hcPPVertex{x_i}{\typeOf[i]}]\right) \hcPPValidContinuous{\typeOf}{\hcPPGraph{X}}{\allocation}[\mathbold{x}] \\
		& = \sum_{\numParticles \in \N} \sum_{\typeOf\colon [\numParticles] \to [\numTypes]} \int_{\volume^{\numParticles}} \frac{1}{\numParticles!} \left(\prod_{i \in [\numParticles]} \frac{\hcPPWeight{X}[\hcPPVertex{\allocation[x_i]}{\typeOf[i]}]}{\vol{\invAllocation[\allocation[x_i]]}}\right) \hcPPValidContinuous{\typeOf}{\hcPPGraph{X}}{\allocation}[\mathbold{x}]\, \text{d} \lebesgue{\dimension \times \numParticles} \\
		& = \sum_{\numParticles \in \N} \sum_{\typeOf\colon [\numParticles] \to [\numTypes]} \int_{\volume^{\numParticles}} \hcPPWeightContinuous{\hcPPGraph{X}}{\hcPPWeight{X}}{\allocation}[\mathbold{x}, \typeOf] \, \text{d} \lebesgue{\dimension \times \numParticles} \\
		&= \hcPPPartitionFunctionContinuous{\volume}{\hcPPGraph{X}}{\hcPPWeight{X}}{\allocation} .
		\qedhere
	\end{align*}
\end{proof}

Note that $\hcPartitionFunction{\hcPPGraph{X}}{\hcPPWeight{X}}$ is simply the partition function of the discrete hard-core model, which, as we discussed in the introduction, is well studied from a computational perspective.
Combining \Cref{lemma:discretization_error,lemma:hcPPPartitionFunction_equivalence} immediately yields the following bound on the difference between $\generalPartitionFunction{\volume}{\interactionMatrix}{\fugacity}$ and $\hcPartitionFunction{\hcPPGraph{X}}{\hcPPWeight{X}}$.
\DiscretisationError
% \begin{theorem}
% 	\label{thm:discretization_error}
% 	Let $(\volume, \interactionMatrix, \fugacity)$ be a hard-constraint point process with $\numTypes \in \N_{\ge 1}$ particle types, and assume $\volume \subset \R^{\dimension}$ is star-convex.
% 	Further, set $\fugacity_{\max} = \max_{i \in [\numTypes]} \fugacity[i]$ and set $\radius_{\min} = \min_{i \in [\numTypes]} \radius[i]$.
% 	Let $X \subseteq \volume$ with $4 \fugacity_{\max} \vol{\volume} \le \size{X} < \infty$ and let $(\hcPPGraph{X}, \hcPPWeight{X})$ be the hard-core representation of $(\volume, \interactionMatrix, \fugacity)$ based on $X$.
% 	If, for some $\volumeError \in \left[0, \frac{1}{2}\right]$ and $\distanceError \in \left[0, \radius_{\min}\right]$, $\volume$ has a $\volumeError$-$\distanceError$-allocation for $X$, then it holds that
% 	\[
% 		\absolute{\hcPartitionFunction{\hcPPGraph{X}}{\hcPPWeight{X}} - \generalPartitionFunction{\volume}{\interactionMatrix}{\fugacity}}
% 		\le \left(\eulerE^{\frac{8}{\size{X}} \sum_{i \in [\numTypes]} \fugacity[i]^2 \vol{\volume}^2} \eulerE^{\left( 2 \volumeError + \left(\frac{2 \distanceError }{\radius_{\min}} \right)^{\dimension} \right) \sum_{i \in [\numTypes]} \fugacity[i] \vol{\volume}} - 1\right) \generalPartitionFunction{\volume}{\interactionMatrix}{\fugacity}.
% 		\qedhere
% 	\]
% \end{theorem}
%\begin{proof}
%	The claim follows immediately by combining \Cref{lemma:discretization_error,lemma:hcPPPartitionFunction_equivalence}.
%\end{proof}

We would like to add a remark about the error term, given in \Cref{thm:discretization_error}.
While it seems to depend on three quantities, namely $\volumeError$, $\distanceError$ and $\size{X}$, it usually suffices to focus on $\volumeError$ and $\distanceError$.
This is because decreasing $\distanceError$ requires to increase $\size{X}$.
One way to see this is to observe that for any $\distanceError \in \R_{>0}$, a lower bound for the minimum size of $X \subset \volume$ such that a $\volumeError$-$\distanceError$-allocation exists is $\frac{\vol{\volume}}{\vol{\ball{2 \distanceError}}}$.
Thus, choosing $\distanceError$ sufficiently small typically results in sufficiently large $\size{X}$.

Another interesting aspect is the role of~$\allocation$.
To construct the hard-core representation for the continuous model and apply \Cref{thm:discretization_error}, it is sufficient to know a set of points $X \subset \volume$ such that a suitable $\volumeError$-$\distanceError$-allocation exists.
However, the allocation itself is not required to be known for constructing the discretization.
This is because, even though \Cref{lemma:discretization_error} seems to depend on the specific allocation $\allocation$, \Cref{lemma:hcPPPartitionFunction_equivalence} shows that $\hcPPPartitionFunctionContinuous{\volume}{\hcPPGraph{X}}{\hcPPWeight{X}}{\allocation} = \hcPartitionFunction{\hcPPGraph{X}}{\hcPPWeight{X}}$ independent of $\allocation$.
For our algorithmic results in \Cref{sec:algo}, we focus on the case where $\volume = [0, \sidelength)^{\dimension}$ for some $\sidelength \in \R_{> 0}$ (i.e., $\volume$ is a hypercube).
In this case, we provide an explicit allocation.
Nevertheless, noting that an explicit construction is not required might be of its own interest when generalizing the results to different classes of continuous regions $\volume \subset \R^{\dimension}$.

\section{Canonical discretization for cubic regions} \label{sec:canonical_discretization}
As mentioned earlier, our algorithmic results focus on cubic regions $\volume = [0, \sidelength)^{\dimension}$, $\sidelength \in \R_{> 0}$.
Note that such regions are star-convex.
In what follows, we show how to construct, for any given $\distanceError \in \R_{>0}$, a set of points $X \subset \volume$ such that there is a $0$-$\distanceError$-allocation for $X$ and with $\size{X}$ polynomial in $\vol{\volume}$ and $\frac{1}{\distanceError}$.
We refer to this as the \emph{canonical point set} and to the resulting graph as the \emph{canonical discretization}.

Formally, the canonical point set is parameterized by a resolution $\resolution \in \R_{> 0}$.
Let $\grid{n} = [0, n)^{\dimension} \cap \N^{\dimension}$ be a $\dimension$-dimensional positive integer grid.
For all cubic regions $\volume = [0, \sidelength)^{\dimension}$, $\sidelength \in \R_{> 0}$, we call $\resolution \in \R_{> 0}$ a \emph{feasible resolution} if and only if $\sidelength \resolution \in \N$.
Given a feasible resolution $\resolution \in \R_{> 0}$, we define the canonical point set for resolution~\resolution as
\[
	\canonicalPoints{\resolution}
	= \frac{1}{\resolution} \grid{\sidelength \resolution}
	= \left\{\left(x^{(1)}, \dots, x^{(\dimension)}\right) \in \volume ~\middle\vert~ \forall i \in [\dimension] \exists m \in \N\colon x^{(i)} = \frac{m}{\resolution} \right\} .
\]
For a hard-constraint point process $(\volume, \interactionMatrix, \fugacity)$ we now define the canonical discretization to be the hard-core representation $(\hcPPGraph{\canonicalPoints{\resolution}}, \hcPPWeight{\canonicalPoints{\resolution}})$ based on $\canonicalPoints{\resolution}$.
Note that in the special case of the hard-sphere model, this canonical discretization coincides to the procedure proposed in \cite{friedrich2021spectral}.
However, our approach works for a much more general class of point processes and, as the rest of this section shows, we obtain much stronger error bounds on the partition function.

The following statement shows how the existence of an allocation for a canonical point set $\canonicalPoints{\resolution}$ depends on the resolution $\resolution$.

\begin{lemma}
	\label{lemma:canonical_allocation}
	Let $\sidelength \in \R_{>0}$, let $\volume = [0, \sidelength)^{\dimension}$, and let $\distanceError \in \R_{>0}$.
	For all feasible resolutions $\resolution \ge \frac{\sqrt{\dimension}}{\distanceError}$ there is a $0$-$\distanceError$-allocation of $\volume$ for the canonical point set $\canonicalPoints{\resolution}$.
\end{lemma}

\begin{proof}
	We prove our claim by constructing an allocation $\allocation\colon \volume \to \canonicalPoints{\resolution}$ with the desired properties.
	To this end, let $y = \left(y^{(1)}, \dots, y^{(\dimension)}\right) \in \volume$.
	We set
	\[
		\allocation[y] = \left(\frac{\lfloor \resolution y^{(1)} \rfloor}{\resolution}, \dots, \frac{\lfloor \resolution y^{(\dimension)} \rfloor}{\resolution}\right) .
	\]
	That is, $y$ gets mapped to the closest point in $x = \left(x^{(1)}, \dots, x^{(\dimension)}\right) \in \canonicalPoints{\resolution}$ such that $x^{(i)} \le y^{(i)}$ for all $i \in [\dimension]$.

	Now, let $x = \left(x^{(1)}, \dots, x^{(\dimension)}\right) \in \canonicalPoints{\resolution}$ and observe that
	\[
		\invAllocation[x] = \prod_{i \in [\dimension]} \left[x^{(i)}, x^{(i)} + \frac{1}{\resolution}\right),
	\]
	where the product is meant to be the cartesian product.
	Note that $\invAllocation[x]$ is measurable.
	Further, we have $\size{\canonicalPoints{\resolution}} = \resolution^{\dimension} \vol{\volume}$.
	Thus, we have
	\[
		\vol{\invAllocation[x]}
		= \frac{1}{\resolution^{\dimension}}
		= \frac{\vol{\volume}}{\resolution^{\dimension} \vol{\volume}}
		= \frac{\vol{\volume}}{\size{\canonicalPoints{\resolution}}} .
	\]

	Finally, observe that for all $x \in \canonicalPoints{\resolution}$ and all $y \in \invAllocation[x]$, it holds that
	\[
		\dist{x}{y} \le \sqrt{\dimension \frac{1}{\resolution^2}} = \frac{\sqrt{\dimension}}{\resolution} .
	\]
	Thus, for $\resolution \ge \frac{\sqrt{\dimension}}{\distanceError}$ we have $\dist{x}{y} \le \distanceError$, which proves that $\allocation$ is a $0$-$\distanceError$-allocation for $\canonicalPoints{\resolution}$ and concludes the proof.
\end{proof}

We can now use the hard-core partition function of the canonical discretization as an approximation of the partition function of the hard-constraint point process as stated in the following result.

\begin{theorem}
	\label{thm:canonical_discretization_error}
	Let $(\volume, \interactionMatrix, \fugacity)$ be a hard-constraint point process with $\numTypes \in \N_{\ge 1}$ particle types and $\volume = [0, \sidelength)^{\dimension}$ for some $\sidelength \in \R_{>0}$.
	For all $\discretizationError \in (0, 1]$ there exists $\resolution_{\discretizationError} \in \bigTheta{\discretizationError^{-1/\dimension} \vol{\volume}^{1/\dimension}}$ such that for all feasible resolutions $\resolution \ge \resolution_{\discretizationError}$ we have
	\[
		\eulerE^{-\discretizationError} \generalPartitionFunction{\volume}{\interactionMatrix}{\fugacity}
		\le \hcPartitionFunction{\hcPPGraph{\canonicalPoints{\resolution}}}{\hcPPWeight{\canonicalPoints{\resolution}}}
		\le \eulerE^{\discretizationError} \generalPartitionFunction{\volume}{\interactionMatrix}{\fugacity}.
		\qedhere
	\]
\end{theorem}

\begin{proof}
	Let $\interactionMatrix_{\min} = \inf_{i,j \in [\numTypes]} \{\interactionMatrix[i][j] \mid \interactionMatrix[i][j] > 0\}$ and $\fugacity_{\max} = \max_{i \in [\numTypes]} \fugacity[i]$.
	We set
	\[
		\resolution_{\discretizationError} = \sqrt{\dimension} \left(\frac{48 \numTypes \max\left\{\fugacity_{\max}, \fugacity_{\max}^2\right\} \vol{\volume}}{\discretizationError}\right)^{\frac{1}{\dimension}} \cdot \max\left\{1, \frac{4}{\interactionMatrix_{\min}}\right\} .
	\]
	First observe that $\resolution_{\discretizationError}$ satisfies the required asymptotic behavior.

	Now, for all feasible resolutions $\resolution \ge \resolution_{\discretizationError}$ we know by \Cref{lemma:canonical_allocation} that there is a $\volumeError$-$\distanceError$-allocation for the canonical point set $\canonicalPoints{\resolution}$ with $\volumeError = 0$ and
	\[
		\distanceError
		= \frac{\sqrt{\dimension}}{\resolution}
		\le \left(\frac{\discretizationError}{48 \numTypes \max\left\{\fugacity_{\max}, \fugacity_{\max}^2\right\} \vol{\volume}}\right)^{\frac{1}{\dimension}} \cdot \min\left\{1, \frac{\interactionMatrix_{\min}}{4}\right\}.
	\]

	To prove our claim, we aim for applying \Cref{thm:discretization_error} based on the canonical point set $\canonicalPoints{\resolution}$ and the $\volumeError$-$\distanceError$-allocation obtained from \Cref{lemma:canonical_allocation}.
	We start by checking that $\volumeError$, $\distanceError$, and $\size{\canonicalPoints{\resolution}}$ satisfy the conditions of \Cref{thm:discretization_error}.

	To this end, note that it trivially holds that $\volumeError = 0 \in \left[0, \frac{1}{2}\right]$.
	Further, observe that for sufficiently large $\vol{\volume}$, it holds that $\distanceError \in [0, \interactionMatrix_{\min}/2]$.
	Finally, note that
	\[
		\size{\canonicalPoints{\resolution}}
		= \resolution^{\dimension} \sidelength^{\dimension}
		\ge \dimension^{\frac{\dimension}{2}} \frac{48 \numTypes \max\left\{\fugacity_{\max}, \fugacity_{\max}^2\right\} \vol{\volume}^2}{\discretizationError} \cdot \max\left\{1, \left(\frac{4}{\interactionMatrix_{\min}}\right)^{\dimension}\right\} .
	\]
	Thus, for $\vol{\volume} \ge 1$ we have
	\[
		\size{\canonicalPoints{\resolution}} \ge 4 \fugacity_{\max} \vol{\volume}.
	\]

	We can now apply \Cref{thm:discretization_error}.
	Recall that we have $\volumeError = 0$.
	Further, by
	\[
		\distanceError
		\le \left(\frac{\discretizationError}{48 \numTypes \max\left\{\fugacity_{\max}, \fugacity_{\max}^2\right\} \vol{\volume}}\right)^{\frac{1}{\dimension}} \cdot \min\left\{1, \frac{\interactionMatrix_{\min}}{4}\right\}
		\le \left(\frac{\discretizationError}{6 \sum_{i \in [\numTypes]} \fugacity[i] \vol{\volume}}\right)^{\frac{1}{\dimension}} \frac{\interactionMatrix_{\min}}{4}
	\]
	we get
	\[
		\left(\frac{4 \distanceError}{\interactionMatrix_{\min}}\right)^{\dimension} \sum_{i \in [\numTypes]} \fugacity[i] \vol{\volume} \le \frac{\discretizationError}{6} .
	\]
	Observe that for
	\[
		\size{\canonicalPoints{\resolution}}
		\ge \dimension^{\frac{\dimension}{2}} \frac{48 \numTypes \max\left\{\fugacity_{\max}, \fugacity_{\max}^2\right\} \vol{\volume}^2}{\discretizationError} \cdot \max\left\{1, \left(\frac{4}{\interactionMatrix_{\min}}\right)^{\dimension}\right\}
		\ge \frac{48 \sum_{i \in [\numTypes]} \fugacity[i]^2 \vol{\volume}^2}{\discretizationError}
	\]
	it holds that
	\[
		\frac{8}{\size{\canonicalPoints{\resolution}}} \sum_{i \in [\numTypes]} \fugacity[i]^2 \vol{\volume}^2 \le \frac{\discretizationError}{6}.
	\]
	Thus, using \Cref{thm:discretization_error}, we get
	\[
		\absolute{\hcPartitionFunction{\hcPPGraph{\canonicalPoints{\resolution}}}{\hcPPWeight{\canonicalPoints{\resolution}}} - \generalPartitionFunction{\volume}{\interactionMatrix}{\fugacity}} \le \left(\eulerE^{\frac{\discretizationError}{3}} - 1\right) \generalPartitionFunction{\volume}{\interactionMatrix}{\fugacity}.
	\]
	Note that $\eulerE^{\frac{\discretizationError}{3}} \le \eulerE^{\discretizationError}$ and $2 - \eulerE^{\frac{\discretizationError}{3}} \ge \eulerE^{-\discretizationError}$ for $\discretizationError \in (0, 1]$.
	Therefor, we have
	\[
		\eulerE^{-\discretizationError} \generalPartitionFunction{\volume}{\interactionMatrix}{\fugacity}
		\le \hcPartitionFunction{\hcPPGraph{\canonicalPoints{\resolution}}}{\hcPPWeight{\canonicalPoints{\resolution}}}
		\le \eulerE^{\discretizationError} \generalPartitionFunction{\volume}{\interactionMatrix}{\fugacity},
	\]
	which concludes the proof.
\end{proof}

Note that the above statement works for $\volumeError = 0$, meaning that we do not make full use of the power of \Cref{thm:discretization_error}.
However, we use the theorem to its full extend when investigating discretizations based on (uniformly) random point sets in \Cref{sec:random_discretization}.
This does not directly lead to any new algorithmic results, but it points out an interesting relationship between the hard-core partition functions of  random geometric graphs and continuous hard-constraint point processes, which we believe to be interesting in its own right.

For the rest of this section, we prove some properties of the canonical discretization that come in handy when deriving sufficient conditions for an efficient approximation of $\generalPartitionFunction{\volume}{\interactionMatrix}{\fugacity}$ via $\hcPartitionFunction{\hcPPGraph{\canonicalPoints{\resolution}}}{\hcPPWeight{\canonicalPoints{\resolution}}}$.
We start with the following observation that bounds the number of vertices in $\hcPPGraph{\canonicalPoints{\resolution}}$.
\begin{observation}
	\label{obs:canonical_vertices}
	Let $(\volume, \interactionMatrix, \fugacity)$ be a hard-constraint point process with $\numTypes \in \N_{\ge 1}$ particle types and $\volume = [0, \sidelength)^{\dimension}$ for some $\sidelength \in \R_{>0}$.
	For every feasible resolution $\resolution \in \R_{>0}$ it holds that
	\[
		\size{\hcPPVertices{\canonicalPoints{\resolution}}} = \numTypes \cdot  \left(\resolution \sidelength\right)^{\dimension} = \numTypes \resolution^{\dimension} \vol{\volume} .
		\qedhere
	\]
\end{observation}
Note that, by \Cref{obs:canonical_vertices}, \Cref{thm:canonical_discretization_error} implies that we only need $\bigTheta{\frac{\vol{\volume}^{2}}{\discretizationError}}$ vertices in order to ensure that $\hcPartitionFunction{\hcPPGraph{\canonicalPoints{\resolution}}}{\hcPPWeight{\canonicalPoints{\resolution}}}$ is an $\discretizationError$-approximation of $\generalPartitionFunction{\volume}{\interactionMatrix}{\fugacity}$.
Considering for instance the hard-sphere model, this improves the required number of vertices from super-exponential in $\vol{\volume}$ to quadratic, compared to \cite{friedrich2021spectral}.

Another important property of the canonical discretization are the vertex degrees.
Especially upper bounding those degrees is closely related to our ability to approximate the hard-core partition function.
To analyze how the degrees increase in terms of the resolution $\resolution$, we apply the following lemma.
\begin{lemma}
	\label{lemma:gauss_circle}
	For $\dimension \in \N_{\ge 1}$ and $s \in \R_{>0}$, let $\integerSphere{\dimension}[s]$ denote the number of integer points in the ball $\ball{s}[\mathbf{0}]$, where $\mathbf{0}$ is the origin of $\R^{\dimension}$.
	Further, let $\degreeError \in (0, 1]$.
	For all $s \ge \frac{2 \dimension^{3/2}}{\degreeError}$ it holds that
	\[
		\integerSphere{\dimension}[s] \le (1 + \degreeError) \vol{\ball{s}}.
		\qedhere
	\]
\end{lemma}

\begin{proof}
	Consider $\ball{s + \sqrt{\dimension}}[\mathbf{0}]$ and observe that for every integer point $(x^{(1)}, \dots, x^{(\dimension)}) \in \ball{s}[\mathbf{0}] \cap \Z^{\dimension}$ it holds that the hypercube $[x^{(1)}, x^{(1)} + 1) \times \dots \times [x^{(\dimension)}, x^{(\dimension)} + 1)$ of volume $1$ is completely contained in $\ball{s + \sqrt{\dimension}}[\mathbf{0}]$.
	Further, note that for different integer points, the corresponding hypercubes are disjoint.
	Thus, we obtain
	\[
		\integerSphere{\dimension}[s] \le \vol{\ball{s + \sqrt{\dimension}}}.
	\]

	It remains to bound the volume $\vol{\ball{s + \sqrt{\dimension}}}$ to get the desired result.
	First, note that
	\[
		\vol{\ball{s + \sqrt{\dimension}}}
		= \vol{\ball{s}} \left(1 + \frac{\sqrt{\dimension}}{s}\right)^{\dimension}
		\le \eulerE^{\frac{\dimension^{3/2}}{s}} \vol{\ball{s}} .
	\]
	By $s \ge \frac{2 \dimension^{3/2}}{\degreeError}$ we have $\frac{\dimension^{3/2}}{s} \le 1$ and $\frac{2 \dimension^{3/2}}{s} \le \degreeError$.
	Thus, we get
	\[
		\eulerE^{\frac{\dimension^{3/2}}{s}}
		\le 1 + \frac{2 \dimension^{3/2}}{s}
		\le 1 + \degreeError
	\]
	which concludes the proof.
\end{proof}

Before we use this lemma to bound the vertex degrees in $\hcPPGraph{\canonicalPoints{\resolution}}$, we need to introduce some additional notation.
Let $(\volume, \interactionMatrix, \fugacity)$ be a hard-constraint point process with $\numTypes \in \N_{\ge 1}$ particle types and $\volume = [0, \sidelength)^{\dimension}$ for some $\sidelength \in \R_{>0}$.
For any feasible resolution $\resolution \in \R_{>0}$ and any vertex $\hcPPVertex{x}{i} \in \hcPPVertices{\canonicalPoints{\resolution}}$ we write $\hcPPDegree{\canonicalPoints{\resolution}}[\hcPPVertex{x}{i}]$ for the number of neighbors of $\hcPPVertex{x}{i}$ in $\hcPPGraph{\canonicalPoints{\resolution}}$.
Further, for any $j \in [\numTypes]$, we write $\hcPPDegree{\canonicalPoints{\resolution}}[\hcPPVertex{x}{i}][j]$ to denote the number of neighbors of $\hcPPVertex{x}{i}$ that are from $\hcPPVertices{\canonicalPoints{\resolution}}[j]$.
The following lemma gives a bound on the vertex degrees in the canonical discretization in terms of the resolution $\resolution$ and the volume exclusion matrix $\volumeMatrix$ of the model.

\begin{lemma}
	\label{lemma:canonical_degree_bounds}
	Let $(\volume, \interactionMatrix, \fugacity)$ be a hard-constraint point process with $\numTypes \in \N_{\ge 1}$ particle types and $\volume = [0, \sidelength)^{\dimension}$ for some $\sidelength \in \R_{>0}$.
	Let $\volumeMatrix$ be the corresponding volume exclusion matrix.
	Set $\interactionMatrix_{\min} = \inf_{i,j \in [\numTypes]} \{\interactionMatrix[i][j] \mid \interactionMatrix[i][j] > 0\}$ and let $\degreeError \in (0, 1]$.
	For all feasible resolutions $\resolution \ge \frac{2 \dimension^{3/2}}{\degreeError \interactionMatrix_{\min}}$ we have for all $i, j \in [\numTypes]$ and $x \in \canonicalPoints{\resolution}$ that
	\[
		\hcPPDegree{\canonicalPoints{\resolution}}[\hcPPVertex{x}{i}][j] \le (1 + \degreeError) \resolution^{\dimension} \volumeMatrix[i][j]
	\]
	and consequently
	\[
		\hcPPDegree{\canonicalPoints{\resolution}}[\hcPPVertex{x}{i}] \le (1 + \degreeError)  \resolution^{\dimension} \sum_{j \in [\numTypes]} \volumeMatrix[i][j] .
		\qedhere
	\]
\end{lemma}

\begin{proof}
	First, note that the bound on $\hcPPDegree{\canonicalPoints{\resolution}}[\hcPPVertex{x}{i}]$ follows immediately from the bound on $\hcPPDegree{\canonicalPoints{\resolution}}[\hcPPVertex{x}{i}][j]$ as the vertex sets $\hcPPVertices{\canonicalPoints{\resolution}}[j]$ for $j \in [\numTypes]$ are a partitioning of $\hcPPVertices{\canonicalPoints{\resolution}}$.

	We proceed by bounding $\hcPPDegree{\canonicalPoints{\resolution}}[\hcPPVertex{x}{i}][j]$.
	Let $y \in \canonicalPoints{\resolution}$ and observe that there is only an edge between $\hcPPVertex{x}{i}$ and $\hcPPVertex{y}{j}$ in $\hcPPGraph{\canonicalPoints{\resolution}}$ if $\dist{x}{y} < \interactionMatrix[i][j]$.
	Further, note that for every fixed $i, j \in [\numTypes]$ and every $x \in \canonicalPoints{\resolution}$ there are at most $\integerSphere{\dimension}[\resolution \interactionMatrix[i][j]]$ points $y \in \canonicalPoints{\resolution}$ such that $\dist{x}{y} < \interactionMatrix[i][j]$.
	This is
	\[
		\hcPPDegree{\canonicalPoints{\resolution}}[\hcPPVertex{x}{i}][j]
		\le \integerSphere{\dimension}[\resolution \interactionMatrix[i][j]] .
	\]

	Further, note that for $\resolution \ge \frac{2\dimension^{3/2}}{\degreeError \interactionMatrix_{\min}}$ we have $\resolution \interactionMatrix[i][j] \ge \frac{2 \dimension^{3/2}}{\degreeError}$.
	Thus, using \Cref{lemma:gauss_circle}, we get
	\[
		\hcPPDegree{\canonicalPoints{\resolution}}[\hcPPVertex{x}{i}][j]
		\le (1 + \degreeError) \vol{\ball{\resolution \interactionMatrix[i][j]}}
		=  (1 + \degreeError) \resolution^{\dimension} \volumeMatrix[i][j] ,
	\]
	which proves the claim.
\end{proof}

\subsection{Tightness of the error bound} \label{subsec:tightness}
In what follows, we argue that the bound for the number of points, required for the canonical discretization, might be considered as asymptotically tight in the volume $\vol{\volume}$ and the discretization error $\discretizationError$.
Recall that, by \Cref{thm:canonical_discretization_error}, we require a resolution of $\resolution \in \bigTheta{\discretizationError^{-1/\dimension} \vol{\volume}^{-1/\dimension}}[\big]$, resulting in $\size{\canonicalPoints{\resolution}} \in \bigTheta{\discretizationError^{-1} \vol{\volume}^2}$ points to ensure that
\[
	\eulerE^{-\discretizationError} \generalPartitionFunction{\volume}{\interactionMatrix}{\fugacity}
	\le \hcPartitionFunction{\hcPPGraph{\canonicalPoints{\resolution}}}{\hcPPWeight{\canonicalPoints{\resolution}}}
	\le \eulerE^{\discretizationError} \generalPartitionFunction{\volume}{\interactionMatrix}{\fugacity}.
\]

To argue that this is in general optimal, we construct a simple case in which we actually require $\bigOmega{\discretizationError^{-1} \vol{\volume}^2}$.
Consider a hard-constrain point process $(\volume, \interactionMatrixZero, \fugacity)$ with $\numTypes = 1$ particle type and $\interactionMatrixZero$ being the $1 \times 1$-matrix with entry $0$.
That is, $(\volume, \interactionMatrixZero, \fugacity)$ is a simple Poisson point process on $\volume$ with intensity $\fugacity$ (we treat $\fugacity$ as a constant here and omit the function notation).
The partition function of $(\volume, \interactionMatrixZero, \fugacity)$ is then
\[
	\generalPartitionFunction{\volume}{\interactionMatrixZero}{\fugacity} = \eulerE^{\fugacity \vol{\volume}} .
\]
Note that for all $\numPoints \in \N_{\ge 1}$ and all $X \subset \volume$ with $\size{X} = \numPoints$, it holds that
\[
	\hcPartitionFunction{\hcPPGraph{X}}{\hcPPWeight{X}}
	= \left(1 + \hcPPWeight{X}\right)^{\size{X}}
	= \left(1 + \frac{\fugacity \vol{\volume}}{\numPoints}\right)^{\numPoints} ,
\]
independent of the specific set of points $X$.

We proceed by lower-bounding the multiplicative error between these two partition functions.
The following elementary bound helps us to do so.
\begin{lemma}
	\label{lemma:quadratic_gap}
	For all $x, y \in \R_{>0}$ with $y \ge x$, it holds that
	\[
		\left(1 + \frac{x}{y}\right)^{y}
		\le \eulerE^{-\frac{x^2}{6 y}} \cdot \eulerE^{x} .
		\qedhere
	\]
\end{lemma}
\begin{proof}
	First note that
	\[
		\left(1 + \frac{x}{y}\right)^{y} = \eulerE^{y \ln\left(1 + \frac{x}{y}\right)} .
	\]
	Further, as $0 \le \frac{x}{y} \le 1$, we use the Taylor expansion at~$1$ and obtain
	\begin{align*}
		\ln\left(1 + \frac{x}{y}\right)
		= \sum_{k \in \N_{\ge 1}} (-1)^{k+1} \frac{x^k}{k y^k}
		\le \frac{x}{y} - \frac{x^2}{2 y^2} + \frac{x^3}{3 y^3}
		\le \frac{x}{y} - \frac{x^2}{6 y^2} .
	\end{align*}
	Thus, we conclude that
	\[
		\left(1 + \frac{x}{y}\right)^{y}
		\le \eulerE^{y \left(\frac{x}{y} - \frac{x^2}{6 y^2}\right)}
		= \eulerE^{-\frac{x^2}{6 y}} \cdot \eulerE^{x} .
        \qedhere
	\]
\end{proof}

The following statement is an implication of \Cref{lemma:quadratic_gap}.
\begin{lemma}
	\label{lemma:quadratic_gap_handy}
	Let $x > 6$ and $\discretizationError \in (0, 1]$.
	For all $y \in \big(0, \frac{x^2}{6 \discretizationError}\big)$, it holds that
	\[
		\left(1 + \frac{x}{y}\right)^{y} < \eulerE^{-\discretizationError} \cdot \eulerE^{x} .
        \qedhere
	\]
\end{lemma}
\begin{proof}
	Let $y_0 \in [x, \frac{x^2}{2 \discretizationError})$ and observe that by \Cref{lemma:quadratic_gap} we have
	\[
		\left(1 + \frac{x}{y_0}\right)^{y_0}
		\le \eulerE^{- \frac{x^2}{6 y_0}} \cdot \eulerE^{x}
		< \eulerE^{-\discretizationError} \cdot \eulerE^{x} .
	\]
	Further, let $y \in (0, x)$ and note that $y < y_0$.
	Thus, by Bernoullis inequality, we have
	\[
		1 + \frac{x}{y} \le \left(1 + \frac{x}{y_0}\right)^{\frac{y_0}{y}} .
	\]
	It follows that 
	\[
		\left(1 + \frac{x}{y}\right)^{y}
		\le \left(1 + \frac{x}{y_0}\right)^{y_0}
		< \eulerE^{-\discretizationError} \cdot \eulerE^{x} ,
	\]
	which proves the claim.
\end{proof}

We obtain the following result.
\begin{proposition}
	\label{prop:lower_bound}
	Let $(\volume, \interactionMatrixZero, \fugacity)$ be a hard-constraint point process with $\numTypes = 1$ particle type and $\interactionMatrixZero$ being the $1 \times 1$-matrix with entry $0$.
	Assume that $\vol{\volume} > \frac{6}{\fugacity}$.
	For all $\discretizationError \in (0, 1]$ and all non-empty $X \subset \volume$ with $\size{X} = \numPoints < \frac{\lambda^2 \vol{\volume}^2}{6 \discretizationError}$, it holds that
	\[
		\hcPartitionFunction{\hcPPGraph{X}}{\hcPPWeight{X}} < \eulerE^{-\discretizationError} \generalPartitionFunction{\volume}{\interactionMatrixZero}{\fugacity} ,
	\]
	meaning that $\hcPartitionFunction{\hcPPGraph{X}}{\hcPPWeight{X}}$ is \emph{not} an $\discretizationError$-approximation of $\generalPartitionFunction{\volume}{\interactionMatrixZero}{\fugacity}$.
\end{proposition}
\begin{proof}
	Recall that
	\[
		\generalPartitionFunction{\volume}{\interactionMatrixZero}{\fugacity} = \eulerE^{\fugacity \vol{\volume}}
		\text{  and  }
		\hcPartitionFunction{\hcPPGraph{X}}{\hcPPWeight{X}} = \left(1 + \frac{\fugacity \vol{\volume}}{\numPoints}\right)^{\numPoints} .
	\]
	The claim follows from applying \Cref{lemma:quadratic_gap_handy} with $x = \fugacity \vol{\volume} > 6$ and $y = \numPoints \in \big(0, \frac{x^2}{6 \discretizationError}\big)$.
\end{proof}

We would like to add a few remarks about \Cref{prop:lower_bound}.
Firstly, note that it applies to all measurable and bounded regions $\volume$ with sufficiently large volume.
Thus, it especially holds for cubic regions $\volume = [0, \sidelength)^{\dimension}$ as investigated in \Cref{thm:canonical_discretization_error}.
Furthermore, the considered hard-constraint point process is the most simple one that can be constructed, namely a point process without any constraints.
This means, \Cref{prop:lower_bound} does not give a lower bound for the number of points, required for discretizing any given hard-constraint point process.
However, it gives a lower bound for the general class of models, and intuition suggests that additional constraints should not decrease the number of points that are required.
Finally, the result only applies to our method of discretization.
That is, it proves that our bounds for the number of points are tight for our discretization method, but it does not exclude the existence of more efficient discretizations.

\section{Efficient approximation} \label{sec:algo}
In this section, we investigate sufficient conditions for approximating the partition function of hard-constraint point processes on cubic regions $\volume = [0, \sidelength)^{\dimension}$ with $\sidelength \in \R_{>0}$, using the canonical discretization proposed in \Cref{sec:canonical_discretization}.
The characterization of cases for which we give an efficient approximation of the partition function is based on the volume exclusion matrix~\volumeMatrix, defined in \Cref{sec:model}.
We split this up into two main parts.

First, we investigate a condition that is especially useful in the setting of uniform fugacities (i.e., $\fugacity$ is a constant function on $[\numTypes]$).
In this setting, the canonical discretization results in a univariate hard-core model, which allows us to apply the variety of algorithmic tools that are known for this model.

In the second part, we focus on the case of non-uniform fugacities.
Our algorithms in this setting are based on tools that where initially developed for abstract polymer models, which are multivariate hard-core models.

\subsection{Uniform fugacity}
Consider a hard-constraint point process $(\volume, \interactionMatrix, \fugacity)$ with $\numTypes \in \N_{\ge 1}$ particle types where the fugacity $\fugacity$ is a constant function on $[\numTypes]$.
To simplify notation, we use $\fugacity$ to denote the value of that function and omit the function notation here.

Note that for any fixed finite non-empty set of points $X \subseteq \volume$ with corresponding hard-core representation $(\hcPPGraph{X}, \hcPPWeight{X})$ this implies that $\hcPPWeight{X}[\hcPPVertex{x}{i}] = \fugacity \frac{\vol{\volume}}{\size{X}}$ for all $x \in X$ and $i \in [\numTypes]$.
That is, $\hcPPWeight{X}$ is a constant function on $\hcPPVertices{X}$ and $(\hcPPGraph{X}, \hcPPWeight{X})$ is a univariate hard-core model.
For the sake of brevity, we write $\hcPPWeight{X} = \fugacity \frac{\vol{\volume}}{\size{X}}$ to denote that constant.

Based on \Cref{thm:hc_fpras_univariate} and \Cref{thm:hc_fptas_univariate} we can now obtain our main approximation result for the partition function of hard-constraint point processes with uniform fugacities.

\uniformapprox
% \begin{theorem}
% 	\label{thm:uniform_approx}
% 	Let $(\volume, \interactionMatrix, \fugacity)$ be a hard-constraint point process with $\numTypes \in \N_{\ge 1}$ particle types, $\volume = [0, \sidelength)^{\dimension}$ for some $\sidelength \in \R_{>0}$ and $\fugacity$ being a constant.
% 	Let $\volumeMatrix$ be the corresponding volume exclusion matrix and denote by $\lOneNorm{\volumeMatrix}$ its $L_{1}$-norm.
% 	If
% 	\[
% 		\fugacity < \frac{\eulerE}{\lOneNorm{\volumeMatrix}},
% 	\]
% 	then for all $\approxError \in (0, 1]$ there is randomized $\approxError$-approximation algorithm for $\generalPartitionFunction{\volume}{\interactionMatrix}{\fugacity}$ with running time in $\poly{\frac{\vol{\volume}}{\approxError}}$ and a quasipolynomial deterministic approximation algorithm for $\generalPartitionFunction{\volume}{\interactionMatrix}{\fugacity}$ with running time $\left(\frac{\vol{\volume}}{\approxError}\right)^{\bigTheta{\ln \left(\vol{\volume}/\approxError\right)}}$.
% \end{theorem}

\begin{proof}
	Let
	\[
		\fugacity = \left(1 - \fugacityGap\right) \frac{\eulerE}{\lOneNorm{\volumeMatrix}}
	\]
	for some $\fugacityGap \in (0, 1]$.
	Set $\degreeError = \fugacityGap' = \frac{\fugacityGap}{2}$.

	By \Cref{thm:canonical_discretization_error} we know that we can choose a resolution $\resolution \in \bigTheta{\approxError^{-1/\dimension} \vol{\volume}^{1/\dimension}}$ such that the canonical discretization $(\hcPPGraph{\canonicalPoints{\resolution}}, \hcPPWeight{\canonicalPoints{\resolution}})$ satisfies
	\[
		\eulerE^{-\frac{\approxError}{2}}\generalPartitionFunction{\volume}{\interactionMatrix}{\fugacity}
		\le \hcPartitionFunction{\hcPPGraph{\canonicalPoints{\resolution}}}{\hcPPWeight{\canonicalPoints{\resolution}}}
		\le \eulerE^{\frac{\approxError}{2}}\generalPartitionFunction{\volume}{\interactionMatrix}{\fugacity} .
	\]

	To approximate the hard-core partition function of $(\hcPPGraph{\canonicalPoints{\resolution}}, \hcPPWeight{\canonicalPoints{\resolution}})$ we aim for applying \Cref{thm:hc_fpras_univariate} and \Cref{thm:hc_fptas_univariate}.
	Thus, we need to bound the maximum degree of $\hcPPGraph{\canonicalPoints{\resolution}}$ and compare $\hcPPWeight{\canonicalPoints{\resolution}}$ with the corresponding tree threshold.

	Set $\interactionMatrix_{\min} = \min_{i,j \in [\numTypes]} \{\interactionMatrix[i][j] \mid \interactionMatrix[i][j] > 0\}$ and note that, for $\vol{\volume}$ sufficiently large, we can assume that $\resolution \ge \frac{2 \dimension^{3/2}}{\degreeError \interactionMatrix_{\min}}$.
	Thus, by \Cref{lemma:canonical_degree_bounds} we have
	\[
		\hcPPDegree{\canonicalPoints{\resolution}}[\hcPPVertex{x}{i}]
		\le (1 + \degreeError)  \resolution^{\dimension} \sum_{j \in [\numTypes]} \volumeMatrix[i][j] .
	\]
	Let $\hcPPDegree{\canonicalPoints{\resolution}}^{*}$ denote the maximum degree in $\hcPPGraph{\canonicalPoints{\resolution}}$ and observe that
	\[
		\hcPPDegree{\canonicalPoints{\resolution}}^{*} = \max_{x \in \canonicalPoints{\resolution}, i \in [\numTypes]} \hcPPDegree{\canonicalPoints{\resolution}}[\hcPPVertex{x}{i}] \le (1 + \degreeError) \resolution^{\dimension} \lOneNorm{\volumeMatrix}.
	\]
	Further, note that
	\[
		\hcPPWeight{\canonicalPoints{\resolution}}
		= \fugacity \frac{\vol{\volume}}{\size{\canonicalPoints{\resolution}}}
		= \frac{\fugacity}{\resolution^{\dimension}}
	\]
	and so we get
	\[
		\hcPPWeight{\canonicalPoints{\resolution}}
		= (1 - \fugacityGap) \frac{\eulerE}{\resolution^{\dimension} \lOneNorm{\volumeMatrix}}
		\le (1 - \fugacityGap') \frac{\eulerE}{(1+\degreeError) \resolution^{\dimension} \lOneNorm{\volumeMatrix}}
		\le (1 - \fugacityGap') \frac{\eulerE}{\hcPPDegree{\canonicalPoints{\resolution}}^{*}} .
	\]
	Now, observe that $\criticalWeight{\degree}$ converges to $\frac{\eulerE}{\degree}$ from above for large $\degree$.
	Thus, we obtain
	\[
		\hcPPWeight{\canonicalPoints{\resolution}}
		\le (1 - \fugacityGap') \criticalWeight{\hcPPDegree{\canonicalPoints{\resolution}}^{*}},
	\]
	where $\fugacityGap' > 0$.

	Using \Cref{obs:canonical_vertices}, we know that
	\[
		\size{\hcPPVertices{\canonicalPoints{\resolution}}} = \numTypes \resolution^{\dimension} \vol{\volume} \in \bigTheta{\frac{\vol{\volume}^2}{\approxError}} .
	\]
	By \Cref{thm:hc_fpras_univariate} we can immediately conclude the existence of a randomized approximation within the desired running time bounds.
	Further, note that
	\[
		\hcPPDegree{\canonicalPoints{\resolution}}^{*} \le (1 + \degreeError) \resolution^{\dimension} \lOneNorm{\volumeMatrix} \in \bigTheta{\frac{\vol{\volume}}{\approxError}} .
	\]
	Thus, applying \Cref{thm:hc_fptas_univariate} we obtain the deterministic approximation, which concludes the proof.
\end{proof}

We now demonstrate which bounds can be obtained based on \Cref{thm:uniform_approx} for models commonly studied in statistical physics.
We start with the monoatomic hard-sphere model.

% \begin{corollary}\label{cor:monoatomic_hs}
% 	Let $\volume = [0, \sidelength)^{\dimension}$ for some $\sidelength \in \R_{>0}$.
% 	Further, let $\radius \in \R_{>0}$ and $\fugacity \in \R_{\ge 0}$.
% 	Denote by $\hsPartitionFunction{\volume}{\radius}{\fugacity}$ the hard-sphere partition function on $\volume$ with particles of radius $\radius$ and fugacity~$\fugacity$.
% 	If
% 	\[
% 		\fugacity < \frac{\eulerE}{2^{\dimension} \vol{\ball{\radius}}}
% 	\]
% 	 then for all $\approxError \in (0, 1]$ there is a randomized $\approxError$-approximation algorithm for $\hsPartitionFunction{\volume}{\radius}{\fugacity}$ with running time in $\poly{\frac{\vol{\volume}}{\approxError}}$ and a quasipolynomial deterministic $\approxError$-approximation algorithm for $\hsPartitionFunction{\volume}{\radius}{\fugacity}$ with running time $\left(\frac{\vol{\volume}}{\approxError}\right)^{\bigTheta{\ln \left(\vol{\volume}/\approxError\right)}}$.
% \end{corollary}

\monoatomichs
% \begin{corollary}\label{cor:monoatomic_hs}
% 	Let $\volume = [0, \sidelength)^{\dimension}$ for some $\sidelength \in \R_{>0}$.
% 	Further, let $\radius \in \R_{>0}$ and $\fugacity \in \R_{\ge 0}$.
% 	Denote by $\hsPartitionFunction{\volume}{\radius}{\fugacity}$ the hard-sphere partition function on $\volume$ with particles of radius $\radius$ and fugacity $\fugacity$.
% 	If
% 	\[
% 		\fugacity < \frac{\eulerE}{2^{\dimension} \vol{\ball{\radius}}}
% 	\]
% 	 then for all $\approxError \in (0, 1]$ there is a randomized $\approxError$-approximation algorithm for $\hsPartitionFunction{\volume}{\radius}{\fugacity}$ with running time in $\poly{\frac{\vol{\volume}}{\approxError}}$ and a quasipolynomial deterministic $\approxError$-approximation algorithm for $\hsPartitionFunction{\volume}{\radius}{\fugacity}$ with running time $\left(\frac{\vol{\volume}}{\approxError}\right)^{\bigTheta{\ln \left(\vol{\volume}/\approxError\right)}}$.
% \end{corollary}

\Cref{cor:monoatomic_hs} follows from \Cref{thm:uniform_approx} by observing that the interaction matrix corresponding to the hard-sphere model is a $1 \times 1$ matrix containing the value $2 \radius$ and that $\lOneNorm{\volumeMatrix} = \vol{\ball{2 \radius}} = 2^{\dimension} \vol{\ball{\radius}}$.
Due to the fact that we improved the number of points, required for discretization, from super-exponential in $\vol{\volume}$ to quadratic, we can run Glauber dynamics in order to obtain the approximation. 
This greatly simplifies the discretization based randomized algorithm, introduced by \textcite{friedrich2021spectral}.
Further, we present the first efficient deterministic approximation.
This partially answers an open question, posed by \textcite{perkinsAlgorithm}.

\uniformwr
%
% \begin{corollary}\label{cor:uniform_wr}
% 	Let $\volume = [0, \sidelength)^{\dimension}$ for some $\sidelength \in \R_{>0}$.
% 	Further, let $\radius \in \R_{>0}$, $\fugacity \in \R_{\ge 0}$ and $\numTypes \in \N_{\ge 1}$.
% 	Denote by $\wrPartitionFunction{\volume}{\radius}{\fugacity}$ the Widom-Rowlinson partition function on $\volume$ with $\numTypes$ particle types, each of radius $\radius$ and fugacity $\fugacity$.
% 	If
% 	\[
% 		\fugacity < \frac{\eulerE}{(\numTypes - 1) 2^{\dimension} \vol{\ball{\radius}}}
% 	\]
% 	then for all $\approxError \in (0, 1]$ there is a randomized $\approxError$-approximation algorithm for $\wrPartitionFunction{\volume}{\radius}{\fugacity}$ with running time in $\poly{\frac{\vol{\volume}}{\approxError}}$ and a quasipolynomial deterministic $\approxError$-approximation algorithm for $\wrPartitionFunction{\volume}{\radius}{\fugacity}$ with running time $\left(\frac{\vol{\volume}}{\approxError}\right)^{\bigTheta{\ln \left(\vol{\volume}/\approxError\right)}}$.
% \end{corollary}

\Cref{cor:uniform_wr} follows from \Cref{thm:uniform_approx} by observing that the interaction matrix is a $\numTypes \times \numTypes$ matrix, which is $0$ on the diagonal and $2 \radius$ otherwise.
To the best of our knowledge, no rigorous computational result for this model has been known before.

\begin{remark}
	\label{remark:univariate_to_multivariate}
	In fact, based on \Cref{remark:univariate_hc} we can apply \Cref{thm:uniform_approx} also in the non-uniform setting when the condition is replaced by
	\[
		\fugacity_{\max} < \frac{\eulerE}{\lOneNorm{\volumeMatrix}},
	\]
	where $\fugacity_{\max} = \max_{i \in [\numTypes]} \fugacity[i]$.
%	Even though \Cref{thm:hc_fpras_univariate} and \Cref{thm:hc_fptas_univariate} are only stated for the unvariate hard-core model, the results can be extended to the multivariate setting by considering the maximum of the fugacities.
\end{remark}

\subsection{Non-uniform fugacity}
We proceed by looking at which conditions for efficient approximation we can get for hard-constraint point processes with non-uniform fugacities based on \Cref{thm:clique_dynamics_condition_approx}.
\nonUniformApprox
% \todo{find argument to remove running time dependency on $f$}
% \begin{theorem}
% 	\label{thm:non_uniform_approx}
% 	Let $(\volume, \interactionMatrix, \fugacity)$ be a hard-constraint point process with $\numTypes \in \N_{\ge 1}$ particle types and $\volume = [0, \sidelength)^{\dimension}$ for some $\sidelength \in \R_{>0}$.
% 	Further, let $\volumeMatrix$ be the corresponding volume exclusion matrix.
% 	If there is a function $f: [\numTypes] \to \R_{>0}$ such that for all $i \in [\numTypes]$ it holds that
% 	\[
% 		f(i) > \sum_{j \in [\numTypes]} \volumeMatrix[i][j] f(j) \fugacity[j] ,
% 	\]
% 	then for all $\approxError \in (0, 1]$ there is randomized $\approxError$-approximation algorithm for $\generalPartitionFunction{\volume}{\interactionMatrix}{\fugacity}$ with running time in $\poly{\frac{\vol{\volume}}{\approxError}}$.
% \end{theorem}
\begin{proof}
	We prove the theorem by applying \Cref{thm:clique_dynamics_condition_approx} to the canonical discretization $(\hcPPGraph{\canonicalPoints{\resolution}}, \hcPPWeight{\canonicalPoints{\resolution}})$ with a sufficiently large resolution $\resolution$.
	To this end, we construct a function $g\colon \hcPPVertices{\canonicalPoints{\resolution}} \to \R_{>0}$ that satisfies the condition of \Cref{thm:clique_dynamics_condition_approx}.

	Assume that for all $i \in [\numTypes]$ we have
	\[
		f(i) \ge \alpha_i \sum_{j \in [\numTypes]} \volumeMatrix[i][j] f(j) \fugacity[j]
	\]
	for some $\alpha_i > 1$ and set $\alpha = \min_{i \in [\numTypes]} \alpha_i$.
	By \Cref{thm:canonical_discretization_error}, we know that we can choose a resolution in $\bigTheta{\approxError^{-1/\dimension} \vol{\volume}^{1/\dimension}}$ such that
	\[
		\eulerE^{-\frac{\approxError}{2}}\generalPartitionFunction{\volume}{\interactionMatrix}{\fugacity}
		\le \hcPartitionFunction{\hcPPGraph{\canonicalPoints{\resolution}}}{\hcPPWeight{\canonicalPoints{\resolution}}}
		\le \eulerE^{\frac{\approxError}{2}}\generalPartitionFunction{\volume}{\interactionMatrix}{\fugacity} .
	\]
	Set $\interactionMatrix_{\min} = \min_{i,j \in [\numTypes]} \{\interactionMatrix[i][j] \mid \interactionMatrix[i][j] > 0\}$ and note that, for $\vol{\volume}$ sufficiently large, it holds that $\resolution \ge \frac{2 \dimension^{3/2}}{(\alpha - 1) \interactionMatrix_{\min}}$.
	Thus, by \Cref{lemma:canonical_degree_bounds}, we know that for all $x \in \canonicalPoints{\resolution}$ and $i, j \in [\numTypes]$ it holds that
	\[
		\hcPPDegree{\canonicalPoints{\resolution}}[\hcPPVertex{x}{i}][j] \le \alpha \resolution^{\dimension} \volumeMatrix[i][j].
	\]

	Define a function $g\colon \hcPPVertices{\canonicalPoints{\resolution}} \to \R_{>0}$ with $g\left(\hcPPVertex{x}{i}\right) = f(i)$ for all $i \in [\numTypes]$ and all $x \in \canonicalPoints{\resolution}$.
	For all $i \in [\numTypes]$ and all $x \in \canonicalPoints{\resolution}$, we have
	\begin{align*}
		\sum_{w \in \neighbors{\hcPPVertex{x}{i}}} g(w) \frac{\hcPPWeight{\canonicalPoints{\resolution}}[w]}{1 + \hcPPWeight{\canonicalPoints{\resolution}}[w]}
		&\le \sum_{j \in [\numTypes]} \hcPPDegree{\canonicalPoints{\resolution}}[\hcPPVertex{x}{i}][j] f(j) \fugacity[j] \resolution^{-\dimension} \\
		&\le \alpha \sum_{j \in [\numTypes]} \volumeMatrix[i][j] f(j) \fugacity[j] \\
		&\le f(i) \\
		&\le g\left(\hcPPVertex{x}{i}\right) .
	\end{align*}
	Thus, we can apply \Cref{thm:clique_dynamics_condition_approx} to obtain the desired approximation.
\end{proof}

%Note that in the above statement, we omitted the logarithmic running time dependency on the function $f$.
%However, as our discretization has size polynomial in $\vol{\volume}$ and we assume the fugacities to be bounded by a constant, there will be no need to choose $f$ super-exponentially large (or small) in $\vol{\volume}$.
%Thus, we can safely omit this term from the running time. \todo{can we? would be nice}

In the case of uniform fugacities, the approximation result obtained from \Cref{thm:non_uniform_approx} is worse than the one obtained from \Cref{thm:uniform_approx}.
However, in the non-uniform case, \Cref{thm:non_uniform_approx} can lead to interesting parameter regimes that are not covered by \Cref{thm:uniform_approx}.
To demonstrate that, we look into a Widom--Rowlinson model with two particle types of identical radii but different fugacities.
% Note that this is a model of interest, as such unbalanced Windom--Rowlinson models can for example be used to express the area-interaction process. \todo{add reference}

\unbalancedWR
% \begin{corollary}
% 	\label{cor:unbalanced_WR}
% 	Let $\volume = [0, \sidelength)^{\dimension}$ for some $\sidelength \in \R_{>0}$.
% 	Further, let $\radius \in \R_{>0}$, $\fugacity_1, \fugacity_2 \in \R_{\ge 0}$.
% 	Denote by $\wrPartitionFunction{\volume}{\radius}{\fugacity_1, \fugacity_2}$ the Widom-Rowlinson partition function on $\volume$ with $2$ particle types, both with the same radius $\radius$ but (possibly) different fugacities $\fugacity_1, \fugacity_2$.
% 	If
% 	\[
% 		\fugacity_1 \fugacity_2 < \frac{1}{4^{\dimension} \vol{\ball{\radius}}^{2}},
% 	\]
% 	then for all $\approxError \in (0, 1]$ there is a randomized $\approxError$-approximation algorithm for $\wrPartitionFunction{\volume}{\radius}{\fugacity_1, \fugacity_2}$ with running time in $\poly{\frac{\vol{\volume}}{\approxError}}$.
% \end{corollary}
\begin{proof}
	Assume
	\begin{align}
		\label[ineq]{cor:unbalanced_WR:eq1}
		\fugacity_1 \fugacity_2 \le \left(1 - \alpha\right) \frac{1}{4^{\dimension} \vol{\ball{\radius}}^{2}}
	\end{align}
	for some $\alpha \in (0, \frac{1}{2}]$.
	We construct the equivalent hard-constraint point process $(\volume, \interactionMatrix, \fugacity)$ with $\numTypes=2$ particle types, where $\interactions$ is a $2\times2$ matrix with $0$ on the diagonal and $1$ for all off-diagonal entries and $\fugacity[i] = \fugacity_{i}$ for all $i \in \{1, 2\}$.
	Note that \cref{cor:unbalanced_WR:eq1}, $2^{\dimension} \vol{\ball{\radius}} = \vol{\ball{2\radius}}$, and the definition of~\volumeMatrix imply that
	\[
		\volumeMatrix[2][1] \fugacity[1] \volumeMatrix[1][2] \fugacity[2] \le 1 - \alpha .
	\]
	Set $f(1) = 1$ and $f(2) = (1 + \beta) \volumeMatrix[2][1] \fugacity[1]$ for some $\beta \in \left(0, \frac{\alpha}{1 - \alpha}\right)$.
	Observe that
	\[
		\volumeMatrix[2][1] f(1) \fugacity[1] = \volumeMatrix[2][1] \fugacity[1]
		< f(2)
	\]
	and
	\[
		\volumeMatrix[1][2] f(2) \fugacity[2]
		= (1 + \beta) \volumeMatrix[2][1] \fugacity[1] \volumeMatrix[1][2]  \fugacity[2]
		\le (1 + \beta) (1 - \alpha)
		< f(1) .
	\]
	Applying \Cref{thm:non_uniform_approx} concludes the proof.
\end{proof}

In the balanced case $\fugacity_1 = \fugacity_2 = \fugacity$, \Cref{cor:unbalanced_WR} gives a bound on $\fugacity$ that is worse by a factor of $\eulerE$ compared to what we can get from \Cref{thm:uniform_approx} and \Cref{remark:univariate_to_multivariate} (compare \Cref{cor:uniform_wr}).
However, by using \Cref{thm:non_uniform_approx} we see that we can make one of the two fugacities much larger, as long as we scale down the other one appropriately.
For example, according to \Cref{cor:unbalanced_WR} we can get a polynomial time randomized approximation for $\lambda_{1} = \frac{\eulerE^2}{2^{\dimension} \vol{\ball{\radius}}}$ and $\lambda_{2} < \frac{1}{\eulerE^2 2^{\dimension} \vol{\ball{\radius}}}$, which is obviously not possible using \Cref{thm:uniform_approx} and \Cref{cor:uniform_wr}.
This leads to an efficient approximation as long as the fugacities are sufficiently unbalanced.
% \todo{discuss other approaches for unbalanced bipartite graphs?}
% \todo{discuss general theorem for unbalanced bipartite hard-core models?}
% \todo{search for non-phase transition result in that setting}
% \todo{can we also give deterministic result for non-uniform setting?´}

\section{Sampling via random perturbations} \label{sec:sampling}
So far, we only presented approximation results for the partition function $\generalPartitionFunction{\volume}{\interactionMatrix}{\fugacity}$ of a hard-constraint point process $(\volume, \interactionMatrix, \fugacity)$.
As we briefly discussed in the introduction, when it comes to sampling from the Gibbs distribution $\generalGibbs{\volume}{\interactionMatrix}{\fugacity}$, a natural barrier is that outputting a sample (i.e., a tuple of points $\mathbold{x} \in \volume^{\numParticles}$ and a type assignment $\typeOf\colon [\numParticles] \to [\numTypes]$ that represents a valid configuration) would, in fact, require infinite floating-point precision.
% However, in computer science a discrete computational model is commonly assumed, where no (approximate) sampling algorithm with meaningful error bounds in terms of total-variation distance can be obtained.
If we allow ourselves to use a computational model that can handle floating-point values with arbitrary precision, we can actually use our hard-core representation $(\hcPPGraph{X}, \hcPPWeight{X})$ to recover an approximate sampler for $\generalGibbs{\volume}{\interactionMatrix}{\fugacity}$, given that we know a suitable allocation $\allocation$ for $X$. This is captured by the following assumption, required for the algorithmic results of this section.

% For our sampling algorihtms we require the following assumption.
\begin{assumption}
	\label{assumption:floating_pointer}
	The computational model can, in constant time,
	\begin{itemize}
		\item output floating-point numbers with arbitrary precision,
		\item perform basic arithmetic operations (addition, subtraction, ...) on floating-point numbers with arbitrary precision,
		\item compare floating-point numbers with arbitrary precision, and
		\item produce, for all $a, b \in \R$ with $a < b$, a uniformly random floating-point number from the intervals $[a, b]$, $[a, b)$, $(a, b]$, and $(a, b)$.
        \qedhere
	\end{itemize}
\end{assumption}
% We are aware that \Cref{assumption:floating_pointer} is not satisfied by an ordinary computer.\todo{Copy from the intro, once it is done?}
This assumption is only required to give meaningful theoretical error bounds in terms of total-variation distance, which are impossible to achieve otherwise.
% For practical considerations, the floating pointer precision of common discrete computational models might be seen as sufficient for applying our sampling approach.

We proceed by describing our sampling procedure for a given hard-constraint point process $(\volume, \interactionMatrix, \fugacity)$.
Let $X \subset \volume$ be finite and non-empty and let $\allocation$ be an allocation for $X$.
From \Cref{sec:discretization}, recall the definition of the intermediate continuous model, defined via the weight function $\hcPPWeightContinuous{\hcPPGraph{X}}{\hcPPWeight{X}}{\allocation}$ as well as the corresponding Gibbs density $\hcPPGibbsContinuous{\volume}{\hcPPGraph{X}}{\hcPPWeight{X}}{\allocation}$ and partition function $\hcPPPartitionFunctionContinuous{\volume}{\hcPPGraph{X}}{\hcPPWeight{X}}{\allocation}$.
Instead of directly sampling from $\generalPartitionFunction{\volume}{\interactionMatrix}{\fugacity}$, we approximately sample from $\hcPPGibbsContinuous{\volume}{\hcPPGraph{X}}{\hcPPWeight{X}}{\allocation}$.
To this end, we first bound the total-variation distance between these two distributions.
The following lemma helps us use \Cref{lemma:discretization_error} for this purpose. %, by bounding the total-variation distance in terms of
% \[
% 	\sum_{\numParticles \in \N} \sum_{\typeOf\colon [\numParticles] \to [\numTypes]} \int_{\volume^{\numParticles}}  \absolute{\generalWeight{\interactionMatrix}{\fugacity}[\mathbold{x}, \typeOf] - \hcPPWeightContinuous{\hcPPGraph{X}}{\hcPPWeight{X}}{\allocation}[\mathbold{x}, \typeOf]} \, \text{d} \lebesgue{\dimension \times \numParticles}.
% \]
\begin{lemma}
	\label{lemma:dtv_discretization}
	Let $(\volume, \interactionMatrix, \fugacity)$ be a hard-constraint point process with $\numTypes \in \N_{\ge 1}$ particle types.
	Further, let $X \subset \volume$ be finite and non-empty and, for $\volumeError \in [0, 1)$ and $\distanceError \in \R_{>0}$, let $\allocation$ be a $\volumeError$-$\distanceError$-allocation for $X$.
	It holds that
	\[
		\dtv{\generalGibbs{\volume}{\interactionMatrix}{\fugacity}}{\hcPPGibbsContinuous{\volume}{\hcPPGraph{X}}{\hcPPWeight{X}}{\allocation}}
		\le \frac{1}{\generalPartitionFunction{\volume}{\interactionMatrix}{\fugacity}}
		\sum_{\numParticles \in \N} \sum_{\typeOf\colon [\numParticles] \to [\numTypes]} \int_{\volume^{\numParticles}}  \absolute{\generalWeight{\interactionMatrix}{\fugacity}[\mathbold{x}, \typeOf] - \hcPPWeightContinuous{\hcPPGraph{X}}{\hcPPWeight{X}}{\allocation}[\mathbold{x}, \typeOf]} \, \text{d} \lebesgue{\dimension \times \numParticles}.
		\qedhere
	\]
\end{lemma}

\begin{proof}
	First, note that the total-variation distance can be written as
	\begin{align}
		\label{lemma:dtv_discretization:eq_dtv}
		\dtv{\generalGibbs{\volume}{\interactionMatrix}{\fugacity}}{\hcPPGibbsContinuous{\volume}{\hcPPGraph{X}}{\hcPPWeight{X}}{\allocation}}
		=  \frac{1}{2}\sum_{\numParticles \in \N} \sum_{\typeOf\colon [\numParticles] \to [\numTypes]} \int_{\volume^{\numParticles}}  \absolute{\generalGibbs{\volume}{\interactionMatrix}{\fugacity}[\mathbold{x}, \typeOf] - \hcPPGibbsContinuous{\volume}{\hcPPGraph{X}}{\hcPPWeight{X}}{\allocation}[\mathbold{x}, \typeOf]} \, \text{d} \lebesgue{\dimension \times \numParticles} .
	\end{align}
	For all $\mathbold{x} \in \volume^{\numParticles}$ and $\typeOf\colon [\numParticles] \to [\numTypes]$, observe that
	\begin{align*}
		&\absolute{\generalGibbs{\volume}{\interactionMatrix}{\fugacity}[\mathbold{x}, \typeOf] - \hcPPGibbsContinuous{\volume}{\hcPPGraph{X}}{\hcPPWeight{X}}{\allocation}[\mathbold{x}, \typeOf]} \\
		&\hspace{3em}= \absolute{\frac{\generalWeight{\interactionMatrix}{\fugacity}[\mathbold{x}, \typeOf]}{\generalPartitionFunction{\volume}{\interactionMatrix}{\fugacity}} - \frac{\hcPPWeightContinuous{\hcPPGraph{X}}{\hcPPWeight{X}}{\allocation}[\mathbold{x}, \typeOf]}{\hcPPPartitionFunctionContinuous{\volume}{\hcPPGraph{X}}{\hcPPWeight{X}}{\allocation}}} \\
		&\hspace{3em}= \frac{1}{\generalPartitionFunction{\volume}{\interactionMatrix}{\fugacity}} \absolute{\generalWeight{\interactionMatrix}{\fugacity}[\mathbold{x}, \typeOf] - \frac{\generalPartitionFunction{\volume}{\interactionMatrix}{\fugacity}}{\hcPPPartitionFunctionContinuous{\volume}{\hcPPGraph{X}}{\hcPPWeight{X}}{\allocation}}\hcPPWeightContinuous{\hcPPGraph{X}}{\hcPPWeight{X}}{\allocation}[\mathbold{x}, \typeOf]} \\
		&\hspace{3em}\le \frac{1}{\generalPartitionFunction{\volume}{\interactionMatrix}{\fugacity}} \left( \absolute{\generalWeight{\interactionMatrix}{\fugacity}[\mathbold{x}, \typeOf] - \hcPPWeightContinuous{\hcPPGraph{X}}{\hcPPWeight{X}}{\allocation}[\mathbold{x}, \typeOf]} + \absolute{\frac{\generalPartitionFunction{\volume}{\interactionMatrix}{\fugacity}}{\hcPPPartitionFunctionContinuous{\volume}{\hcPPGraph{X}}{\hcPPWeight{X}}{\allocation}} - 1}\hcPPWeightContinuous{\hcPPGraph{X}}{\hcPPWeight{X}}{\allocation}[\mathbold{x}, \typeOf]\right),
	\end{align*}
	where the last inequality follows from the triangle inequality. 
	Substituting into \cref{lemma:dtv_discretization:eq_dtv}, we have
	\begin{align*}
		\dtv{\generalGibbs{\volume}{\interactionMatrix}{\fugacity}}{\hcPPGibbsContinuous{\volume}{\hcPPGraph{X}}{\hcPPWeight{X}}{\allocation}}
        &\le \frac{1}{2\generalPartitionFunction{\volume}{\interactionMatrix}{\fugacity}} \Bigg(
		\sum_{\numParticles \in \N} \sum_{\typeOf\colon [\numParticles] \to [\numTypes]} \int_{\volume^{\numParticles}}  \absolute{\generalWeight{\interactionMatrix}{\fugacity}[\mathbold{x}, \typeOf] - \hcPPWeightContinuous{\hcPPGraph{X}}{\hcPPWeight{X}}{\allocation}[\mathbold{x}, \typeOf]} \, \text{d} \lebesgue{\dimension \times \numParticles} \\
        &\hspace*{2.6 cm}+ \absolute{\generalPartitionFunction{\volume}{\interactionMatrix}{\fugacity} - \hcPPPartitionFunctionContinuous{\volume}{\hcPPGraph{X}}{\hcPPWeight{X}}{\allocation}}\Bigg) \\
		&\le \frac{1}{\generalPartitionFunction{\volume}{\interactionMatrix}{\fugacity}}
		\sum_{\numParticles \in \N} \sum_{\typeOf\colon [\numParticles] \to [\numTypes]} \int_{\volume^{\numParticles}}  \absolute{\generalWeight{\interactionMatrix}{\fugacity}[\mathbold{x}, \typeOf] - \hcPPWeightContinuous{\hcPPGraph{X}}{\hcPPWeight{X}}{\allocation}[\mathbold{x}, \typeOf]} \, \text{d} \lebesgue{\dimension \times \numParticles},
	\end{align*}
	where the last inequality comes from the fact that
	\[
		\absolute{\generalPartitionFunction{\volume}{\interactionMatrix}{\fugacity} - \hcPPPartitionFunctionContinuous{\volume}{\hcPPGraph{X}}{\hcPPWeight{X}}{\allocation}} \le \sum_{\numParticles \in \N} \sum_{\typeOf\colon [\numParticles] \to [\numTypes]} \int_{\volume^{\numParticles}}  \absolute{\generalWeight{\interactionMatrix}{\fugacity}[\mathbold{x}, \typeOf] - \hcPPWeightContinuous{\hcPPGraph{X}}{\hcPPWeight{X}}{\allocation}[\mathbold{x}, \typeOf]} \, \text{d} \lebesgue{\dimension \times \numParticles} .
		\qedhere
	\]
\end{proof}
Note that under the conditions of \Cref{lemma:discretization_error}, \Cref{lemma:dtv_discretization} immediately gives a bound for the total-variation distance between $\generalGibbs{\volume}{\interactionMatrix}{\fugacity}$ and $\hcPPGibbsContinuous{\volume}{\hcPPGraph{X}}{\hcPPWeight{X}}{\allocation}$.

\begin{algorithm}[t]
    \caption{
        \label{alg:perturbation_algorithm}
        The perturbation algorithm.
        It takes an independent set of a given hard-core representation of a hard-constraint point process and an allocation, and it returns a perturbation of the given independent set with respect to the allocation.
    }
    \KwIn
    {%
        hard-constraint point process $(\volume, \interactionMatrix, \fugacity)$ with $\numTypes \in \N_{\ge 1}$ particle types\\
        \phantom{\textbf{Input: }}finite and non-empty $X \subset \volume$\\
        \phantom{\textbf{Input: }}hard-core representation $(\hcPPGraph{X}, \hcPPWeight{X})$ of $(\volume, \interactionMatrix, \fugacity)$ based on $X$\\
        \phantom{\textbf{Input: }}independent set $\independentSet \in \independentSets{\hcPPGraph{X}}$ of cardinality~$\numParticles$\\
        \phantom{\textbf{Input: }}allocation~$\allocation\colon \volume \to X$
    }
    \KwOut
    {%
        $\mathbold{x} \in \volume^{\numParticles}$\\
        \phantom{\textbf{Output: }}$\typeOf\colon [\numParticles] \to [\numTypes]$
    }
    construct a permutation $\big(\hcPPVertex{y_i}{\typeOf_i}\big)_{i \in [\numParticles]}$ of the vertices in $\independentSet$ uniformly at random\;
    \lFor{$i \in [\numParticles]$}
    {%
        $x_i \gets $ a point from $\invAllocation[y_i]$, drawn uniformly at random%
    }
    $\mathbold{x} \gets (x_i)_{i \in [\numParticles]}$\;
    \lFor{$i \in [\numParticles]$}
    {%
        $\typeOf[i] \gets \typeOf_i$%
    }
\end{algorithm}

In order to use $\hcPPGibbsContinuous{\volume}{\hcPPGraph{X}}{\hcPPWeight{X}}{\allocation}$ as an approximation for $\generalGibbs{\volume}{\interactionMatrix}{\fugacity}$, we use \Cref{alg:perturbation_algorithm}.
This algorithm adds random perturbations to a sample from the discretized distribution $\hcGibbs{\hcPPGraph{X}}{\hcPPWeight{X}}$, based on $\invAllocation$.
%This interpretation becomes even more apparent when looking for example at the proof of \Cref{thm:uniform_sampling}, where \Cref{lemma:appx_sampling_hcPPGibbsContinuous} is applied.
Applying \Cref{alg:perturbation_algorithm} shows a connection between $\hcGibbs{\hcPPGraph{X}}{\hcPPWeight{X}}$ and $\hcPPGibbsContinuous{\volume}{\hcPPGraph{X}}{\hcPPWeight{X}}{\allocation}$.
\begin{observation}
	\label{obs:sampling_hcPPGibbsContinuous}
	Let $(\volume, \interactionMatrix, \fugacity)$ be a hard-constraint point process with $\numTypes \in \N_{\ge 1}$ particle types.
	Further, let $X \subset \volume$ be finite and non-empty and, for $\volumeError \in [0, 1)$ and $\distanceError \in \R_{>0}$, let $\allocation$ be a $\volumeError$-$\distanceError$-allocation for $X$.
	Let $(\hcPPGraph{X}, \hcPPWeight{X})$ be the hard-core representation of $(\volume, \interactionMatrix, \fugacity)$ based on $X$, and let $\independentSet \in \independentSets{\hcPPGraph{X}}$ be drawn according to $\hcGibbs{\hcPPGraph{X}}{\hcPPWeight{X}}$.
	Consider a tuple of points $\mathbold{x} \in \volume^{\numParticles}$ and a type assignment $\typeOf\colon [\numParticles] \to [\numTypes]$ produced by applying \Cref{alg:perturbation_algorithm} to $\independentSet$.
	Then $(\mathbold{x}, \typeOf)$ is distributed according to the density $\hcPPGibbsContinuous{\volume}{\hcPPGraph{X}}{\hcPPWeight{X}}{\allocation}$.
\end{observation}
The observation follows directly from writing down the resulting density of~$(\mathbold{x}, \typeOf)$.

In practice, we might not have the ability to exactly sample an independent set of $\hcPPGraph{X}$ according to $\hcGibbs{\hcPPGraph{X}}{\hcPPWeight{X}}$.
The following lemma extends \Cref{obs:sampling_hcPPGibbsContinuous} to the setting where we have access to an approximate sampler for $\hcGibbs{\hcPPGraph{X}}{\hcPPWeight{X}}$.
\begin{lemma}
	\label{lemma:appx_sampling_hcPPGibbsContinuous}
	Let $(\volume, \interactionMatrix, \fugacity)$ be a hard-constraint point process with $\numTypes \in \N_{\ge 1}$ particle types.
	Further, let $X \subset \volume$ be finite and non-empty and, for $\volumeError \in [0, 1)$ and $\distanceError \in \R_{>0}$, let $\allocation$ be a $\volumeError$-$\distanceError$-allocation for~$X$.
	Let $(\hcPPGraph{X}, \hcPPWeight{X})$ be the hard-core representation of $(\volume, \interactionMatrix, \fugacity)$ based on $X$.
	For $\samplingError \in [0, 1]$, let $\independentSet \in \independentSets{\hcPPGraph{X}}$ be drawn according to an $\samplingError$-approximation of $\hcGibbs{\hcPPGraph{X}}{\hcPPWeight{X}}$.
	Consider a tuple of points $\mathbold{x} \in \volume^{\numParticles}$ and a type assignment $\typeOf\colon [\numParticles] \to [\numTypes]$ produced by applying \Cref{alg:perturbation_algorithm} to $\independentSet$.
	Then $(\mathbold{x}, \typeOf)$ is distributed according to an $\samplingError$-approximation of $\hcPPGibbsContinuous{\volume}{\hcPPGraph{X}}{\hcPPWeight{X}}{\allocation}$.
\end{lemma}
\begin{proof}
	Let $\hcGibbsAppx{\hcPPGraph{X}}{\hcPPWeight{X}}$ be the given $\samplingError$-approximation for $\hcGibbs{\hcPPGraph{X}}{\hcPPWeight{X}}$, and let $\hcPPGibbsContinuousAppx{\volume}{\hcPPGraph{X}}{\hcPPWeight{X}}{\allocation}$ be the resulting distribution over configuration $(\mathbold{x}, \typeOf)$.
	We prove our claim by constructing a coupling of $\hcPPGibbsContinuous{\volume}{\hcPPGraph{X}}{\hcPPWeight{X}}{\allocation}$ and $\hcPPGibbsContinuousAppx{\volume}{\hcPPGraph{X}}{\hcPPWeight{X}}{\allocation}$ such that they produce the same $(\mathbold{x}, \typeOf)$ with probability at least $1 - \samplingError$.
	This results in the desired bound on the total-variation distance by the coupling lemma (see \cite[Theorem~$2.4$]{den2012probability}).

	To construct this coupling of $\hcPPGibbsContinuous{\volume}{\hcPPGraph{X}}{\hcPPWeight{X}}{\allocation}$ and $\hcPPGibbsContinuousAppx{\volume}{\hcPPGraph{X}}{\hcPPWeight{X}}{\allocation}$, we start by considering an optimal coupling between $\hcGibbsAppx{\hcPPGraph{X}}{\hcPPWeight{X}}$ and $\hcGibbs{\hcPPGraph{X}}{\hcPPWeight{X}}$ (see \cite[Theorem~$2.12$]{den2012probability}).
	For a tuple $(\independentSet, \independentSet') \in \independentSets{\hcPPGraph{X}} \times \independentSets{\hcPPGraph{X}}$ drawn from that optimal coupling, it holds that $\independentSet = \independentSet'$ with probability $1 - \samplingError$.
	If this is the case, we produce $(\mathbold{x}, \typeOf)$ by applying \Cref{alg:perturbation_algorithm} to $\independentSet = \independentSet'$, and we return the tuple of configurations $((\mathbold{x}, \typeOf), (\mathbold{x}, \typeOf))$.
	However, if we draw $\independentSet \neq \independentSet'$, then we produce corresponding configurations $(\mathbold{x}, \typeOf)$ by applying \Cref{alg:perturbation_algorithm} to $\independentSet$, and $(\mathbold{x'}, \typeOf')$ by applying \Cref{alg:perturbation_algorithm} to $\independentSet'$, and return $((\mathbold{x}, \typeOf), (\mathbold{x'}, \typeOf'))$.

	Note that the description above gives a valid coupling of $\hcPPGibbsContinuous{\volume}{\hcPPGraph{X}}{\hcPPWeight{X}}{\allocation}$ and $\hcPPGibbsContinuousAppx{\volume}{\hcPPGraph{X}}{\hcPPWeight{X}}{\allocation}$, as it preserves the marginal distributions.
	Further, note that for a tuple $((\mathbold{x}, \typeOf), (\mathbold{x'}, \typeOf'))$ produced by this coupling, it holds that $(\mathbold{x}, \typeOf) = (\mathbold{x'}, \typeOf')$ if $\independentSet = \independentSet'$ in the first step.
	This happens with probability at least $1 - \samplingError$, which proves the desired bound on the total-variation distance.
\end{proof}
To state some explicit approximate sampling results while keeping our statements more simple, we restrict ourselves once again to cubic regions $\volume=[0, \sidelength)^{\dimension}$ for some $\sidelength \in \R_{>0}$.
First, we derive a bound on the resolution $\resolution$ that we need for the canonical discretization to obtain a desired total-variation distance between $\generalGibbs{\volume}{\interactionMatrix}{\fugacity}$ and $\hcPPGibbsContinuous{\volume}{\hcPPGraph{X}}{\hcPPWeight{X}}{\allocation}$.
\begin{lemma}
	\label{lemma:canonical_sampling_error}
	Let $(\volume, \interactionMatrix, \fugacity)$ be a hard-constraint point process with $\numTypes \in \N_{\ge 1}$ particle types and $\volume = [0, \sidelength)^{\dimension}$ for some $\sidelength \in \R_{>0}$.
	For all $\samplingError \in (0, 1]$, there exists $\resolution_{\samplingError} \in \bigTheta{\samplingError^{-1/\dimension} \vol{\volume}^{1/\dimension}}$ such that for all feasible resolutions $\resolution \ge \resolution_{\samplingError}$, there is an allocation $\allocation$ for $\canonicalPoints{\resolution}$ such that
	\[
		\dtv{\generalGibbs{\volume}{\interactionMatrix}{\fugacity}}{\hcPPGibbsContinuous{\volume}{\hcPPGraph{\canonicalPoints{\resolution}}}{\hcPPWeight{\canonicalPoints{\resolution}}}{\allocation}}
		\le \samplingError .
		\qedhere
	\]
\end{lemma}
\begin{proof}
	Most of the proof is similar to the proof of \Cref{thm:canonical_discretization_error}.
	Let $\interactionMatrix_{\min} = \min_{i, j \in [\numTypes]} \{\interactionMatrix[i][j] \mid \interactionMatrix[i][j] > 0\}$, let $\fugacity_{\max} = \max_{i \in [\numTypes]} \fugacity[i]$, set
	\[
		\resolution_{\samplingError} = \sqrt{\dimension} \left(\frac{32 \numTypes \max\left\{\fugacity_{\max}, \fugacity_{\max}^2\right\} \vol{\volume}}{\samplingError}\right)^{\frac{1}{\dimension}} \cdot \max\left\{1, \frac{4}{\interactionMatrix_{\min}}\right\},
	\]
	and observe that $\resolution_{\samplingError}$ satisfies the asymptotic assumptions of the statement.

	For all feasible resolutions $\resolution \ge \resolution_{\samplingError}$, let $\allocation$ be an allocation for the canonical point set $\canonicalPoints{\resolution}$ as constructed in the proof of \Cref{lemma:canonical_allocation}.
	This especially means that $\allocation$ is a $\volumeError$-$\distanceError$-allocation for $\canonicalPoints{\resolution}$ with $\volumeError = 0$ and
	\[
		\distanceError
		= \frac{\sqrt{\dimension}}{\resolution}
		\le \left(\frac{\samplingError}{32 \numTypes \max\left\{\fugacity_{\max}, \fugacity_{\max}^2\right\} \vol{\volume}}\right)^{\frac{1}{\dimension}} \cdot \min\left\{1, \frac{\interactionMatrix_{\min}}{4}\right\}.
	\]

	We prove our claim, by combining \Cref{lemma:discretization_error,lemma:dtv_discretization} based on the canonical point set $\canonicalPoints{\resolution}$ and the $\volumeError$-$\distanceError$-allocation $\allocation$.
	First, we check that $\volumeError$, $\distanceError$, and $\size{\canonicalPoints{\resolution}}$ satisfy the conditions of \Cref{thm:discretization_error}.

	It trivially holds that $\volumeError = 0 \in \left[0, \frac{1}{2}\right]$.
	Moreover, we have $\distanceError \in [0, \interactionMatrix_{\min}/2]$ for $\vol{\volume}$ sufficiently large.
	Observe that
	\[
		\size{\canonicalPoints{\resolution}}
		= \resolution^{\dimension} \sidelength^{\dimension}
		\ge \dimension^{\frac{\dimension}{2}} \frac{32 \numTypes \max\left\{\fugacity_{\max}, \fugacity_{\max}^2\right\} \vol{\volume}^2}{\samplingError} \cdot \max\left\{1, \left(\frac{4}{\interactionMatrix_{\min}}\right)^{\dimension}\right\} .
	\]
	Thus, we have $\size{\canonicalPoints{\resolution}} \ge 4 \fugacity_{\max} \vol{\volume}$.

	We proceed by upper-bounding
	\[
		\sum_{\numParticles \in \N} \sum_{\typeOf\colon [\numParticles] \to [\numTypes]} \int_{\volume^{\numParticles}}  \absolute{\generalWeight{\interactionMatrix}{\fugacity}[\mathbold{x}, \typeOf] - \hcPPWeightContinuous{\hcPPGraph{X}}{\hcPPWeight{X}}{\allocation}[\mathbold{x}, \typeOf]} \, \text{d} \lebesgue{\dimension \times \numParticles}
	\]
	via \Cref{thm:discretization_error}.
	Recall that we have $\volumeError = 0$.
	Further, note that
	\[
		\distanceError
		\le \left(\frac{\samplingError}{32 \numTypes \max\left\{\fugacity_{\max}, \fugacity_{\max}^2\right\} \vol{\volume}}\right)^{\frac{1}{\dimension}} \cdot \min\left\{1, \frac{\interactionMatrix_{\min}}{4}\right\}
		\le \left(\frac{\samplingError}{4 \sum_{i \in [\numTypes]} \fugacity[i] \vol{\volume}}\right)^{\frac{1}{\dimension}} \frac{\interactionMatrix_{\min}}{4}
	\]
	and thus
	\[
		\left(\frac{4 \distanceError}{\interactionMatrix_{\min}}\right)^{\dimension} \sum_{i \in [\numTypes]} \fugacity[i] \vol{\volume} \le \frac{\samplingError}{4} .
	\]
	Moreover, for
	\[
		\size{\canonicalPoints{\resolution}}
		\ge \dimension^{\frac{\dimension}{2}} \frac{32 \numTypes \max\left\{\fugacity_{\max}, \fugacity_{\max}^2\right\} \vol{\volume}^2}{\samplingError} \cdot \max\left\{1, \left(\frac{4}{\interactionMatrix_{\min}}\right)^{\dimension}\right\}
		\ge \frac{32 \sum_{i \in [\numTypes]} \fugacity[i]^2 \vol{\volume}^2}{\samplingError} ,
	\]
	it holds that
	\[
		\frac{8}{\size{\canonicalPoints{\resolution}}} \sum_{i \in [\numTypes]} \fugacity[i]^2 \vol{\volume}^2 \le \frac{\samplingError}{4}.
	\]
	By \Cref{thm:discretization_error}, we obtain
	\[
		\sum_{\numParticles \in \N} \sum_{\typeOf\colon [\numParticles] \to [\numTypes]} \int_{\volume^{\numParticles}}  \absolute{\generalWeight{\interactionMatrix}{\fugacity}[\mathbold{x}, \typeOf] - \hcPPWeightContinuous{\hcPPGraph{X}}{\hcPPWeight{X}}{\allocation}[\mathbold{x}, \typeOf]} \, \text{d} \lebesgue{\dimension \times \numParticles}
		\le \left(\eulerE^{\frac{\samplingError}{2}} - 1\right) \generalPartitionFunction{\volume}{\interactionMatrix}{\fugacity},
	\]
	and, by \Cref{lemma:dtv_discretization}, this implies $\dtv{\generalGibbs{\volume}{\interactionMatrix}{\fugacity}}{\hcPPGibbsContinuous{\volume}{\hcPPGraph{X}}{\hcPPWeight{X}}{\allocation}} \le \eulerE^{\frac{\samplingError}{2}} - 1$.

	Finally, observe that, for $\samplingError \in (0, 1]$, it holds that $\eulerE^{\frac{\samplingError}{2}} - 1 \le \samplingError$, which concludes the proof.
\end{proof}

We now state the sampling analogue of \Cref{thm:uniform_approx} and \Cref{thm:non_uniform_approx}.
We start with the case of uniform fugacities.
\uniformSampling
% \begin{theorem}
% 	\label{thm:uniform_sampling}
% 	Suppose a computational model that satisfies \Cref{assumption:floating_pointer}.
% 	Let $(\volume, \interactionMatrix, \fugacity)$ be a hard-constraint point process with $\numTypes \in \N_{\ge 1}$ particle types, $\volume = [0, \sidelength)^{\dimension}$ for some $\sidelength \in \R_{>0}$ and $\fugacity$ being a constant.
% 	Let $\volumeMatrix$ be the corresponding volume exclusion matrix and denote by $\lOneNorm{\volumeMatrix}$ its $L_{1}$-norm.
% 	If
% 	\[
% 		\fugacity < \frac{\eulerE}{\lOneNorm{\volumeMatrix}},
% 	\]
% 	then, for all $\samplingError \in (0, 1]$, there is an $\samplingError$-approximate sampler for $\generalGibbs{\volume}{\interactionMatrix}{\fugacity}$ with running time in $\poly{\frac{\vol{\volume}}{\samplingError}}$.
% \end{theorem}

\begin{proof}
	The rough proof sketch is to choose the resolution $\resolution$ and the canonical point set $\canonicalPoints{\resolution}$ with allocation $\allocation$ that such $\hcPPGibbsContinuous{\volume}{\hcPPGraph{\canonicalPoints{\resolution}}}{\hcPPWeight{\canonicalPoints{\resolution}}}{\allocation}$ is an $\frac{\samplingError}{2}$-approximation of $\generalGibbs{\volume}{\interactionMatrix}{\fugacity}$.
	We then sample $\frac{\samplingError}{2}$-approximately from $\hcGibbs{\hcPPGraph{\canonicalPoints{\resolution}}}{\hcPPWeight{\canonicalPoints{\resolution}}}$ to obtain an $\frac{\samplingError}{2}$-approximation of $\hcPPGibbsContinuous{\volume}{\hcPPGraph{\canonicalPoints{\resolution}}}{\hcPPWeight{\canonicalPoints{\resolution}}}{\allocation}$.
	Obviously, the resulting error is then bounded by $\frac{\samplingError}{2}+\frac{\samplingError}{2} = \samplingError$.

	By \Cref{lemma:canonical_sampling_error}, there is a resolution $\resolution \in \bigTheta{\samplingError^{-1/\dimension}\vol{\volume}^{1/\dimension}}$ such that
	\[
		\dtv{\generalGibbs{\volume}{\interactionMatrix}{\fugacity}}{\hcPPGibbsContinuous{\volume}{\hcPPGraph{\canonicalPoints{\resolution}}}{\hcPPWeight{\canonicalPoints{\resolution}}}{\allocation}}
		\le \frac{\samplingError}{2} ,
	\]
	where, as described in \Cref{lemma:canonical_allocation},
	\[
		\allocation[x] = \prod_{i \in [\dimension]} \left[x^{(i)}, x^{(i)} + \frac{1}{\resolution}\right) .
	\]
	Note that the corresponding graph $\hcPPGraph{\canonicalPoints{X}}$ has $\bigTheta{\frac{\vol{\volume}}{\samplingError}}$ vertices.

	Further, analogously to the proof of \Cref{thm:uniform_approx}, if $\fugacity < \frac{\eulerE}{\lOneNorm{\volumeMatrix}}$, constructing the hard-core representation $(\hcPPGraph{\canonicalPoints{\resolution}}, \hcPPWeight{\canonicalPoints{\resolution}})$ can be done in such a way that $\hcPPWeight{\canonicalPoints{\resolution}} < \criticalWeight{\hcPPDegree{\canonicalPoints{\resolution}}^{*}}$.
	Thus, by \Cref{thm:hc_sampling_univariate}, we can sample $\frac{\samplingError}{2}$-approximately from $\hcGibbs{\hcPPGraph{\canonicalPoints{\resolution}}}{\hcPPWeight{\canonicalPoints{\resolution}}}$ in time $\poly{\frac{\size{\hcPPVertices{\canonicalPoints{\resolution}}}}{\samplingError}}[\Big]$, which is also in $\poly{\frac{\vol{\volume}}{\samplingError}}$.

	We now obtain our sampler by drawing an independent set $\independentSet \in \independentSets{\hcPPGraph{\canonicalPoints{\resolution}}}$ according to this $\frac{\samplingError}{2}$-approximation of $\hcGibbs{\hcPPGraph{\canonicalPoints{\resolution}}}{\hcPPWeight{\canonicalPoints{\resolution}}}$ and applying \Cref{alg:perturbation_algorithm} to it.
	To this end, note that for all $y \in \canonicalPoints{\resolution}$, a uniform sample $x \in \invAllocation{y}$ is obtained by generating, for each $i \in [\dimension]$, a uniformly random floating-point number $z^{(i)} \in \left[0, \frac{1}{\resolution}\right)$ and setting $x = y + (z^{(i)})_{i \in [\dimension]}$, where the addition is component-wise.
	\Cref{lemma:appx_sampling_hcPPGibbsContinuous} yields the desired bound on the total-variation distance.
\end{proof}

The sampling analogues of \Cref{cor:monoatomic_hs,cor:uniform_wr} follow immediately.

\begin{corollary}
	\label{cor:monoatomic_hs_sampling}
	Suppose a computational model that satisfies \Cref{assumption:floating_pointer}.
	Let $\volume = [0, \sidelength)^{\dimension}$ for some $\sidelength \in \R_{>0}$.
	Further, let $\radius \in \R_{>0}$ and $\fugacity \in \R_{\ge 0}$.
	Denote by $\hsGibbs{\volume}{\radius}{\fugacity}$ the Gibbs distribution of the hard-sphere model on $\volume$ with particles of radius $\radius$ and fugacity~$\fugacity$.
	If
	\[
		\fugacity < \frac{\eulerE}{2^{\dimension} \vol{\ball{\radius}}},
	\]
	then, for all $\samplingError \in (0, 1]$, there is an $\samplingError$-approximate sampler for $\hsGibbs{\volume}{\radius}{\fugacity}$ with running time in $\poly{\frac{\vol{\volume}}{\samplingError}}[\big]$.
\end{corollary}

\begin{corollary}
	\label{cor:uniform_wr_sampling}
	Suppose a computational model that satisfies \Cref{assumption:floating_pointer}.
	Let $\volume = [0, \sidelength)^{\dimension}$ for some $\sidelength \in \R_{>0}$.
	Further, let $\radius \in \R_{>0}$, $\fugacity \in \R_{\ge 0}$ and $\numTypes \in \N_{\ge 1}$.
	Denote by $\wrGibbs{\volume}{\radius}{\fugacity}$ the Gibbs distribution of the Widom--Rowlinson model on $\volume$ with $\numTypes$ particle types, each of radius $\radius$ and fugacity $\fugacity$.
	If
	\[
		\fugacity < \frac{\eulerE}{(\numTypes - 1) 2^{\dimension} \vol{\ball{\radius}}},
	\]
	then, for all $\samplingError \in (0, 1]$, there is an $\samplingError$-approximate sampler for $\wrGibbs{\volume}{\radius}{\fugacity}$ with running time in $\poly{\frac{\vol{\volume}}{\samplingError}}[\big]$.
\end{corollary}

\begin{remark}
	\label{remark:univariate_to_multivariate_sampling}
	Analogously to \Cref{thm:uniform_approx}, letting $\fugacity_{\max} = \max_{i \in [\numTypes]} \fugacity[i]$, \Cref{thm:uniform_sampling} extends to the non-uniform setting when the condition is replaced by
	\[
    	\fugacity_{\max} < \frac{\eulerE}{\lOneNorm{\volumeMatrix}} .
        \qedhere
	\]
\end{remark}

Just as for \Cref{thm:uniform_approx}, we also state a sampling version of \Cref{thm:non_uniform_approx} under \Cref{assumption:floating_pointer}.
\nonUniformSampling
% \begin{theorem}
% 	\label{thm:non_uniform_sampling}
% 	Suppose a computational model that satisfies \Cref{assumption:floating_pointer}.
% 	Let $(\volume, \interactionMatrix, \fugacity)$ be a hard-constraint point process with $\numTypes \in \N_{\ge 1}$ particle types and $\volume = [0, \sidelength)^{\dimension}$ for some $\sidelength \in \R_{>0}$.
% 	Further, let $\volumeMatrix$ be the corresponding volume exclusion matrix.
% 	If there is a function $f: [\numTypes] \to \R_{>0}$ such that for all $i \in [\numTypes]$ it holds that
% 	\[
% 		f(i) > \sum_{j \in [\numTypes]} \volumeMatrix[i][j] f(j) \fugacity[j] ,
% 	\]
% 	then, for all $\samplingError \in (0, 1]$, there is an $\samplingError$-approximate sampler for $\generalGibbs{\volume}{\interactionMatrix}{\fugacity}$ with running time in $\poly{\frac{\vol{\volume}}{\approxError}, \ln\left(\frac{\max_{i \in [\numTypes]} f(i)}{\min_{i \in [\numTypes]} f(i)}\right)}$.
% \end{theorem}

\begin{proof}
	As in the proof of \Cref{thm:uniform_sampling}, we start by using \Cref{lemma:canonical_sampling_error}, by which there exists a resolution $\resolution \in \bigTheta{\samplingError^{-1/\dimension}\vol{\volume}^{1/\dimension}}[\big]$ such that
	\[
		\dtv{\generalGibbs{\volume}{\interactionMatrix}{\fugacity}}{\hcPPGibbsContinuous{\volume}{\hcPPGraph{\canonicalPoints{\resolution}}}{\hcPPWeight{\canonicalPoints{\resolution}}}{\allocation}}
		\le \frac{\samplingError}{2} ,
	\]
	where again
	\[
		\allocation[x] = \prod_{i \in [\dimension]} \left[x^{(i)}, x^{(i)} + \frac{1}{\resolution}\right).
	\]
	The corresponding graph $\hcPPGraph{\canonicalPoints{X}}$ consists of $\bigTheta{\frac{\vol{\volume}}{\samplingError}}[\big]$ vertices.

	Analogously to the proof of \Cref{thm:non_uniform_approx}, there is a function $f\colon [\numTypes] \to \R_{>0}$ such that for all $i \in [\numTypes]$, it holds that $f(i) > \sum_{j \in [\numTypes]} \volumeMatrix[i][j] f(j) \fugacity[j]$, which implies that there is also a function $g\colon \hcPPVertices{\canonicalPoints{\resolution}} \to \R_{>0}$ that satisfies the condition of \Cref{thm:clique_dynamics_condition_sampling}.
	Thus, we can sample $\frac{\samplingError}{2}$-approximately from $\hcGibbs{\hcPPGraph{\canonicalPoints{\resolution}}}{\hcPPWeight{\canonicalPoints{\resolution}}}$ in the desired running time.
	Similarly to the proof of \Cref{thm:uniform_sampling}, applying \Cref{alg:perturbation_algorithm} to independent sets drawn from this $\frac{\samplingError}{2}$-approximation of $\hcGibbs{\hcPPGraph{\canonicalPoints{\resolution}}}{\hcPPWeight{\canonicalPoints{\resolution}}}$ yields the desired sampler, and by \Cref{lemma:appx_sampling_hcPPGibbsContinuous} we obtain the required bound on the total-variation distance.
\end{proof}

The sampling version of \Cref{cor:unbalanced_WR} follows.

\begin{corollary}
	\label{cor:unbalanced_WR_sampling}
	Suppose a computational model that satisfies \Cref{assumption:floating_pointer}.
	Let $\volume = [0, \sidelength)^{\dimension}$ for some $\sidelength \in \R_{>0}$.
	Further, let $\radius \in \R_{>0}$, $\fugacity_1, \fugacity_2 \in \R_{\ge 0}$.
	Denote by $\wrGibbs{\volume}{\radius}{\fugacity_1, \fugacity_2}$ the Gibbs distribution of the Widom--Rowlinson model on $\volume$ with $2$ particle types, both with the same radius $\radius$ but (possibly) different fugacities $\fugacity_1, \fugacity_2$.
	If
	\[
		\fugacity_1 \fugacity_2 < \frac{1}{4^{\dimension} \vol{\ball{\radius}}^{2}},
	\]
	then, for all $\samplingError \in (0, 1]$, there is an $\samplingError$-approximate sampler for $\wrGibbs{\volume}{\radius}{\fugacity_1, \fugacity_2}$ with running time in $\poly{\frac{\vol{\volume}}{\samplingError}}[\big]$.
\end{corollary}

Note that the role of the allocation $\allocation$ in our sampling algorithms is fundamentally different than in our approximations for the partition function.
For approximating the partition function, it is sufficient to know that a suitable allocation for a given point set $X$ exists. %to make sure that $\hcPartitionFunction{\hcPPGraph{X}}{\hcPPWeight{X}}$ approximates $\generalPartitionFunction{\volume}{\interactionMatrix}{\fugacity}$.
In contrast, \Cref{lemma:appx_sampling_hcPPGibbsContinuous}, which is a core component of our sampling method, actually requires an explicit representation of~$\allocation$ in order to sample for each $x \in X$ from $\invAllocation[x]$.
This raises the question whether there is a way to circumvent this.

\section{Concentration of random discretizations}\label{sec:random_discretization}
In this section we investigate the behavior of discretizations of a hard-constraint point process $(\volume, \interactionMatrix, \fugacity)$ based on uniformly random finite point sets $X \subseteq \volume$.
We show sufficient conditions for the partition function of such random discretizations to get concentrated around the partition function of the continuous point process for a polynomial number of points.

The first part of this section focuses on how such a concentration result can be obtained from \Cref{thm:discretization_error}.
Consequently, such an approach relies on star-convexity.

In the second part of the section, we show that a slightly weaker concentration result can be obtained without assuming star-convexity. 
To this end, we first argue that the discrete partition function concentrates from below and use this to derive concentration from above.

In both cases, our sufficient condition for the concentration is closely related to the concept of a specific partitioning of $\volume$.

\edpartitioning

the following lemma relates the notion of a $\volumeScaling$-$\distanceError$-partitioning with our previously introduced concept of a $\volumeError$-$\distanceError$-allocation.

\begin{lemma}\label{lemma:partitioning_allocation}
	Let $\volume \subset \R^{\dimension}$ be bounded and measurable with $\vol{\volume} > 0$.
	Let $\volumeScaling_{1} \in (0, 1]$ and $\distanceError \in \R_{>0}$ and assume there is a $\volumeScaling_{1}$-$\distanceError$-partitioning of $\volume$ with size $\numPartitions$.
	Further, let $\volumeError_2 \in (0, 1]$, $\allocationProb \in (0, 1]$, and draw $X \subset \volume$ with $\size{X} = \numPoints \ge 48 \left(\frac{1}{\volumeError_{2}}\right)^2 \frac{1}{\volumeScaling_{1}} \numPartitions \ln \left(\frac{2 \numPartitions}{\allocationProb}\right)$ uniformly at random.
	With probability at least $1 - \allocationProb$ there is a $\volumeError_{2}$-$\distanceError$-allocation for $X$.
\end{lemma}

\begin{proof}
	Let $(\volume_{i})_{i \in [\numPartitions]}$ be a $\volumeScaling_1$-$\distanceError$-partitioning of $\volume$ with size $\numPartitions$.
	Further, let $X \subset \volume$ with $\numPoints = 48 \left(\frac{1}{\volumeError_2}\right)^2 \frac{1}{\volumeScaling_1} \numPartitions \ln \left(\frac{2 \numPartitions}{\allocationProb}\right)$ be drawn uniformly at random.

	The main idea of the proof is to show that with probability at least $1 - \allocationProb$ it holds that for all partitions $\volume_{i}$ for $i \in [\numPartitions]$ the number of points from $X$ that are in $\volume_{i}$ is close to $\frac{\vol{\volume_{i}}}{\vol{\volume}} \cdot \size{X}$, and especially all partitions contain at least one point from $X$.
	If this is the case, we can construct an allocation by fairly allocating each partition $\volume_{i}$ to the points in $X \cap \volume_{i}$ while ignoring any distance constraints.
	As we argue below, this always results in an allocation with the desired properties.

	For all $i \in [\numPartitions]$, let $X_i = X \cap \volume_i$ and let $Y_i = \size{X_i}$.
	We now show that, with probability at least $1 - \allocationProb$, it holds that
	\begin{align} \label[ineq]{eq:partitioning_allocation:Y_i}
		\left(1 - \frac{\volumeError_2}{4}\right) \E{Y_i}
		\le Y_i
		\le \left(1 + \frac{\volumeError_2}{4}\right) \E{Y_i}
	\end{align}
	for all $i \in [\numPartitions]$.
	To prove this claim, note that for each $i \in [\numPartitions]$ the random variable $Y_i$ follows a binomial distribution with $\numPoints$ trials and success probability $\frac{\vol{\volume_i}}{\vol{\volume}}$.
	Applying Chernoff's inequality yields
	\[
		\Pr{\absolute{Y_i - \E{Y_i}} \ge \frac{\volumeError_2}{4}} \le 2 \cdot \eulerE^{- \frac{\volumeError_2^2}{48} \E{Y_i}} .
	\]
	Further, note that we have $\E{Y_i} = \numPoints \frac{\vol{\volume_i}}{\vol{\volume}} \ge \volumeScaling_1 \frac{\numPoints}{\numPartitions}$, which gives us
	\[
		\Pr{\absolute{Y_i - \E{Y_i}} \ge \frac{\volumeError_2}{4}} \le 2 \cdot \eulerE^{- \frac{\volumeError_2^2}{48} \volumeScaling_1 \frac{\numPoints}{\numPartitions}} .
	\]
	By our choice of $\numPoints$, we obtain
	\[
		\Pr{\absolute{Y_i - \E{Y_i}} \ge \frac{\volumeError_2}{4}} \le \frac{\allocationProb}{\numPartitions}.
	\]
	Finally, applying union bound over $i \in [\numPartitions]$ proves that \cref{eq:partitioning_allocation:Y_i} fails with probability at most $\allocationProb$.

	Now, let us assume \cref{eq:partitioning_allocation:Y_i} is satisfied.
	Note that this especially implies that $Y_i > 0$.
	We now construct our allocation $\allocation$ as follows.
	For each $i \in [\numPartitions]$, partition $\volume_i$ into sets $\left(\volume_i^{(j)}\right)_{j \in [Y_i]}$ such that all $\volume_i^{(j)}$ are measurable and have the same volume.
	Note that such a division of $\volume_i$ always exists.
	Now, for each such $\volume_i^{(j)}$, fix exactly one $x \in X_i$ and set $\allocation[y] = x$ for all $y \in \volume_i^{(j)}$.
	It remains to argue that $\allocation$ is a $\volumeError_2$-$\distanceError$-allocation.
	To this end, choose some $x \in X$ and assume that $x \in X_i$ for some $i \in [\numPartitions]$.
	Further, let $y \in \invAllocation[x]$ and observe that by construction both, $x$ and $y$, are in $\volume_i$.
	Thus, we have
	\[
		\dist{x}{y} \le \distanceError.
	\]
	Moreover, because $\E{Y_i} = \numPoints \frac{\vol{\volume_i}}{\vol{\volume}}$ we have that $\vol{\volume_i} = \E{Y_i} \frac{\vol{\volume}}{\numPoints}$.
	Thus, by \cref{eq:partitioning_allocation:Y_i}, we obtain
	\[
		\vol{\invAllocation[x]}
		= \frac{\vol{\volume_i}}{Y_i}
		= \frac{\E{Y_i}}{Y_i} \cdot \frac{\vol{\volume}}{\numPoints}
		\le \frac{1}{1 - \volumeError_2 / 4} \cdot  \frac{\vol{\volume}}{\numPoints}
		\le \left(1 + \volumeError_2\right) \frac{\vol{\volume}}{\numPoints}
	\]
	and analogously
	\[
		\vol{\invAllocation[x]}
		\ge \frac{1}{1 + \volumeError_2 / 4} \cdot \frac{\vol{\volume}}{\numPoints}
		\ge \left(1 - \volumeError_2\right) \frac{\vol{\volume}}{\numPartitions} ,
	\]
	which concludes the proof.
\end{proof}

\subsection{Concentration with star-convexity} \label{subsec:concentration_star_convex}
Assuming that the considered region is star-convex, we can combine \Cref{thm:discretization_error,lemma:partitioning_allocation} to obtain the following probabilistic bound on the difference between the partition function of a hard-constraint point process and the hard-core partition function of a random discretization.

\begin{corollary}
	\label{cor:random_discretization_error}
	Let $(\volume, \interactionMatrix, \fugacity)$ be a hard-constraint point process with $\numTypes \in \N_{\ge 1}$ particle types, and assume $\volume \subset \R^{\dimension}$ is star-convex.
	Further, set $\fugacity_{\max} = \max_{i \in [\numTypes]} \fugacity[i]$ and set $\interactionMatrix_{\min} = \min_{i, j \in [\numTypes]} \{\interactionMatrix[i][j] \mid \interactionMatrix[i][j] > 0\}$.
	Let $\distanceError \in (0, \interactionMatrix_{\min}/2]$ and $\volumeScaling_1 \in (0, 1]$ and assume there is a $\volumeScaling_{1}$-$\distanceError$-partitioning of $\volume$ size $\numPartitions$.
	For any $\volumeError_2 \in \left(0, \frac{1}{2}\right]$, $\allocationProb \in (0, 1]$ and $\numPoints \in \N$ with
	\[
		\max\left\{4 \fugacity_{\max} \vol{\volume}, 48 \left(\frac{1}{\volumeError_{2}}\right)^2 \frac{1}{\volumeScaling_{1}} \numPartitions \ln \left(\frac{2 \numPartitions}{\allocationProb}\right)\right\}
		\le \numPoints
	\]
	draw $X \subset \volume$ with $\size{X} = \numPoints$ uniformly at random and let $(\hcPPGraph{X}, \hcPPWeight{X})$ be the corresponding discretization of $(\volume, \interactionMatrix, \fugacity)$.
	It holds that
	\[
		\absolute{\hcPartitionFunction{\hcPPGraph{X}}{\hcPPWeight{X}} - \generalPartitionFunction{\volume}{\interactionMatrix}{\fugacity}}
		\le \left(\eulerE^{\frac{8}{\size{X}} \sum_{i \in [\numTypes]} \fugacity[i]^2 \vol{\volume}^2} \eulerE^{\left( 2 \volumeError_2 + \left(\frac{4 \distanceError }{\interactionMatrix_{\min}} \right)^{\dimension} \right) \sum_{i \in [\numTypes]} \fugacity[i] \vol{\volume}} - 1\right) \generalPartitionFunction{\volume}{\interactionMatrix}{\fugacity}
	\]
	with probability at least $1-\allocationProb$.
\end{corollary}

\Cref{cor:random_discretization_error} follows immediately from \Cref{thm:discretization_error} and \Cref{lemma:partitioning_allocation}.
Note that \Cref{cor:random_discretization_error} means that, if we can construct $\volumeScaling_1$-$\distanceError$-partitionings for some constant $\volumeScaling_1 > 0$ and arbitrarily small $\distanceError > 0$, we can use this to derive a concentration result for the hard-core partition function of discretizations based on random point sets.
A general statement for this is given in the following corollary.
\randomdiscretizationconcentration

A more precise statement can be obtained when looking at specific regions of $\dimension$-dimensional Euclidean space.
To demonstrate this, we once again consider cubic regions $\volume = [0, \sidelength)^{\dimension}$ for some $\sidelength \in \R_{>0}$.
For such regions, we have the following observation.

\begin{lemma}
	\label{lemma:hypercube_partitioning}
	Let $\volume = [0, \sidelength)^{\dimension}$ for some $\sidelength \in \R_{>0}$.
	For all $\distanceError \in \R_{>0}$ there is a $1$-$\distanceError$-partitioning of $\volume$ of size $\left\lceil \sqrt{\dimension} \frac{\sidelength}{\distanceError} \right\rceil^{\dimension}$.
\end{lemma}

\begin{proof}
	Set $a = \frac{\distanceError}{\sqrt{\dimension}}$ and divide $\volume$ into disjoint hypercubes of side length $a' = \sidelength \left\lceil \frac{\sidelength}{a} \right\rceil^{-1}$.
	Note that $\frac{\sidelength}{a'} \in \N$.
	Formally, this gives the partitioning
	\[
		\volume_{i_1, \dots, i_{\dimension}} = \prod_{j \in [\dimension]} \left[i_j a', (i_j + 1) a'\right)
	\]
	for $(i_1, \dots, i_{\dimension}) \in \left\{0, \dots, \frac{\sidelength}{a'}-1\right\}^{\dimension}$.
	Obviously, this results in a partitioning of size $\left\lceil \sqrt{\dimension} \frac{\sidelength}{\distanceError} \right\rceil^{\dimension}$.
	Further, each partition has the same volume and their diameter is upper bounded by $\sqrt{\dimension} a' \le \sqrt{\dimension} a = \distanceError$, which proves that the partitioning is a $1$-$\distanceError$-partitioning.
\end{proof}

\begin{remark}
	Note that the partitioning constructed in the proof of \Cref{lemma:hypercube_partitioning} can in fact be generated by taking the canonical point set $\canonicalPoints{\resolution}$ for the smallest feasible resolution $\resolution \ge \frac{\sqrt{\dimension}}{\distanceError}$ and taking the pre-image of the $0$-$\distanceError$-allocation $\allocation$, constructed in the proof of \Cref{lemma:canonical_allocation}, for each point in $\canonicalPoints{\resolution}$.
	This is not a coincidence.
	In fact, for every point set $X$ with a $\volumeError$-$\distanceError$-allocation we can construct a $(1-\volumeError)$-$2\distanceError$-partitioning, which, together with \Cref{lemma:hypercube_partitioning}, shows an interesting equivalence between both concepts.
	The results in this section could as well be stated in terms of allocations.
	However, as we are not interested in discretizations for specific point sets, it is more natural to state our results in terms of partitionings.
\end{remark}

\begin{corollary}
	\label{cor:hypercube_random_discretization_concentration}
	Let $(\volume, \interactionMatrix, \fugacity)$ be a hard-constraint point process with $\numTypes \in \N_{\ge 1}$ particle types and $\volume = [0, \sidelength)^{\dimension}$ for some $\sidelength \in \R_{>0}$. Let $X \subset \volume$ with $\size{X} = \numPoints$, chosen uniformly at random.
	For all $\discretizationError \in (0, 1]$ and $\allocationProb \in (0, 1]$, there exists an $\numPoints_{\discretizationError, \allocationProb} \in \bigTheta{\vol{\volume}^{4}\discretizationError^{-3} \ln \big(\frac{\vol{\volume}}{\discretizationError \allocationProb}\big)}[\big]$ such that for all $\numPoints \ge \numPoints_{\discretizationError, \allocationProb}$, with probability at least $1-\allocationProb$ it holds that
	$
	\eulerE^{-\discretizationError}\generalPartitionFunction{\volume}{\interactionMatrix}{\fugacity}
	\le \hcPartitionFunction{\hcPPGraph{X}}{\hcPPWeight{X}}
	\le \eulerE^{\discretizationError}\generalPartitionFunction{\volume}{\interactionMatrix}{\fugacity}
	$.
\end{corollary}
% \begin{corollary}
% 	\label{cor:hypercube_random_discretization_concentration}
% 	Let $(\volume, \interactionMatrix, \fugacity)$ be a hard-constraint point process with $\numTypes \in \N_{\ge 1}$ particle types and $\volume = [0, \sidelength)^{\dimension}$ for some $\sidelength \in \R_{>0}$.
% 	For all $\discretizationError \in (0, 1]$ and $\allocationProb \in (0, 1]$ there exists $\numPoints_{\discretizationError, \allocationProb} \in \bigTheta{\vol{\volume}^{4}\discretizationError^{-3} \ln \left(\frac{\vol{\volume}}{\discretizationError \allocationProb}\right)}$ such that for all $\numPoints \ge \numPoints_{\discretizationError, \allocationProb}$ it holds for $X \subset \volume$ with $\size{X} = \numPoints$ uniformly at random that
% 	\[
% 		\eulerE^{-\discretizationError}\generalPartitionFunction{\volume}{\interactionMatrix}{\fugacity}
% 		\le \hcPartitionFunction{\hcPPGraph{X}}{\hcPPWeight{X}}
% 		\le \eulerE^{\discretizationError}\generalPartitionFunction{\volume}{\interactionMatrix}{\fugacity}
% 	\]
% 	with probability at least $1-\allocationProb$.
% \end{corollary}

\begin{proof}
	Set $\fugacity_{\max} = \max_{i \in [\numTypes]} \fugacity[i]$ and $\interactionMatrix_{\min} = \min_{i,j \in [\numTypes]} \{\interactionMatrix[i][j] \mid \interactionMatrix[i][j] > 0\}$.
	Note that by \Cref{lemma:hypercube_partitioning} we know that for
	\[
		\numPartitions = \left\lceil \sqrt{\dimension} \sidelength \left(\frac{48 \numTypes \max\left\{\fugacity_{\max}, \fugacity_{\max}^2\right\} \vol{\volume}}{\discretizationError}\right)^{\frac{1}{\dimension}} \max\left\{1, \frac{4}{\interactionMatrix_{\min}}\right\} \right\rceil^{\dimension}
	\]
	there is a $1$-$\distanceError$-partitioning of $\volume$ with size $\numPartitions$ for some
	\[
		\distanceError \le \left(\frac{\discretizationError}{48 \numTypes \max\left\{\fugacity_{\max}, \fugacity_{\max}^2\right\} \vol{\volume}}\right)^{\frac{1}{\dimension}} \min\left\{1, \frac{\interactionMatrix_{\min}}{4}\right\} .
	\]

	Set
	\[
		\volumeError = \frac{\discretizationError}{8 \numTypes \fugacity_{\max} \vol{\volume}}
	\]
	and let
	\[
		\numPoints \ge \left(\frac{1}{\volumeError}\right)^{2} \numPartitions \ln \left(\frac{2 \numPartitions}{\allocation}\right) \in \bigTheta{\frac{\vol{\volume}^{4}}{\discretizationError^{3}} \ln \left(\frac{\vol{\volume}}{\discretizationError \allocationProb}\right)} .
	\]
	Due to \Cref{lemma:partitioning_allocation} we know that for $X \subset \volume$ with $\size{X} = \numPoints$ uniformly at random there is a $\volumeError$-$\distanceError$-allocation for $X$ with probability at least $1 - \allocationProb$.

	The rest of the proof is mostly similar to the proof of \Cref{thm:canonical_discretization_error}.
	Let $X \subset \volume$ with $\size{X} = \numPoints$ such that a $\volumeError$-$\distanceError$-allocation for $X$ exists.
	For $\vol{\volume}$ sufficiently large we have that $\distanceError \le \frac{\interactionMatrix_{\min}}{2}$, $\volumeError \le \frac{1}{2}$ and $\size{X} = \numPoints \ge 4 \fugacity_{\max} \vol{\volume}$.
	Thus, we can apply \Cref{thm:discretization_error} to $X$.
	Now, observe that
	\begin{align*}
		\frac{8}{\size{X}} \sum_{i \in [\numTypes]} \fugacity[i]^2 \vol{\volume} \le \frac{\discretizationError}{6} \\
		\left(\frac{\distanceError}{4 \interactionMatrix_{\min}}\right)^{\dimension} \sum_{i \in [\numTypes]} \fugacity[i] \vol{\volume} \le \frac{\discretizationError}{12} \\
		2 \volumeError \sum_{i \in [\numTypes]} \fugacity[i] \vol{\volume} \le \frac{\discretizationError}{12} .
	\end{align*}
	This yields
	\[
		\absolute{\hcPartitionFunction{\hcPPGraph{X}}{\hcPPWeight{X}} - \generalPartitionFunction{\volume}{\interactionMatrix}{\fugacity}}
		\le \left(\eulerE^{\frac{\discretizationError}{3}} - 1\right) \generalPartitionFunction{\volume}{\interactionMatrix}{\fugacity}
	\]
	for such point sets $X$.
	Now, note that $\eulerE^{\frac{\discretizationError}{3}} \le \eulerE^{\discretizationError}$ and $2 - \eulerE^{\frac{\discretizationError}{3}} \ge \eulerE^{-\discretizationError}$ for $\discretizationError \in (0, 1]$.
	Therefore, we have
	\[
		\eulerE^{-\discretizationError} \generalPartitionFunction{\volume}{\interactionMatrix}{\fugacity}
		\le \hcPartitionFunction{\hcPPGraph{X}}{\hcPPWeight{X}}
		\le \eulerE^{\discretizationError} \generalPartitionFunction{\volume}{\interactionMatrix}{\fugacity},
	\]
	which concludes the proof.
\end{proof}

Although \Cref{cor:hypercube_random_discretization_concentration} does not directly yield additional algorithmic results beyond what was already proved in \Cref{sec:algo}, it shows an interesting relationship between the hard-core model on random geometric graph structures and continuous hard-constraint point process.
An immediate consequence of \Cref{cor:hypercube_random_discretization_concentration} is that the hard-core partition functions of uniformly random geometric graphs on $\volume = [0, \sidelength)^{\dimension}$ with edge connection threshold $2 \radius$ for appropriately chosen fugacities concentrates around the partition function of a hard-sphere model with particle radius $\radius$ on $\volume$ as the number of vertices is increased.
We believe that this connection between hard-core models on random geometric graphs and the continuous hard-sphere model is of its own interest.

\subsection{Concentration without star-convexity}
All results considered so far used the assumption that the region $\volume$ is star-convex.
This is essentially due to \Cref{lemma:discretization_error:1}, which bounds the error introduced by mapping each point according to an allocation $\allocation$.
However, it turns out that a slightly weaker version of the concentration result holds without the assumption of star-convexity.

The core idea is to argue that $\generalPartitionFunction{\volume}{\interactionMatrix}{\fugacity}$ is not significantly larger than $\allocationPartitionFunction{\volume}{\interactionMatrix}{\fugacity}{\allocation}$,  instead of arguing that they are close in general as done in \Cref{lemma:discretization_error:1}.
This argument does not require star-convexity.
By randomizing the set of points, we obtain sufficient conditions for $\hcPartitionFunction{\hcPPGraph{X}}{\hcPPWeight{X}}$ to concentrate from below.
The following general lemma then establishes that the partition function of the hard-core representation also concentrates from above.
\modifiedmarkov

\begin{proof}
	Let $A$ denote the event that $\randomVariable \ge (1-\genericError) \E{\randomVariable}$ and set $\modifiedRandomVariable = \ind{A} (\randomVariable - (1 - \genericError) \E{\randomVariable})$.
	Using the fact that $\randomVariable$ is non-negative and $\Pr{A} \ge 1 - \genericFailingProb$ we obtain
	\begin{align*}
		\E{\modifiedRandomVariable}
		&= \E{\ind{A} \randomVariable} - (1 - \genericError) \E{\randomVariable} \Pr{A} \\
		&\le \E{\randomVariable} - (1 - \genericFailingProb)(1 - \genericError) \E{\randomVariable} \\
		&= (\genericFailingProb + \genericError - \genericFailingProb \genericError) \E{\randomVariable} 
	\end{align*}
	Further, observe that $\modifiedRandomVariable$ is non-negative and that $\randomVariable \ge (1 + c \genericError) \E{\randomVariable}$ if and only if $\modifiedRandomVariable \ge (c+1) \genericError \E{\randomVariable}$.
	By using Markov's inequality, we get
	\begin{align*}
		\Pr{\randomVariable \ge (1 + c \genericError) \E{\randomVariable}}
		&= \Pr{\modifiedRandomVariable \ge (c+1) \genericError \E{\randomVariable}} \\
		&\le \frac{\E{\modifiedRandomVariable}}{(c+1) \genericError \E{\randomVariable}} \\
		&\le \frac{(\genericFailingProb + \genericError - \genericFailingProb \genericError) \E{\randomVariable}}{(c+1) \genericError \E{\randomVariable}} \\
		&= \frac{1}{c+1} \left(1 + \frac{\genericFailingProb (1 - \genericError)}{\genericError}\right) ,
	\end{align*}
	which concludes the proof.
\end{proof}

Next, given a suitable allocation $\allocation$, we show that $\allocationPartitionFunction{\volume}{\interactionMatrix}{\fugacity}{\allocation}$ is not significantly smaller than $\generalPartitionFunction{\volume}{\interactionMatrix}{\fugacity}$.

\begin{lemma}
	\label{lemma:lower_bound}
	Let $(\volume, \interactionMatrix, \fugacity)$ be a hard-constraint point process with $\numTypes \in \N_{\ge 1}$ particle types.
	Further, set $\interactionMatrix_{\max} = \max_{i,j \in [\numTypes]} \{\interactionMatrix[i][j]\}$, and set $\fugacity_{\max} = \max_{i \in [\numTypes]} \fugacity[i]$.
	Let $\volumeError \in \left[0, 1\right]$ and $\distanceError \in \left[0, \frac{1}{2}\right]$ and let $X \subseteq \volume$ such that there is a $\volumeError$-$\distanceError$-allocation $\allocation$ for $X$.
	Then
	\[
	\sum_{\numParticles \in \N} \sum_{\typeOf\colon [\numParticles] \to [\numTypes]} \int_{\volume^{\numParticles}} \left(\generalWeight{\interactionMatrix}{\fugacity}[\mathbold{x}, \typeOf] - \generalWeight{\interactionMatrix}{\fugacity}[\allocation[\mathbold{x}], \typeOf]\right) \, \text{d} \lebesgue{\dimension \times \numParticles}
	\le \distanceError \numTypes \fugacity_{\max}^2 \vol{\volume} \vol{\ball{\interactionMatrix_{\max} + 1}} \generalPartitionFunction{\volume}{\interactionMatrix}{\fugacity} .
	\qedhere
	\]
\end{lemma}

\begin{proof}
	By definition, we have
	\begin{align}
		\notag
		\sum_{\numParticles \in \N} \sum_{\typeOf\colon [\numParticles] \to [\numTypes]} \int_{\volume^{\numParticles}} &\left(\generalWeight{\interactionMatrix}{\fugacity}[\mathbold{x}, \typeOf] - \generalWeight{\interactionMatrix}{\fugacity}[\allocation[\mathbold{x}], \typeOf]\right) \, \text{d} \lebesgue{\dimension \times \numParticles} \\
		\notag
		&= \sum_{\numParticles \in \N} \frac{1}{\numParticles!} \sum_{\typeOf\colon [\numParticles] \to [\numTypes]}
		\left(\prod_{i \in [\numParticles]} \fugacity[\typeOf[i]] \right) \int_{\volume^{\numParticles}} \left( \valid{\typeOf}{\interactionMatrix}[\mathbold{x}] - \valid{\typeOf}{\interactionMatrix}[\allocation[\mathbold{x}]] \right) \, \text{d} \lebesgue{\dimension \times \numParticles} \\
		\label{lemma:lower_bound:eq_1}
		&= \sum_{\numParticles \ge 2} \sum_{\typeOf\colon [\numParticles] \to [\numTypes]} \frac{1}{\numParticles!}
		\left(\prod_{i \in [\numParticles]} \fugacity[\typeOf[i]] \right) \int_{\volume^{\numParticles}} \left( \valid{\typeOf}{\interactionMatrix}[\mathbold{x}] - \valid{\typeOf}{\interactionMatrix}[\allocation[\mathbold{x}]] \right) \, \text{d} \lebesgue{\dimension \times \numParticles}.
	\end{align}
	We proceed by bounding
	\[
	\int_{\volume^{\numParticles}} \left( \valid{\typeOf}{\interactionMatrix}[\mathbold{x}] - \valid{\typeOf}{\interactionMatrix}[\allocation[\mathbold{x}]] \right) \, \text{d} \lebesgue{\dimension \times \numParticles}
	\]
	from above.
	Set
	\[
	L^{(\interactionMatrix)}_{\typeOf, \allocation} = \left\{ \mathbold{x} \in \volume^{\numParticles} \mid \valid{\typeOf}{\interactionMatrix}[\mathbold{x}] = 1 \text{ and } \valid{\typeOf}{\interactionMatrix}[\allocation[\mathbold{x}]] = 0\right\}
	\]
	and observe that
	\[
	\int_{\volume^{\numParticles}} \left( \valid{\typeOf}{\interactionMatrix}[\mathbold{x}] - \valid{\typeOf}{\interactionMatrix}[\allocation[\mathbold{x}]] \right) \, \text{d} \lebesgue{\dimension \times \numParticles}
	\le \vol{L^{(\interactionMatrix)}_{\typeOf, \allocation}} .
	\]
	Now, note that for all $\mathbold{x} = (x_1 \dots, x_{\numParticles}) \in \volume^{\numParticles}$ it holds that $\valid{\typeOf}{\interactionMatrix}[\allocation[\mathbold{x}]] = 0$ if and only if there are $i, j \in [\numParticles]$ with $i \neq j$ such that $\dist{\allocation[x_i]}{\allocation[x_j]} < \interactionMatrix[\typeOf[i]][\typeOf[j]]$.
	Let
	\[
	L^{(\interactionMatrix)}_{\typeOf, \allocation}(i, j) = \left\{ \mathbold{x} \in \volume^{\numParticles} \mid \valid{\typeOf}{\interactionMatrix}[\mathbold{x}] = 1 \text{ and } \dist{\allocation[x_i]}{\allocation[x_j]} < \interactionMatrix[\typeOf[i]][\typeOf[j]]\right\}
	\]
	and observe that
	\[
	\vol{L^{(\interactionMatrix)}_{\typeOf, \allocation}} \le \sum_{\substack{i, j \in [\numParticles]:\\ i<j}} \vol{L^{(\interactionMatrix)}_{\typeOf, \allocation}(i, j)} .
	\]
	Thus, we obtain
	\[
	\int_{\volume^{\numParticles}} \left( \valid{\typeOf}{\interactionMatrix}[\mathbold{x}] - \valid{\typeOf}{\interactionMatrix}[\allocation[\mathbold{x}]] \right) \, \text{d} \lebesgue{\dimension \times \numParticles}
	\le \sum_{\substack{i, j \in [\numParticles]:\\ i<j}} \vol{L^{(\interactionMatrix)}_{\typeOf, \allocation}(i, j)} .
	\]
	Next, we upper bound $\vol{L^{(\interactionMatrix)}_{\typeOf, \allocation}(i, j)}$ for each $i, j \in [\numParticles]$ with $i < j$.
	To this end, we first introduce some additional notation.
	For a finite set $S \subset \N$ and $i \in S$ let $\order{S}[i]$ denote the order of $i$ in $S$.
	That is, $\order{S}[i]\colon S \to [\size{S}]$ with $\order{S}(i) = \size{\{j \in S \mid j \le i\}}$.
	Let $\numParticles \ge 2$, $\mathbold{x} = (x_1, \dots, x_{\numParticles}) \in \volume^{\numParticles}$ and $\typeOf\colon [\numParticles] \to [\numTypes]$.
	For any $S \subseteq [\numParticles]$ let $\mathbold{x}_{-S} \in \volume^{\numParticles - \size{S}}$ denote $\mathbold{x}$ after removing all components index by $S$.
	Further, define $\typeOf_{-S}: [\numParticles - \size{S}] \to [\numTypes]$ where for each $i \in [\numParticles] \setminus S$ we set $\typeOf_{-S}(\order{[\numParticles] \setminus S}[i]) = \typeOf[i]$.
	Informally speaking, $\mathbold{x}_{-S}$ and $\typeOf_{-S}$ represent the configuration that results from $\mathbold{x}$ and $\typeOf$ after removing the particles index by the set $S$.
	
	With this additional notation, we proceed with upper bounding $\vol{L^{(\interactionMatrix)}_{\typeOf, \allocation}(i, j)}$.
	Recall that $\allocation$ is a $\volumeError$-$\distanceError$-allocation, which implies $\dist{x}{\allocation[x]} \le \distanceError$ for all $x \in \volume$.
	We obtain
	\begin{align*}
		\vol{L^{(\interactionMatrix)}_{\typeOf, \allocation}(i, j)}
		&= \int_{\volume^{\numParticles}}  \valid{\typeOf}{\interactionMatrix}[\mathbold{x}] \ind{\dist{\allocation[x_i]}{\allocation[x_j]} < \interactionMatrix[\typeOf[i]][\typeOf[j]]} \, \text{d} \lebesgue{\dimension \times \numParticles} \\
		&\le \int_{\volume^{\numParticles}}  \valid{\typeOf_{-\{i, j\}}}{\interactionMatrix}[\mathbold{x}_{-\{i, j\}}] \ind{\dist{x_i}{x_j} \ge \interactionMatrix[\typeOf[i]][\typeOf[j]] \text{ and } \dist{\allocation[x_i]}{\allocation[x_j]} < \interactionMatrix[\typeOf[i]][\typeOf[j]]} \, \text{d} \lebesgue{\dimension \times \numParticles} \\
		&\le \int_{\volume^{\numParticles}}  \valid{\typeOf_{-\{i, j\}}}{\interactionMatrix}[\mathbold{x}_{-\{i, j\}}] \ind{\interactionMatrix[\typeOf[i]][\typeOf[j]] \le \dist{x_i}{x_j} < \interactionMatrix[\typeOf[i]][\typeOf[j]] + 2\distanceError} \, \text{d} \lebesgue{\dimension \times \numParticles} \\
		&\le \left(\vol{\ball{\interactionMatrix[\typeOf[i]][\typeOf[j]] + 2 \distanceError}} - \vol{\ball{\interactionMatrix[\typeOf[i]][\typeOf[j]]}}\right) \vol{\volume} \int_{\volume^{\numParticles-2}}  \valid{\typeOf_{-\{i, j\}}}{\interactionMatrix}[\mathbold{x}] \text{d} \lebesgue{\dimension \times (\numParticles-2)} .
	\end{align*}
	Further, observe that for $\distanceError \le \frac{1}{2}$ it holds that
	\begin{align*}
		\left(\vol{\ball{\interactionMatrix[\typeOf[i]][\typeOf[j]] + 2 \distanceError}} - \vol{\ball{\interactionMatrix[\typeOf[i]][\typeOf[j]]}}\right)
		&\le \vol{\ball{1}} \left((\interactionMatrix_{\max} + 2 \distanceError)^{\dimension} - \interactionMatrix_{\max}^{\dimension}\right) \\
		&= \vol{\ball{1}} \sum_{i = 1}^{\dimension} \binom{\dimension}{i} \left(2 \distanceError\right)^{i} \interactionMatrix_{\max}^{\dimension - i} \\
		&= 2 \distanceError \vol{\ball{1}} \sum_{i = 1}^{\dimension} \binom{\dimension}{i} \left(2 \distanceError\right)^{i-1} \interactionMatrix_{\max}^{\dimension - i} \\
		&\le 2 \distanceError \vol{\ball{1}} (\interactionMatrix_{\max} + 1)^{\dimension} \\
		&= 2 \distanceError \vol{\ball{\interactionMatrix_{\max} + 1}}.
	\end{align*}
	Thus, we have
	\begin{align*}
		\int_{\volume^{\numParticles}} \left( \valid{\typeOf}{\interactionMatrix}[\mathbold{x}] - \valid{\typeOf}{\interactionMatrix}[\allocation[\mathbold{x}]] \right) \, \text{d} \lebesgue{\dimension \times \numParticles}
		&\le 2 \distanceError \vol{\ball{\interactionMatrix_{\max} + 1}} \vol{\volume}  \sum_{\substack{i, j \in [\numParticles]:\\ i<j}} \int_{\volume^{\numParticles-2}}  \valid{\typeOf_{-\{i, j\}}}{\interactionMatrix}[\mathbold{x}] \text{d} \lebesgue{\dimension \times (\numParticles-2)} .
	\end{align*}
	Substituting this back into \cref{lemma:lower_bound:eq_1} and observing that
	\begin{align*}
		2 \distanceError \vol{\ball{\interactionMatrix_{\max} + 1}} \vol{\volume} \frac{1}{\numParticles!}
		&\left(\prod_{l \in [\numParticles]} \fugacity[\typeOf[l]] \right) \sum_{\substack{i, j \in [\numParticles]:\\ i<j}} \int_{\volume^{\numParticles-2}}  \valid{\typeOf_{-\{i, j\}}}{\interactionMatrix}[\mathbold{x}] \text{d} \lebesgue{\dimension \times (\numParticles-2)} \\
		&= 2 \distanceError \vol{\ball{\interactionMatrix_{\max} + 1}} \vol{\volume} \sum_{\substack{i, j \in [\numParticles]:\\ i<j}} \int_{\volume^{\numParticles-2}} \frac{1}{\numParticles!}
		\left(\prod_{l \in [\numParticles]} \fugacity[\typeOf[l]] \right) \valid{\typeOf_{-\{i, j\}}}{\interactionMatrix}[\mathbold{x}] \text{d} \lebesgue{\dimension \times (\numParticles-2)} \\
		&\le 2 \distanceError \frac{\fugacity_{\max}^2}{\numParticles(\numParticles-1)} \vol{\ball{\interactionMatrix_{\max} + 1}} \vol{\volume} \sum_{\substack{i, j \in [\numParticles]:\\ i<j}} \int_{\volume^{\numParticles-2}} \frac{1}{(\numParticles-2)!}
		\left(\prod_{\substack{l \in [\numParticles]: \\ l \notin \{i, j\}}} \fugacity[\typeOf[l]] \right) \valid{\typeOf_{-\{i, j\}}}{\interactionMatrix}[\mathbold{x}] \text{d} \lebesgue{\dimension \times (\numParticles-2)} \\
		&= 2 \distanceError \frac{\fugacity_{\max}^2}{\numParticles(\numParticles-1)} \vol{\ball{\interactionMatrix_{\max} + 1}} \vol{\volume} \sum_{\substack{i, j \in [\numParticles]:\\ i<j}} \int_{\volume^{\numParticles-2}} \generalWeight{\interactionMatrix}{\fugacity}[\mathbold{x}, \typeOf_{-\{i, j\}}] \text{d} \lebesgue{\dimension \times (\numParticles-2)}
	\end{align*}
	gives
	\begin{align*}
		\sum_{\numParticles \in \N} \sum_{\typeOf\colon [\numParticles] \to [\numTypes]} \int_{\volume^{\numParticles}} &\left(\generalWeight{\interactionMatrix}{\fugacity}[\mathbold{x}, \typeOf] - \generalWeight{\interactionMatrix}{\fugacity}[\allocation[\mathbold{x}], \typeOf]\right) \, \text{d} \lebesgue{\dimension \times \numParticles} \\
		&\le 2 \distanceError \frac{\fugacity_{\max}^2}{\numParticles(\numParticles-1)} \vol{\ball{\interactionMatrix_{\max} + 1}} \vol{\volume} \sum_{\numParticles \ge 2} \sum_{\substack{i, j \in [\numParticles]:\\ i<j}} \sum_{\typeOf\colon [\numParticles] \to [\numTypes]} \int_{\volume^{\numParticles-2}} \generalWeight{\interactionMatrix}{\fugacity}[\mathbold{x}, \typeOf_{-\{i, j\}}] \text{d} \lebesgue{\dimension \times (\numParticles-2)}  \\
		&= 2 \distanceError \numTypes^2 \frac{\fugacity_{\max}^2}{\numParticles(\numParticles-1)} \vol{\ball{\interactionMatrix_{\max} + 1}} \vol{\volume} \sum_{\numParticles \ge 2} \sum_{\substack{i, j \in [\numParticles]:\\ i<j}} \sum_{\typeOf\colon [\numParticles - 2] \to [\numTypes]} \int_{\volume^{\numParticles-2}} \generalWeight{\interactionMatrix}{\fugacity}[\mathbold{x}, \typeOf] \text{d} \lebesgue{\dimension \times (\numParticles-2)}  \\
		&= \distanceError \numTypes^2 \fugacity_{\max}^2 \vol{\ball{\interactionMatrix_{\max} + 1}} \vol{\volume} \sum_{\numParticles \ge 2} \sum_{\typeOf\colon [\numParticles - 2] \to [\numTypes]} \int_{\volume^{\numParticles-2}} \generalWeight{\interactionMatrix}{\fugacity}[\mathbold{x}, \typeOf] \text{d} \lebesgue{\dimension \times (\numParticles-2)} \\
		&= \distanceError \numTypes^2 \fugacity_{\max}^2 \vol{\ball{\interactionMatrix_{\max} + 1}} \vol{\volume} \sum_{\numParticles \in \N} \sum_{\typeOf\colon [\numParticles] \to [\numTypes]} \int_{\volume^{\numParticles}} \generalWeight{\interactionMatrix}{\fugacity}[\mathbold{x}, \typeOf] \text{d} \lebesgue{\dimension \times \numParticles} \\
		&= \distanceError \numTypes \fugacity_{\max}^2 \vol{\volume} \vol{\ball{\interactionMatrix_{\max} + 1}} \generalPartitionFunction{\volume}{\interactionMatrix}{\fugacity} ,
	\end{align*}
	which concludes the proof.
\end{proof}

Combining \Cref{lemma:lower_bound} and \Cref{lemma:discretization_error:2} yields that $\hcPartitionFunction{\hcPPGraph{X}}{\hcPPWeight{X}}$ also is not much smaller than $\generalPartitionFunction{\volume}{\interactionMatrix}{\fugacity}$.
\begin{lemma}
	\label{lemma:deterministic_lower_bound}
	Let $(\volume, \interactionMatrix, \fugacity)$ be a hard-constraint point process with $\numTypes \in \N_{\ge 1}$ particle types.
	For all $\discretizationError \in \left(0, 1\right]$ there are $\distanceError, \volumeError \in \bigTheta{\discretizationError \vol{\volume}^{-1}}$ such that, for all finite $X \subset \volume$ that have a $\volumeError$-$\distanceError$-allocation and satisfy $\size{X} \ge 64 \numTypes \fugacity_{\max}^2 \vol{\volume}^2 \discretizationError^{-1}$, it holds that
	\begin{align*}
		\hcPartitionFunction{\hcPPGraph{X}}{\hcPPWeight{X}} \ge (1 - \discretizationError)\generalPartitionFunction{\volume}{\interactionMatrix}{\fugacity} ,
	\end{align*}
	where $(\hcPPGraph{X}, \hcPPWeight{X})$ is the hard-core representation of $(\volume, \interactionMatrix, \fugacity)$ based on $X$.
\end{lemma}

\begin{proof}
	Set $\interactionMatrix_{\max} = \max_{i,j \in [\numTypes]} \interactionMatrix[i][j]$ and $\fugacity_{\max} = \max_{i \in [\numTypes]} \fugacity[i]$.
	Further, set
	\[
	\distanceError
	= \frac{\discretizationError}{2 \numTypes \fugacity_{\max}^2 \vol{\ball{\interactionMatrix_{\max} + 1}} \vol{\volume}}
	\in \bigTheta{\discretizationError \vol{\volume}^{-1}}
	\]
	and
	\[
	\volumeError
	= \frac{\discretizationError}{16 \numTypes \fugacity_{\max} \vol{\volume}}
	\in \bigTheta{\discretizationError \vol{\volume}^{-1}} .
	\]
	For $\vol{\volume}$ sufficiently large we have $\distanceError \le \frac{1}{2}$ and thus, by \Cref{lemma:lower_bound}, we obtain
	\begin{align*}
		\generalPartitionFunction{\volume}{\interactionMatrix}{\fugacity} - \allocationPartitionFunction{\volume}{\interactionMatrix}{\fugacity}{\allocation}
		&= \sum_{\numParticles \in \N} \sum_{\typeOf\colon [\numParticles] \to [\numTypes]} \int_{\volume^{\numParticles}} \left(\generalWeight{\interactionMatrix}{\fugacity}[\mathbold{x}, \typeOf] - \generalWeight{\interactionMatrix}{\fugacity}[\allocation[\mathbold{x}], \typeOf]\right) \, \text{d} \lebesgue{\dimension \times \numParticles} \\
		&\le \frac{\discretizationError}{2} \generalPartitionFunction{\volume}{\interactionMatrix}{\fugacity},
	\end{align*}
	which implies
	\begin{align}
		\label[ineq]{lemma:deterministic_lower_bound:eq1}
		\allocationPartitionFunction{\volume}{\interactionMatrix}{\fugacity}{\allocation}\ge \left(1 - \frac{\discretizationError}{2}\right) \generalPartitionFunction{\volume}{\interactionMatrix}{\fugacity} .
	\end{align}
	
	Further, note that $\size{X} \ge 64 \numTypes \fugacity_{\max}^2 \vol{\volume}^2 \discretizationError^{-1}  \ge 4 \fugacity_{\max} \vol{\volume}$.
	Thus, by \Cref{lemma:hcPPPartitionFunction_equivalence,lemma:discretization_error:2}, we have
	\begin{align*}
		\absolute{\allocationPartitionFunction{\volume}{\interactionMatrix}{\fugacity}{\allocation} - \hcPartitionFunction{\hcPPGraph{X}}{\hcPPWeight{X}}}
		&\le \sum_{\numParticles \in \N}
		\sum_{\typeOf\colon [\numParticles] \to [\numTypes]} \int_{\volume^{\numParticles}}
		\absolute{\generalWeight{\interactionMatrix}{\fugacity}[\allocation[\mathbold{x}], \typeOf] - \hcPPWeightContinuous{\hcPPGraph{X}}{\hcPPWeight{X}}{\allocation}[\mathbold{x}, \typeOf]} \, \text{d} \lebesgue{\dimension \times \numParticles} \\
		%		&\le \left(\eulerE^{\frac{8}{\size{X}} \sum_{i \in [\numTypes]} \fugacity[i]^2 \vol{\volume}^2} \eulerE^{2 \volumeError \sum_{i \in [\numTypes]} \fugacity[i] \vol{\volume}} - 1\right)	\allocationPartitionFunction{\volume}{\interactionMatrix}{\fugacity}{\allocation} \\
		&\le \left(\eulerE^{\frac{\discretizationError}{4}} - 1\right) \allocationPartitionFunction{\volume}{\interactionMatrix}{\fugacity}{\allocation}.
	\end{align*}
	As $\discretizationError \le 2$ it holds that
	\begin{align*}
		\hcPartitionFunction{\hcPPGraph{X}}{\hcPPWeight{X}}
		&\ge \left(2 - \eulerE^{\frac{\discretizationError}{4}}\right) \allocationPartitionFunction{\volume}{\interactionMatrix}{\fugacity}{\allocation}\\
		&\ge \left(1 - \frac{\discretizationError}{2}\right) \allocationPartitionFunction{\volume}{\interactionMatrix}{\fugacity}{\allocation} .
	\end{align*}
	Combining this with \cref{lemma:deterministic_lower_bound:eq1} finally gives us
	\begin{align*}
		\hcPartitionFunction{\hcPPGraph{X}}{\hcPPWeight{X}}
		&\ge \left(1 - \frac{\discretizationError}{2}\right) \allocationPartitionFunction{\volume}{\interactionMatrix}{\fugacity}{\allocation} \\
		&\ge \left(1 - \frac{\discretizationError}{2}\right)^2 \generalPartitionFunction{\volume}{\interactionMatrix}{\fugacity} \\
		&\ge (1 - \discretizationError) \generalPartitionFunction{\volume}{\interactionMatrix}{\fugacity} ,
	\end{align*}
	which proves the claim.
\end{proof}

By randomizing the considered set of points, we turn \Cref{lemma:deterministic_lower_bound} into a bound on how strong $\hcPartitionFunction{\hcPPGraph{X}}{\hcPPWeight{X}}$ concentrates from below.
Combining this with \Cref{lemma:modified_markov} then gives two-sided concentration.
However, note that \Cref{lemma:modified_markov} only works for concentration around the expected value, but \Cref{lemma:deterministic_lower_bound} is given in terms of $\generalPartitionFunction{\volume}{\interactionMatrix}{\fugacity}$.
The following lemma will solves this problem by relating both quantities.

\begin{lemma}
	\label{lemma:expectation}
	Let $(\volume, \interactionMatrix, \fugacity)$ be a hard-constraint point process with $\numTypes \in \N_{\ge 1}$ particle types.
	For all $\numPoints \ge 1$ it holds that
	\[
	\E{\hcPartitionFunction{\hcPPGraph{X}}{\hcPPWeight{X}}} \le \generalPartitionFunction{\volume}{\interactionMatrix}{\fugacity} ,
	\]
	where the expectation is taken over $X \subset \volume$ with $\size{X} = \numPoints$ drawn uniformly at random.
\end{lemma}

\begin{proof}
	We start by rewriting $\hcPartitionFunction{\hcPPGraph{X}}{\hcPPWeight{X}}$ as
	\begin{align*}
		\hcPartitionFunction{\hcPPGraph{X}}{\hcPPWeight{X}}
		= 1 + \sum_{\numParticles = 1}^{\numPoints} \frac{\vol{\volume}^{\numParticles}}{\numPoints^{\numParticles}} \sum_{\substack{(\numParticles_1, \dots, \numParticles_{\numTypes}) \in \N^{\numTypes}: \\ \numParticles_1 + \dots + \numParticles_{\numTypes} = \numParticles}} \left(\prod_{i \in [\numTypes]} \fugacity[i]^{\numParticles_i}\right) \sum_{\substack{S_1 \subseteq X:\\ \size{S_1} = \numParticles_1}} \dots \sum_{\substack{S_{\numTypes} \subseteq X:\\ \size{S_\numTypes} = \numParticles_\numTypes}} \ind{\substack{\forall i, j \in [\numTypes] \forall x_1 \in S_i, x_2 \in S_j:\\ (i=j \text{ and } x_1 = x_2)\\ \text{ or } \dist{x_1}{x_2} \ge \interactionMatrix[i][j]}} ~,
	\end{align*}
	where the condition $(i=j \text{ and } x_1 = x_2)$ accounts for the fact that $\hcPPGraph{X}$ does not contain self loops.
	We proceed by rewriting this in terms of tuples of points from $X$.
	To this end, for every $(\numParticles_1, \dots, \numParticles_{\numTypes}) \in \N^{\numTypes}$ define a canonical type assignment $\typeOf_{\numParticles_1, \dots, \numParticles_{\numTypes}}: [\numParticles_1 + \dots + \numParticles_{\numTypes}] \to [\numTypes]$ with $\typeOf_{\numParticles_1, \dots, \numParticles_{\numTypes}}(i) = j$ such that $\sum_{l=1}^{j-1} \numParticles_l < i \le \sum_{l=1}^{j} \numParticles_l$.
	For any $\numParticles \in \N$ let $X^{\numParticles}_{\neq} = \{(x_1, \dots, x_{\numParticles}) \mid \forall i, j \in [\numParticles]\colon i \neq j \text{ implies } x_i \neq x_j\}$.
	Further, for $\mathbold{x} = (x_1, \dots, x_{\numParticles_1}) \in X^{\numParticles_1}$ and $\mathbold{y} = (y_1, \dots, y_{\numParticles_2}) \in X^{\numParticles_2}$ write $\mathbold{x} \concatenation \mathbold{y}$ for the concatenation $(x_1, \dots, x_{\numParticles_1}, y_1, \dots, y_{\numParticles_2}) \in X^{\numParticles_1 + \numParticles_2}$ of both tuples.
	Observe that
	\begin{align*}
		\sum_{\substack{S_1 \subseteq X:\\ \size{S_1} = \numParticles_1}} \dots \sum_{\substack{S_{\numTypes} \subseteq X:\\ \size{S_\numTypes} = \numParticles_\numTypes}} \ind{\substack{\forall i, j \in [\numTypes] \forall x_1 \in S_i, x_2 \in S_j:\\ (i=j \text{ and } x_1 = x_2)\\ \text{ or } \dist{x_1}{x_2} \ge \interactionMatrix[i][j]}}
		&= \frac{1}{\numParticles_1 !} \sum_{\mathbold{x_1} \in X^{\numParticles_1}_{\neq}} \dots \frac{1}{\numParticles_{\numTypes} !} \sum_{\mathbold{x_{\numTypes}} \in X^{\numParticles_{\numTypes}}_{\neq}} \valid{\typeOf_{\numParticles_1, \dots, \numParticles_{\numTypes}}}{\interactionMatrix}[\mathbold{x_1} \concatenation \dots \concatenation \mathbold{x_{\numTypes}}] \\
		&\le \left(\prod_{i \in [\numTypes]} \frac{1}{\numParticles_{i}!}\right) \sum_{\mathbold{x} \in X^{\numParticles}} \valid{\typeOf_{\numParticles_1, \dots, \numParticles_{\numTypes}}}{\interactionMatrix}[\mathbold{x}] .
	\end{align*}
	Thus, we have
	\[
	\hcPartitionFunction{\hcPPGraph{X}}{\hcPPWeight{X}}
	= 1 + \sum_{\numParticles = 1}^{\numPoints} \frac{\vol{\volume}^{\numParticles}}{\numPoints^{\numParticles}} \sum_{\substack{(\numParticles_1, \dots, \numParticles_{\numTypes}) \in \N^{\numTypes}: \\ \numParticles_1 + \dots + \numParticles_{\numTypes} = \numParticles}} \left(\prod_{i \in [\numTypes]} \fugacity[i]^{\numParticles_i}\right) \cdot \left(\prod_{i \in [\numTypes]} \frac{1}{\numParticles_{i}!}\right) \sum_{\mathbold{x} \in X^{\numParticles}} \valid{\typeOf_{\numParticles_1, \dots, \numParticles_{\numTypes}}}{\interactionMatrix}[\mathbold{x}] ,
	\]
	which, by monotonicity and linearity of expectation yields
	\begin{align}
		\notag
		\E{\hcPartitionFunction{\hcPPGraph{X}}{\hcPPWeight{X}}}
		&\le 1 + \sum_{\numParticles = 1}^{\numPoints} \frac{\vol{\volume}^{\numParticles}}{\numPoints^{\numParticles}} \sum_{\substack{(\numParticles_1, \dots, \numParticles_{\numTypes}) \in \N^{\numTypes}: \\ \numParticles_1 + \dots + \numParticles_{\numTypes} = \numParticles}} \left(\prod_{i \in [\numTypes]} \fugacity[i]^{\numParticles_i}\right) \cdot \left(\prod_{i \in [\numTypes]} \frac{1}{\numParticles_{i}!}\right) \E{\sum_{\mathbold{x} \in X^{\numParticles}} \valid{\typeOf_{\numParticles_1, \dots, \numParticles_{\numTypes}}}{\interactionMatrix}[\mathbold{x}]} \\
		\label[ineq]{lemma:expectation:eq1}
		&\le 1 + \sum_{\numParticles \in \N_{\ge 1}} \frac{\vol{\volume}^{\numParticles}}{\numPoints^{\numParticles}} \sum_{\substack{(\numParticles_1, \dots, \numParticles_{\numTypes}) \in \N^{\numTypes}: \\ \numParticles_1 + \dots + \numParticles_{\numTypes} = \numParticles}} \left(\prod_{i \in [\numTypes]} \fugacity[i]^{\numParticles_i}\right) \cdot \left(\prod_{i \in [\numTypes]} \frac{1}{\numParticles_{i}!}\right) \E{\sum_{\mathbold{x} \in X^{\numParticles}} \valid{\typeOf_{\numParticles_1, \dots, \numParticles_{\numTypes}}}{\interactionMatrix}[\mathbold{x}]}
	\end{align}
	Let $x_1, \dots, x_{\numPoints}$ some arbitrary ordering of the points in $X$.
	Now, note that
	\[
	\sum_{\mathbold{x} \in X^{\numParticles}} \valid{\typeOf_{\numParticles_1, \dots, \numParticles_{\numTypes}}}{\interactionMatrix}[\mathbold{x}]
	= \sum_{(i_1, \dots, i_{\numParticles}) \in [\numPoints]^{\numParticles}} \valid{\typeOf_{\numParticles_1, \dots, \numParticles_{\numTypes}}}{\interactionMatrix}[x_{i_1}, \dots, x_{i_{\numParticles}}] .
	\]
	Thus, we get
	\begin{align*}
		\E{\sum_{\mathbold{x} \in X^{\numParticles}} \valid{\typeOf_{\numParticles_1, \dots, \numParticles_{\numTypes}}}{\interactionMatrix}[\mathbold{x}]}
		= \sum_{(i_1, \dots, i_{\numParticles}) \in [\numPoints]^{\numParticles}} \E{\valid{\typeOf_{\numParticles_1, \dots, \numParticles_{\numTypes}}}{\interactionMatrix}[x_{i_1}, \dots, x_{i_{\numParticles}}]}
		= \frac{\numPoints^{\numParticles}}{\vol{\volume}^{\numParticles}} \int_{\volume^{\numParticles}} \valid{\typeOf_{\numParticles_1, \dots, \numParticles_{\numTypes}}}{\interactionMatrix}[\mathbold{x}] \, \text{d} \lebesgue{\dimension \times \numParticles} .
	\end{align*}
	Substituting this back into \cref{lemma:expectation:eq1} and canceling some terms yields
	\[
	\E{\hcPartitionFunction{\hcPPGraph{X}}{\hcPPWeight{X}}} \le 1 + \sum_{\numParticles \in \N_{\ge 1}}  \sum_{\substack{(\numParticles_1, \dots, \numParticles_{\numTypes}) \in \N^{\numTypes}: \\ \numParticles_1 + \dots + \numParticles_{\numTypes} = \numParticles}} \left(\prod_{i \in [\numTypes]} \fugacity[i]^{\numParticles_i}\right) \cdot \left(\prod_{i \in [\numTypes]} \frac{1}{\numParticles_{i}!}\right)  \int_{\volume^{\numParticles}} \valid{\typeOf_{\numParticles_1, \dots, \numParticles_{\numTypes}}}{\interactionMatrix}[\mathbold{x}] \, \text{d} \lebesgue{\dimension \times \numParticles} .
	\]
	Finally, observe that, for each $\numParticles \in \N_{\ge 1}$ and all $(\numParticles_1, \dots, \numParticles_{\numTypes}) \in \N^{\numTypes}$ with $\numParticles_1 + \dots + \numParticles_{\numTypes} = \numParticles$, there a $\frac{\numParticles!}{\prod_{i \in [\numTypes]} \numParticles_{i} !}$ type assignments $\typeOf\colon [\numParticles] \to [\numTypes]$ that are equivalent to $\typeOf_{\numParticles_1, \dots, \numParticles_{\numTypes}}$ up to permutation of inputs (i.e., $\size{\typeOf^{-1} (i)} = \size{\typeOf^{-1}_{\numParticles_1, \dots, \numParticles_{\numTypes}}(i)}$ for each $i \in [\numTypes]$).
	Further, note that $\prod_{i \in [\numParticles]} \fugacity[\typeOf[i]] = \prod_{i \in [\numTypes]} \fugacity[i]^{\numParticles_i}$ and $\int_{\volume^{\numParticles}} \valid{\typeOf}{\interactionMatrix}[\mathbold{x}] \, \text{d} \lebesgue{\dimension \times \numParticles} = \int_{\volume^{\numParticles}} \valid{\typeOf_{\numParticles_1, \dots, \numParticles_{\numTypes}}}{\interactionMatrix}[\mathbold{x}] \, \text{d} \lebesgue{\dimension \times \numParticles}$ for each such $\typeOf$.
	Thus, it holds that
	\begin{align*}
		\E{\hcPartitionFunction{\hcPPGraph{X}}{\hcPPWeight{X}}} &\le 1 + \sum_{\numParticles \in \N_{\ge 1}} \frac{1}{\numParticles!}  \sum_{\typeOf: [\numParticles] \to [\numTypes]} \left(\prod_{i \in [\numParticles]} \fugacity[\typeOf[i]]\right)   \int_{\volume^{\numParticles}} \valid{\typeOf}{\interactionMatrix}[\mathbold{x}] \, \text{d} \lebesgue{\dimension \times \numParticles} \\
		&= \generalPartitionFunction{\volume}{\interactionMatrix}{\fugacity},
	\end{align*}
	proving the claim.
\end{proof}

We are now ready to state and prove our main theorem about concentration of random discretizations without star-convexity.
\concentrationnonstarconvex

\begin{proof}
	The main idea of the proof is to use \Cref{lemma:partitioning_allocation} and \Cref{lemma:deterministic_lower_bound} to argue that, for a uniformly drawn $X \subset \volume$ that is sufficiently large, the hard-core partition function $\hcPartitionFunction{\hcPPGraph{X}}{\hcPPWeight{X}}$ is with high probability not much smaller than the partition function $\generalPartitionFunction{\volume}{\interactionMatrix}{\fugacity}$ of the continuous model.
	Then we apply \Cref{lemma:modified_markov} to derive a probabilistic upper bound on $\hcPartitionFunction{\hcPPGraph{X}}{\hcPPWeight{X}}$
	
	Let $c = \frac{\allocationProb}{2}$, $\discretizationError' = \frac{\discretizationError}{c} = \frac{\allocationProb}{2} \discretizationError$ and $\allocationProb' = \frac{\discretizationError}{4} \allocationProb = \frac{\discretizationError'}{2} $.
	By \Cref{lemma:deterministic_lower_bound}, we know that there are $\distanceError', \volumeError' \in \bigTheta{\discretizationError' \vol{\volume}^{-1}} = \bigTheta{\discretizationError \allocationProb \vol{\volume}^{-1}}$ such that, for all $X \subset \volume$ with $\size{X} \ge 64 \numTypes \fugacity_{\max}^2 \vol{\volume}^2 \discretizationError'^{-1} = 128 \numTypes \fugacity_{\max}^2 \vol{\volume}^2 \discretizationError^{-1} \allocationProb^{-1}$ that have an $\volumeError$-$\distanceError$-allocation, it holds that
	\[
	\hcPartitionFunction{\hcPPGraph{X}}{\hcPPWeight{X}} \ge \left(1 - \discretizationError'\right) \generalPartitionFunction{\volume}{\interactionMatrix}{\fugacity} .
	\]
	
	Next, note that we assume that there is a $\volumeScaling$-$\distanceError$-partitioning of $\volume$ of size $\poly{\frac{\vol{\volume}}{\distanceError}}$ for some $\volumeScaling \in (0, 1]$ and all $\distanceError \in \R_{>0}$.
	Thus, by \Cref{lemma:partitioning_allocation} we conclude that there is some \[
	\numPoints'_{\discretizationError, \allocationProb}
	\in \poly{\frac{\vol{\volume}}{\distanceError' \volumeError'} \ln\left(\frac{1}{\allocationProb'}\right)}
	= \poly{\frac{\vol{\volume}}{\discretizationError \allocationProb}}
	\]
	such that for all $\numParticles \ge \volume$ it hold that a uniformly random $X \subset \volume$ with $\size{X} = \numPoints$ has a $\volumeError'$-$\distanceError'$-allocation with probability at least $1 - \allocationProb'$.
	
	Setting
	\[
	\numPoints_{\discretizationError, \allocationProb} = \max\{128 \numTypes \fugacity_{\max}^2 \vol{\volume}^2 \discretizationError^{-1} \allocationProb^{-1}, \numPoints'_{\discretizationError, \allocationProb}\} \in \poly{\frac{\vol{\volume}}{\discretizationError \allocationProb}}
	\]
	now ensures that for all $\numPoints \ge \numPoints_{\discretizationError, \allocationProb}$
	\[
	\hcPartitionFunction{\hcPPGraph{X}}{\hcPPWeight{X}} \ge \left(1 - \discretizationError'\right) \generalPartitionFunction{\volume}{\interactionMatrix}{\fugacity}
	\]
	with probability at least $1 - \allocationProb'$.
	Further, \Cref{lemma:expectation} implies that
	\[
	\Pr{\hcPartitionFunction{\hcPPGraph{X}}{\hcPPWeight{X}} < \left(1 - \discretizationError'\right) \E{\hcPartitionFunction{\hcPPGraph{X}}{\hcPPWeight{X}}}} \le \allocationProb'.
	\]
	
	By \Cref{lemma:modified_markov,lemma:expectation}, we obtain
	\begin{align*}
		\Pr{\hcPartitionFunction{\hcPPGraph{X}}{\hcPPWeight{X}} \ge (1 + \discretizationError) \generalPartitionFunction{\volume}{\interactionMatrix}{\fugacity}}
		&\le
		\Pr{\hcPartitionFunction{\hcPPGraph{X}}{\hcPPWeight{X}} \ge (1 + \discretizationError) \E{\hcPartitionFunction{\hcPPGraph{X}}{\hcPPWeight{X}}}} \\
		&= \Pr{\hcPartitionFunction{\hcPPGraph{X}}{\hcPPWeight{X}} \ge (1 + c \discretizationError') \E{\hcPartitionFunction{\hcPPGraph{X}}{\hcPPWeight{X}}}} \\
		&\le \frac{1}{c + 1} \left(1 + \frac{\allocationProb' (1-\discretizationError')}{\discretizationError'}\right) \\
		&\le \frac{1}{c+1} \left(1 + \frac{1}{2}\right) \\
		&\le \frac{3}{4} \allocationProb .
	\end{align*}
	Using the union bound on the failure probabilities, we obtain
	\[
	\left(1 - \discretizationError'\right) \generalPartitionFunction{\volume}{\interactionMatrix}{\fugacity} \le
	\hcPartitionFunction{\hcPPGraph{X}}{\hcPPWeight{X}} \le (1 + \discretizationError) \generalPartitionFunction{\volume}{\interactionMatrix}{\fugacity}
	\]
	with probability at least $1 - \allocationProb' - \frac{3}{4} \allocationProb = 1 - \frac{\discretizationError}{4} \allocationProb - \frac{3}{4} \allocationProb \ge 1 - \allocationProb$.
	Finally, observing that $1 + \discretizationError \le \eulerE^{\discretizationError}$ and $1 - \discretizationError' = 1 - \frac{\allocationProb}{2} \discretizationError \ge \eulerE^{- \discretizationError}$ concludes the proof.
\end{proof}

\begin{remark}
	Note that the concentration result without star-convexity, given in \Cref{thm:prob_bound_non_starconvex}, has a worse dependency on the desired error probability $\allocationProb$ than that obtained with star-convexity in \Cref{cor:random_discretization_concentration}.
	This is due to the application of \Cref{lemma:modified_markov}.
	However, the dependency is still sufficient to construct a fully polynomial randomized approximation scheme (FPRAS) for $\generalPartitionFunction{\volume}{\interactionMatrix}{\fugacity}$.
\end{remark}

\section*{Acknowledgments}
Andreas Göbel was funded by the project PAGES (project No. 467516565) of the German Research Foundation (DFG).
This project has received funding from the European Union's Horizon 2020 research and innovation program under the Marie Skłodowska-Curie grant agreement No. 945298-ParisRegionFP.
This research was partially funded by the HPI Research School on Data Science and Engineering.

\printbibliography

@article{penrose1994self,
  title={Self-avoiding walks and trees in spread-out lattices},
  author={Penrose, Mathew D},
  journal={Journal of Statistical Physics},
  volume={77},
  number={1-2},
  pages={3--15},
  year={1994},
  doi={10.1007/BF02186829},
}

@article{JVV86,
    author={Mark Jerrum and Leslie G. Valiant and Vijay V. Vazirani},
    title={Random generation of combinatorial structures from a uniform distribution},
    journal={Theoretical Computer Science},
    volume={43},
    pages={169--188},
    year={1986},
    doi={10.1016/0304-3975(86)90174-X},
}

@article{ALOG20,
    author    = {Nima Anari and
        Kuikui Liu and
        Shayan Oveis Gharan},
    title     = {Spectral independence in high-dimensio\-nal expanders and applications to the hardcore model},
    journal   = {CoRR},
    volume    = {abs/2001.00303},
    year      = {2020},
    url       = {http://arxiv.org/abs/2001.00303},
}

@article{Chen2020rapid,
    author    = {Zongchen Chen and Kuikui Liu and Eric Vigoda},
    title     = {Rapid mixing of Glauber dynamics up to uniqueness via contraction},
    journal   = {CoRR},
    volume    = {abs/2004.09083},
    year      = {2020},
    url       = {http://arxiv.org/abs/2004.09083},
}

@inproceedings{ChenLV21,
  author    = {Zongchen Chen and
               Kuikui Liu and
               Eric Vigoda},
  title     = {Optimal mixing of Glauber dynamics: entropy factorization via high-dimensional
               expansion},
  booktitle = {{STOC}},
  pages     = {1537--1550},
  publisher = {{ACM}},
  year      = {2021},
  url       = {https://doi.org/10.1145/3406325.3451035},
}

@article{LeeYang,
  title = {Statistical Theory of Equations of State and Phase Transitions. II. Lattice Gas and Ising Model},
  author = {Lee, T. D. and Yang, C. N.},
  journal = {Phys. Rev.},
  volume = {87},
  issue = {3},
  pages = {410--419},
  numpages = {0},
  year = {1952},
  publisher = {American Physical Society},
  doi = {10.1103/PhysRev.87.410},
  url = {https://link.aps.org/doi/10.1103/PhysRev.87.410}
}

@article{chen2021alldegrees,
  title={Rapid mixing of Glauber dynamics via spectral independence for all degrees},
  author={Chen, Xiaoyu and Feng, Weiming and Yin, Yitong and Zhang, Xinyuan},
  journal={arXiv preprint arXiv:2105.15005},
  year={2021}
}

@article{friedrich2020polymer,
    author    = {Tobias Friedrich and Andreas Göbel and Martin S. Krejca and Marcus Pappik},
    title     = {Polymer dynamics via cliques: New conditions for approximations},
    journal   = {CoRR},
    volume    = {abs/2007.08293},
    year      = {2020},
    url       = {http://arxiv.org/abs/2007.08293},
}

@article{friedrich2021spectral,
      title={A spectral independence view on hardspheres via block dynamics},
      author={Tobias Friedrich and Andreas Göbel and Martin S. Krejca and Marcus Pappik},
      year={2021},
      eprint={2102.07443},
      archivePrefix={arXiv},
      primaryClass={math.PR}
}

@article{michelen2021potential,
  title={Potential-weighted connective constants and uniqueness of Gibbs measures},
  author={Michelen, Marcus and Perkins, Will},
  journal={arXiv preprint arXiv:2109.01094},
  year={2021}
}

@book{LiquidsBook,
    title = {Theory of Simple Liquids},
    author = {Jean-Pierre Hansen and Ian R. McDonald},
    booktitle = {Theory of Simple Liquids},
    publisher = {Academic Press},
    edition = {Fourth Edition},
    year = {2013},
    doi = {10.1016/B978-0-12-387032-2.00013-1},
}

@book{boublik1980statistical,
    author={Tomá\v{s} Boublik and Ivo Nezbeda and Karel Hlavaty},
    title={Statistical thermodynamics of simple liquids and their mixtures},
    isbn={9780444416995},
    series={Fundamental Studies in Engineering},
    year={1980},
    publisher={Elsevier},
}

@book{friedli_velenik_2017,
    author={Sacha Friedli and Yvan Velenik},
    title={Statistical Mechanics of Lattice Systems: A Concrete Mathematical Introduction},
    ISBN={978-1-107-18482-4},
    publisher={Cambridge University Press},
    year={2017},
    doi={10.1017/9781316882603},
}

@inproceedings{HPR19,
    author    = {Tyler Helmuth and
        Will Perkins and
        Guus Regts},
    title     = {Algorithmic Pirogov--Sinai theory},
    booktitle = {Proc.~of {STOC}'19},
    pages     = {1009--1020},
    year      = {2019},
    doi       = {10.1145/3313276.3316305},
}

@inproceedings{JKP19,
    author    = {Matthew Jenssen and
        Peter Keevash and
        Will Perkins},
    title     = {Algorithms for {\#}BIS-hard problems on expander graphs},
    booktitle = {Proc.~of {SODA}'19},
    pages     = {2235--2247},
    year      = {2019},
    doi       = {10.1137/1.9781611975482.135},
}

@inproceedings{CGGPSV19,
    author    = {Zongchen Chen and
        Andreas Galanis and
        Leslie Ann Goldberg and
        Will Perkins and
        James Stewart and
        Eric Vigoda},
    title     = {Fast algorithms at low temperatures via Markov chains},
    booktitle = {Proc.~of {APPROX/RANDOM}'19},
    pages     = {41:1--41:14},
    year      = {2019},
    doi       = {10.4230/LIPIcs.APPROX-RANDOM.2019.41},
}

@inproceedings{CP20,
    author    = {Sarah Cannon and
        Will Perkins},
    title     = {Counting independent sets in unbalanced bipartite graphs},
    booktitle = {Proc.~of {SODA}'20},
    pages     = {1456--1466},
    year      = {2020},
    doi       = {10.1137/1.9781611975994.88},
}

@unpublished{perkinsAlgorithm,
  title={Spatial and temporal mixing for repulsive point processes},
  author={Michelen, Marcus and Perkins, Will},
  year={2021},
  note={Manuscript},
}

@article{peters2019conjecture,
  title={On a conjecture of Sokal concerning roots of the independence polynomial},
  author={Peters, Han and Regts, Guus},
  journal={The Michigan Mathematical Journal},
  volume={68},
  number={1},
  pages={33--55},
  year={2019}
}

@inproceedings{2010:Sly:computational_transition,
    author    = {Allan Sly},
    title     = {Computational transition at the uniqueness threshold},
    booktitle = {Proc.~of {FOCS}'10},
    pages     = {287--296},
    year      = {2010},
    doi       = {10.1109/FOCS.2010.34},
}

@article{2014:Galanis:inapproximability_independent_hard_core,
    author    = {Andreas Galanis and
        Qi Ge and
        Daniel Stefankovic and
        Eric Vigoda and
        Linji Yang},
    title     = {Improved inapproximability results for counting independent sets in the hard-core model},
    journal   = {Random Structures \& Algorithms},
    volume    = {45},
    number    = {1},
    pages     = {78--110},
    year      = {2014},
    doi       = {10.1002/rsa.20479},
}

@inproceedings{Weitz2006Counting,
    author    = {Dror Weitz},
    title     = {Counting independent sets up to the tree threshold},
    booktitle = {Proc.~of {STOC}'06},
    pages     = {140--149},
    year      = {2006},
    doi       = {10.1145/1132516.1132538},
}

@article{sinclair2017spatial,
  title={Spatial mixing and the connective constant: Optimal bounds},
  author={Sinclair, Alistair and Srivastava, Piyush and {\v{S}}tefankovi{\v{c}}, Daniel and Yin, Yitong},
  journal={Probability Theory and Related Fields},
  volume={168},
  number={1-2},
  pages={153--197},
  year={2017},
  doi={10.1007/s00440-016-0708-2},
}

@article{patel2017deterministic,
  title={Deterministic polynomial-time approximation algorithms for partition functions and graph polynomials},
  author={Patel, Viresh and Regts, Guus},
  journal={SIAM Journal on Computing},
  volume={46},
  number={6},
  pages={1893--1919},
  year={2017},
  doi={10.1137/16M1101003},
}

@article{ruelle1971existence,
  title={Existence of a phase transition in a continuous classical system},
  author={Ruelle, David},
  journal={Physical Review Letters},
  volume={27},
  number={16},
  pages={1040},
  year={1971},
  publisher={APS}
}

@book{simonbook,
 URL = {http://www.jstor.org/stable/j.ctt7ztwsn},
 author = {Barry Simon},
 publisher = {Princeton University Press},
 title = {The Statistical Mechanics of Lattice Gases, Volume I},
 year = {1993}
}

@article{WR70,
author = {Widom,B.  and Rowlinson,J. S. },
title = {New Model for the Study of Liquid-Vapor Phase Transitions},
shorttitle = {New Model for the Study of Liquid-Vapor Phase \ldots},
journal = {J. Chem. Phys.},
volume = {52},
number = {4},
pages = {1670-1684},
year = {1970},
doi = {10.1063/1.1673203},
URL = {https://doi.org/10.1063/1.1673203}
}

@article{kelly1985,
 URL = {http://www.jstor.org/stable/2345773},
 author = {Frank P. Kelly},
 journal = {P. Roy. Soc. B},
 number = {3},
 pages = {379--395},
 publisher = {[Royal Statistical Society, Wiley]},
 title = {Stochastic Models of Computer Communication Systems},
 volume = {47},
 year = {1985}
}

@article{BRIGHTWELL1999221,
title = "Graph Homomorphisms and Phase Transitions",
journal = "J. Comb. Theory B",
volume = "77",
number = "2",
pages = "221 - 262",
year = "1999",
issn = "0095-8956",
doi = "https://doi.org/10.1006/jctb.1999.1899",
url = "http://www.sciencedirect.com/science/article/pii/S0095895699918993",
author = "Graham R. Brightwell and Peter Winkler"
}

@book{barvinok2016combinatorics,
  title={Combinatorics and complexity of partition functions},
  author={Barvinok, Alexander},
  volume={9},
  publisher={Springer},
  year={2016},
}

@article{den2012probability,
  title={Probability theory: The coupling method},
  author={Den Hollander, Frank},
  journal={Lecture notes available online},
  url={http://websites.math.leidenuniv.nl/probability/lecturenotes/CouplingLectures.pdf},
  year={2012}
}

@article{muller2008two,
  title={Two-point concentration in random geometric graphs},
  author={M{\"u}ller, Tobias},
  journal={Combinatorica},
  volume={28},
  number={5},
  pages={529},
  year={2008},
  publisher={Springer}
}

@article{bachmann2016concentration,
  title={Concentration for Poisson functionals: component counts in random geometric graphs},
  author={Bachmann, Sascha},
  journal={Stochastic Processes and their Applications},
  volume={126},
  number={5},
  pages={1306--1330},
  year={2016},
  publisher={Elsevier}
}

@article{bachmann2018concentration,
  title={Concentration for Poisson U-statistics: subgraph counts in random geometric graphs},
  author={Bachmann, Sascha and Reitzner, Matthias},
  journal={Stochastic Processes and their Applications},
  volume={128},
  number={10},
  pages={3327--3352},
  year={2018},
  publisher={Elsevier}
}

\end{document}